\documentclass[a4paper, 11pt]{article} \pdfoutput=1

\usepackage{amsmath, amsfonts, amscd, amssymb,  amsthm, mathrsfs, ascmac}
\usepackage{latexsym}
\usepackage{longtable, geometry}
\usepackage[english]{babel}
\usepackage[utf8]{inputenc}
\usepackage[dvipdfmx]{graphicx}
\usepackage{bbm}
\usepackage{enumitem}
\usepackage{nicefrac}
\usepackage{bm}
\usepackage[subrefformat=parens]{subcaption}
\usepackage[svgnames,psnames]{xcolor}
\usepackage[colorlinks,citecolor=DarkGreen,linkcolor=FireBrick,linktocpage,unicode]{hyperref}
 \usepackage{braket}
\usepackage{xparse}
\usepackage{tikz}
\usetikzlibrary{intersections,calc,arrows}
\usetikzlibrary{decorations.pathreplacing}

\hypersetup{  urlcolor=blue}
\usepackage{comment}
\usepackage{authblk}
\newtheorem{Thm}{Theorem}[section]

\newtheorem{Lemm}[Thm]{Lemma}
\newtheorem{Prop}[Thm]{Proposition}

\theoremstyle{definition}

\newtheorem{Rem}[Thm]{Remark}
\newcommand{\be}{\begin{equation}}
\newcommand{\ee}{\end{equation}}

\newcommand{\ba}{\begin{align}}
\newcommand{\ea}{\end{align}}

\newcommand{\ben}{\begin{equation*}}
\newcommand{\een}{\end{equation*}}

\def\i<#1>{\langle #1 \rangle}
\def\l<#1>{\left\langle #1 \right\rangle}
\def\b<#1>{\big\langle #1 \big\rangle}
\def\B<#1>{\Big\langle #1 \Big\rangle}

\newcommand{\Ti}{8}
\newcommand{\Sp}{8}
\newcommand{\Ct}{5/2}

\newcommand{\dG}{d\vGa}
\newcommand{\im}{\mathrm{i} }
\newcommand{\la}{\langle}
\newcommand{\ra}{\rangle}
\newcommand{\bs}{\boldsymbol}
\newcommand{\ex}{\mathrm{e}}
\newcommand{\h}{\mathfrak{h}}
\newcommand{\Tr}{\mathrm{Tr}}

\newcommand{\D}{\mathrm{dom}}
\newcommand{\BbbR}{\mathbb{R}}
\newcommand{\BbbN}{\mathbb{N}}
\newcommand{\BbbZ}{\mathbb{Z}}

\newcommand{\BbbC}{\mathbb{C}}

\newcommand{\bome}{{\boldsymbol \omega}}
\newcommand{\bphi}{\bs{\varphi}}
\newcommand{\vepsilon}{\varepsilon}
\newcommand{\vphi}{\varphi}
\newcommand{\no}{\nonumber \\}

\newcommand{\vLa}{\varLambda}
\newcommand{\vGa}{\varGamma}
\newcommand{\vSi}{\varSigma}

\newcommand{\supp}{\mathrm{supp}\, }
\newcommand{\inte}{\mathrm{int}\, }

\newcommand{\bn}{\bs{n}}

\newcommand{\vXi}{\varXi}
\newcommand{\N}{\mathcal{N}}
\newcommand{\bN}{\bs{\mathcal{N}}}
\newcommand{\bE}{\bs{\mathcal{E}}}
\newcommand{\bnd}{\bn_{\bs{D}_c}}
\newcommand{\bne}{\bn_{\bs{D}_c,  \rm e}}
\newcommand{\bnp}{\bn_{\bs{D}_c,  \rm p}}
\newcommand{\hn}{\hat{n}}
\newcommand{\vPsi}{\varPsi}
\newcommand{\Eql}{\overset{\triangle}{=}}

\newcommand{\one}{{\mathchoice {\rm 1\mskip-4mu l} {\rm 1\mskip-4mu l}
{\rm 1\mskip-4.5mu l} {\rm 1\mskip-5mu l}}}
\newcommand{\up}{\uparrow}
\newcommand{\down}{\downarrow}

\makeatletter
  
  \@addtoreset{equation}{section}
\makeatother

\title{ \sf
Stability of charge density waves in electron-phonon systems
}

\date{}

\author[]{Tadahiro Miyao\thanks{miyao@math.sci.hokudai.ac.jp}}

\affil[]{Department of Mathematics,  Hokkaido University

Sapporo 060-0810,  Japan}

\begin{document}
\maketitle

\begin{abstract}
With mathematical rigor, we demonstrate that electron-phonon interactions enhance the stability of charge density waves in low-temperature phases of many-electron systems.
Our proof method involves an appropriate application of the Pirogov--Sinai theory to  electron-phonon systems. Combining our findings with existing results, we obtain rigorous information regarding the low-temperature phase diagram for half-filled electron-phonon systems.
\end{abstract}

\tableofcontents

\section{Introduction}\label{Sec1}

\subsection{Background}

Previous studies in theoretical physics have suggested that many-electron systems exhibit a variety of phases at low temperatures. 
 The endeavors to rigorously understand the emergence of these phases based on the essential characteristics of many-electron systems,  such as Coulomb interactions among electrons, Fermi statistics, spin, and itinerancy,   still need to be exhaustive.  To make the main results of this paper understandable to the reader, we first review some of the previous studies relevant to this paper among these rigorous studies.

For each $L\in \BbbN$, we set  $\vLa_L=[-L, L-1]^d \cap \BbbZ^d$. 
The subscript $L$ is omitted hereafter unless there is a danger of confusion.
$\vLa$ is commonly considered as the canonical graph, with the vertex set being $\vLa$, and the edge set given by $E_{\vLa}=\{\la x, y\ra : x, y\in \vLa, \|x-y\|=1\}$, where  $\|x\|=\max\{|x_i| : i=1, \dots, d\}$,  the maximum norm of a vector $x$ in $\BbbZ^d$.

We consider a many-electron system on $\vLa$ described by the Hubbard model. The Hamiltonian of the system is given by
\begin{align}
H_{{\rm H}, \vLa}=&-t\sum_{\la x, y\ra\in E_{\vLa}}\sum_{\sigma=\up, \down}  (c_{x, \sigma}^* c_{y, \sigma}+c_{y, \sigma}^* c_{x, \sigma})
+ U \sum_{x \in \vLa}  \hn_{x, \up} \hn_{x, \down}
+W \sum_{\la x, y\ra\in E_{\vLa}} \hn_x \hn_y\no
&-\left(\mu+2dW+\frac{U}{2}\right)\sum_{x\in \vLa} \hn_x.\label{HamiHabbard}
\end{align}
$H_{{\rm H}, \vLa}$ is a self-adjoint operator acting  on the fermionic Fock space defined by
\be
\mathfrak{F}_{\rm e}(\vLa)=\bigoplus_{n=0}^{2|\vLa|} \bigwedge^n \ell^2(\vLa\times \{\up, \down\}),
\ee
where, for a given Hilbert space $\h$, $\bigwedge^n \h$ stands for the $n$-fold antisymmetric tensor product of $\h$, with $\bigwedge^0 \h =\BbbC$.
The electron annihilation- and creation operators are represented by $c_{x, \sigma}$ and $ c_{x, \sigma}^*$, respectively.\footnote{For the definitions of annihilation- and creation operators, see, for example, \cite{Arai2016, Bratteli1997}.} These operators satisfy the following anti-commutation relations:
\begin{align}
\{c_{x, \sigma}, c_{y, \tau}^*\}=\delta_{x, y} \delta_{\sigma, \tau},\quad \{c_{x, \sigma}, c_{y, \tau}\}=0, \label{CARs}
\end{align}
where $\delta_{x, y}$ denotes the Kronecker delta.
In addition, the number operators of the electrons are defined as follows:
\be
\hn_{x, \sigma}=c_{x, \sigma}^* c_{x, \sigma},\quad \hn_x=\hn_{x, \up}+\hn_{x, \down}.
\ee
In the case of $\mu=0$, it is noteworthy that $H_{{\rm H}, \vLa}$ describes a half-filled system.

Let us first consider the case $t=0$. In this case, the ground states of $H_{{\rm H}, \vLa}^{t=0}$ in the parameter region:
\be
S_{{\rm e}, 0}=\left\{
(U, \mu) \in \BbbR^2 : U<2dW ,\ |\mu|<2d \min\left\{W,\ W-\frac{U}{4d} \right\}
\right\}
\ee
 exhibit  long-range orders,  called {\it charge density waves} \cite{Borgs_1996}. 
 
 It is logical to ask whether charge density waves are stable in the case $t\neq 0$.
In their paper \cite{Borgs2000}, Borgs and Kotecky proved that charge density waves are stable in the parameter region:
\be
S_{{\rm e}, \vepsilon}=\left\{
(U, \mu) \in \BbbR^2 : U<2d(W-\vepsilon),\ |\mu|<2d \min\left\{W-\vepsilon,\ W-\vepsilon-\frac{U}{4d}\right\}
\right\},
\ee
 if the temperature is low enough and $\vepsilon>0$ and $|t|$ are sufficiently small.
The proof of this stability in \cite{Borgs2000} is based on the {\it Pirogov--Sinai theory}.
The Pirogov--Sinai theory has emerged as a potent approach in classical statistical mechanics for delineating first-order phase transitions and coexistence of phases at low temperatures \cite{Pirogov1975, sinai1982theory}. Noteworthy research works that utilize the Pirogov--Sinai theory to many-electron systems include \cite{Borgs2000, Datta1996, datta1996low, Datta2000}. Apart from these studies, the theory has also been extensively extended to quantum systems. For instance, refer to \cite{Borgs1996, Koteck1999} and references cited therein.

\subsection{Overview of results}\label{Sec1.2}

Electrons in tangible matter engage in interactions with phonons to varying extents.
Therefore, it is essential to examine how certain phases of many-electron systems are affected by electron-phonon interactions for a deeper understanding of the properties of matter.
The aim of this paper is to demonstrate that the stability of charge density waves is enhanced in the presence of electron-phonon interactions.

Consider the {\it Holstein--Hubbard model}, the most fundamental model describing the interaction between electrons and phonons. 
The Hamiltonian of the model is given by
\begin{align}
H_{{\rm HH}, \vLa}=H_{{\rm H}, \vLa}
+g\sum_{x\in \vLa}\hn_x(b_x+b_x^*)+\omega_0 \sum_{x\in \vLa} b_x^*b_x.
\end{align}
$H_{{\rm HH}, \vLa}$ acts on the following Hilbert space:
\be
\mathfrak{H}_{\vLa}=\mathfrak{F}_{\rm e}(\vLa)\otimes \mathfrak{F}_{\rm p}(\vLa), 
\ee
where $\mathfrak{F}_{\rm p}(\vLa)$ is the bosonic Fock space over $\ell^2(\vLa)$:
\be
\mathfrak{F}_{\rm p}(\vLa)=\bigoplus_{n=0}^{\infty} \otimes_{\rm s}^n \ell^2(\vLa);
\ee
$\otimes_{\rm s}^n \ell^2(\vLa)$ stands for the $n$-fold symmetric tensor product of $\ell^2(\vLa)$, with $\otimes_{\rm s}^0 \ell^2(\vLa)=\BbbC$.
The annihilation- and creation operators of phonons are denoted by $b_x$ and $ b_x^*$, respectively. These satisfy the standard commutation relations:\footnote{
More precisely, these commutation relations are satisfied on  finite particle subspaces of $\mathfrak{F}_{\rm p}(\vLa)$.
For the definition of annihilation- and creation operators, see, for example, \cite{Arai2016, Bratteli1997}.
}
\be
[b_x, b_y^*]=\delta_{x, y},\quad [b_x, b_y]=0.
\ee
The phonons are assumed to be dispersionless with energy $\omega_0>0$.
The parameter $g$ is the strength of the electron-phonon interaction.
By using the Kato--Rellich theorem \cite[Theorem X.12]{Reed1975}, we see that  $H_{{\rm HH}, \vLa}$
is  self-adjoint and bounded from below.
Applying the main theorem of this paper, we can prove the existence of charge density waves in the following region:
\be
S_{{\rm ep}, \vepsilon}=\left\{
(U, \mu) \in \BbbR^2 : U<2d(W-\vepsilon)+\frac{2g^2}{\omega_0} ,\ |\mu|<2d \min\left\{W-\vepsilon,\ W-\vepsilon-\frac{U}{4d}+\frac{g^2}{2d\omega_0}
\right\}
\right\}, \label{Se0}
\ee
 provided that $|t|$ is small enough, $\omega_0$ is large enough and at sufficiently low temperatures.
 The result indicates that the electron-phonon interaction has a  significant stabilizing effect on the charge density waves, as evidenced by the larger size of the region $S_{{\rm ep}, \vepsilon}$ in comparison to $S_{{\rm e}, \vepsilon}$.
Such effects have been anticipated in theoretical physics based on numerical computations and discussions relying on certain approximation theories. One of the primary significances of this paper lies in providing a mathematically rigorous foundation for these seemingly natural predictions.

We shall define certain symbols to provide a more comprehensive overview of the main theorem.
For each  $A=(\tilde{x}_1, \dots, \tilde{x}_{2k})$ $(\tilde{x}_i\in \BbbZ^d\times \{\up, \down\}\times \{-1, 1\})$, we define  $h_{A, {\rm e}}$ as 
\be
h_{A, {\rm e}}=c(\tilde{x}_{1})\cdots c(\tilde{x}_{2k}), 
\ee
where 
\begin{align}
c(x, \sigma, \kappa)
=\begin{cases}
c_{x, \sigma}^* & (\kappa=+1)\\
c_{x, \sigma} & (\kappa=-1).
\end{cases}
\end{align}
Let $\mathcal{A}_0$ denote the set of $A=(\tilde{x}_1, \dots, \tilde{x}_{2k})$ that are characterized by the following conditions: 
There exists a $p \in \{0, 1, \dots, 2k\}$ satisfying:
\begin{itemize}
\item  $\kappa_i=1$ for $1\le i  \le p$, $\kappa_i=-1$ for $i>p$; 
\item  
A certain total order can be established in the set $\BbbZ^d\times \{\up, \down\}$. Under this total order, the following conditions hold true: $(x_1, \sigma_1)>(x_2, \sigma_2)>\cdots> (x_p, \sigma_p)$ and $(x_{p+1}, \sigma_{p+1})<\cdots< (x_{2k}, \sigma_{2k})$.
\end{itemize}
For a given $A=(\tilde{x}_1, \dots, \tilde{x}_{2k})\in \mathcal{A}_0\, (\tilde{x}_i=(x_i, \sigma_i, \kappa_i))$, we define the support of $A$ by
$\supp A=\{x_1, \dots, x_{2k}\}$.
Furthermore, we set 
\be
\mathcal{A}_{\vLa}=\{A\in \mathcal{A}_0 : \supp A\subseteq \vLa\}\label{DefAL}
\ee
 and 
\be
O_{\vLa}:=\left\{ h_{A, \rm e} : A\in \mathcal{A}_{\vLa}
\right\}.
\ee

Let $\mathfrak{A}_{\vLa, \rm e}$ be the ${\rm C}^*$-algebra generated by $O_{\vLa}$.
The algebra of local observables for electrons is defined  by
\be
\mathfrak{A}_{\rm e}=\bigcup_{L\in \BbbN} \mathfrak{A}_{\vLa, \rm e}.
\ee
Given $\vPsi\in \mathfrak{A}_{\rm e}$, one can express it as a linear combination of operators of the form $h_{A, \rm e}$ using the canonical  anti-commutation relations \eqref{CARs}. The operators that appear in this expression are denoted as $h_{A_1, \rm e}, \dots, h_{A_n, \rm e}$. The support of $\vPsi$ is defined by $\supp \vPsi=\bigcup_{i=1}^n\supp A_i$.

Let $H_{{\rm P}, \vLa}$ be the Hamiltonian with the periodic boundary conditions imposed. For a detailed definition, see \eqref{DefHamiP}.
Given an observable $\varPsi\in \mathfrak{A}_{\rm e}$, the thermal expectation value of $\varPsi$ with respect to $H_{{\rm P}, \vLa}$ is defined as 
\be
\la \vPsi\ra^{\rm (P)}_{\beta, \vLa}=\frac{\Tr\left[\vPsi \, \ex^{-\beta H_{\rm P, \vLa}}\right]}{Z_{\rm P, \vLa}}, \quad
Z_{\rm P, \vLa}=\Tr\left[\ex^{-\beta H_{\rm P, \vLa}}\right].
\ee

For any periodic state $\la \cdot \ra$ on $\mathfrak{A}_{\rm e}$,
we define the staggered density as 
\be
\Delta=\lim_{L\to \infty} \frac{1}{|\vLa|} \sum_{x\in \vLa}(-1)^x \la \hn_x\ra.\label{DefDelta}
\ee
Note that the subsequent theorem ensures the exsistence of the limit on the right-hand side.

The spin operators,  $(S_x^{ (1)}, S_x^{(2)}, S_x^{ (3)})$,   at site $x$ are defined to be
\be
S_x^{(i)}=\frac{1}{2} \sum_{\sigma, \sigma\rq{}=\up, \down} c_{x, \sigma}^* (s^{(i)})_{\sigma, \sigma\rq{}}c_{x, \sigma},\quad i=1,2,3,
\ee
where $s^{(i)}\ (i=1, 2, 3)$ are the Pauli matrices:
\be
s^{(1)}=\begin{pmatrix}
0 & 1 \\
1 & 0
\end{pmatrix}, \ \ 
s^{(2)}=\begin{pmatrix}
0 & -\im \no
\im & 0
\end{pmatrix},\ \ 
s^{(3)}=\begin{pmatrix}
1 & 0 \no
0 & -1
\end{pmatrix}.
\ee
$(s^{(i)})_{\sigma, \sigma\rq{}}$ represents the matrix elements of $s^{(i)}$, with the correspondence $\up=1,\down=2$\footnote{Under this convention, for example, $(s^{(1)})_{\up, \up}=(s^{(1)})_{1, 1}=0$ and $ (s^{(1)})_{\up, \down}=(s^{(1)})_{1, 2}=1$.}.

Applying the main theorem of this paper (Theorem \ref{MainThm}) to the Holstein--Hubbard model, we can prove the following theorem:
\begin{Thm}\label{BriefTh}
Suppose that $0<\vepsilon<W$ and $(U, \mu)\in S_{{\rm ep}, \vepsilon}$.  There exist certain constants\footnote{
The constant $t_0$ depends on $d$ and $\beta_0$. Additionally, $\omega_*$ is chosen such that $\omega_* > \beta_0^{-1} \log 2$. For more details, refer to Remark \ref{GeInt} and Subsection \ref{ExHH}.
}  
$0<\beta_0<\infty, 0<\omega_*<\infty$ and $0<t_0<\infty$, such that,  if $\beta \ge \beta_0,\omega_0\ge \omega_*$ and $|t|\le t_0$, then the following {\rm (i)-(iii)} hold: 
\begin{itemize}
\item[\rm (i)] 
Given an arbitrary local observable $\vPsi\in \mathfrak{A}_{\rm e}$, the infinite volume limit: 
\ben
\la \vPsi\ra^{(\rm P)}_{\beta}=\lim_{L\to \infty} \la \vPsi\ra^{(\rm P)}_{\beta, \vLa}
\een
converges. The state $\la \cdot\ra_{\beta}^{(\rm P)}$ on $\mathfrak{A}_{\rm e}$ defined in this way can be represented by the convex combination of two pure states: 
\ben
\la \vPsi\ra^{(\rm P)}_{\beta}=\frac{1}{2} \la \vPsi\ra_{\beta}^{(+)}+\frac{1}{2}\la \vPsi\ra_{\beta}^{(-)}.
\een
The states $\la \cdot \ra_{\beta}^{(\pm)}$ describe charge density waves:
\ben
\la \hn_x\ra_{\beta}^{(+)}=\rho+(-1)^x \Delta,\quad\la \hn_x\ra_{\beta}^{(-)}=\rho-(-1)^x \Delta.
\een
Here, $\Delta^{(+)} = -\Delta^{(-)}= \Delta>0$, where $\Delta^{(+)}$ and $\Delta^{(-)}$ are staggered densities defined with respect to states $\la\cdot\ra_{\beta}^{(\pm)}$ in equation \eqref{DefDelta}.  Additionally, $\rho$ is given as 
\be
\rho=\lim_{L\to \infty} \frac{1}{|\vLa|}\sum_{x\in \vLa} \la \hn_x\ra_{\beta}^{(+)}
\ee
and  coincides with  the  density associated with  $\la\cdot \ra_{\beta}^{(-)}$: $
\rho=\lim_{L\to \infty} \frac{1}{|\vLa|}\sum_{x\in \vLa} \la \hn_x\ra_{\beta}^{(-)}
$.

\item[\rm (ii)]  No magnetic order exists. Namely, $\la S^{(i)}_x\ra_{\beta}^{ (\pm)}=0\ (i=1, 2, 3)$ holds for every $x\in \BbbZ^d$.
\item[\rm (iii)] The two-point correlation function concerning  arbitrary local observables $\vPsi,\varPhi\in \mathfrak{A}_{\rm e}$ decays exponentially. Namely, there are constants $C_{\varPsi, \varPhi}>0$ and $
\xi_{\ell}>0$ such that 
\begin{align*}
\left| \la \vPsi \varPhi\ra_{\beta}^{(\pm)}
-\la \varPsi \ra_{\beta}^{ (\pm)}\la \varPhi\ra_{\beta}^{ (\pm)}
\right|
\le C_{\varPsi, \varPhi}
\exp\left\{
-\frac{\mathrm{dist}(\supp \vPsi,  \supp \varPhi)}{\xi_{\ell}}
\right\},
\end{align*}
where, for any two finite subsets $A$ and $B$ of $\mathbb{Z}^d$, the distance between $A$ and $B$ is defined as $\mathrm{dist}(A, B)=\min\{\|x-y\| : x\in A, y\in B\}$.
\end{itemize}

\end{Thm}
This theorem is a direct consequence of the main theorem (Theorem \ref{MainThm}), which will be presented in Section \ref{Sec2}.
The guiding principle of the proof of Theorem \ref{MainThm} is to extend the method of analysis of many-electron systems by the Pirogov--Sinai theory developed in \cite{Borgs1996, Borgs2000} to electron-phonon systems.
So far, few rigorous studies on the Holstein--Hubbard model have been conducted\footnote{We mention \cite{Freericks1995,Miyao2015, Miyao2016, Miyao2018} as one of the few examples.}.  To the best of the author's knowledge, this is the first rigorous analysis of the Holstein--Hubbard model using the Pirogov--Sinai theory.
The Holstein interaction does not conserve the number of phonons. Therefore, analytical methods for bosonic systems, such as \cite{Borgs1996,Ueltschi2002}, cannot be directly applied, resulting in technical challenges.
For a more detailed explanation of these challenges, refer to Subsection \ref{IntSec4}.
This paper overcomes such challenges by employing specific representations of the correlation functions used in quantum field theory.
The fundamental properties of correlation functions for bosonic field operators used in this paper are summarized in Appendix \ref{AppCor}.

We obtain the following phase diagram by combining the existing results with Theorem \ref{BriefTh}:
\begin{Thm}[Brief summary]\label{Overview}

 Consider the half-filled system: $\mu=0$. In this case, the following {\rm (i)} and {\rm (ii)}  hold:
\begin{itemize}
\item[\rm (i)] The case where $U-2dW-2g^2/\omega_0>0$.

\begin{itemize}
\item[\rm (i-a)] In the ground state, a short-range antiferromagnetic order exists:
Set $S_x^{(\pm)}=S_x^{(1)}\pm \im S_x^{(2)}$. We denote by $\la \cdot \ra_{\rm gs}^{\rm (P)}$ the ground state expectation: $\displaystyle \la \varPsi\ra_{\rm gs}^{\rm (P)}
:=\lim_{L\to \infty} \lim_{\beta\to \infty} \la \varPsi\ra_{\beta, \vLa}^{(\rm P)}
$. It holds that 
 \be
 (-1)^{\|x-y\|}\left\la S_x^{(+)} S_y^{(-)}\right\ra_{\rm gs}^{(\rm P)}\ge 0\quad (x, y\in \BbbZ^d).
 \ee
 
\item[\rm (i-b)] There are no charge density waves at any given temperature.

\end{itemize}
\item[\rm (ii)] The case where $U-2dW-2g^2/\omega_0<0$. If $\beta$ and $\omega_0$ are  sufficiently large and $|t|$ is sufficiently small, then the following hold: 
\begin{itemize}
\item[\rm (ii-a)] There are charge density waves.
\item[\rm (ii-b)] No magnetic order exists.
\end{itemize}
\end{itemize}
\end{Thm}
\begin{proof}
(i-a) is proved in \cite{Miyao2016, Miyao2019}. As for (i-b), see \cite{Miyao2015}.
(ii) immediately follows from Theorem \ref{BriefTh}.
\end{proof}

\begin{Rem}
In this paper, we examine the Holstein--Hubbard model on the hypercubic lattice $\BbbZ^d$.
By modifying the lattice, we can showcase that the antiferromagnetic behavior elucidated in Theorem \ref{Overview} (i-a) displays features of long-range order. A notable exemplification of this occurrence can be discerned in the two-dimensional Lieb lattice\footnote{As for the Lieb lattice, see, for example, \cite[Fig. 10.2]{Tasaki2020}.}.
\end{Rem}

\subsection{Organization}
This paper is structured as follows:
In Section  \ref{Sec2}, we first introduce a general model describing electron-phonon interactions.
This comprehensive model encompasses the Holstein--Hubbard model as a specific instance.
Then, some settings and assumptions concerning this model are clarified, and the main theorem of this paper is stated. Theorem \ref{BriefTh}, presented earlier, follows immediately from the main theorem.
Sections  \ref{Sec3} through \ref{Sec5} are devoted to the proof of the main theorem. In Section \ref{Sec3}, the partition function for the model introduced in Section \ref{Sec2} is expressed as a contour model.
To apply the Pirogov--Sinai theory to the contour model, a detailed analysis of contour activities is given in Section \ref{Sec4}. Finally, in Section \ref{Sec5}, we apply the Pirogov--Sinai theory to our model using the results obtained in the previous sections to prove the main theorem.

\subsection*{Acknowledgements}
The author is grateful to Dr. Kameoka for pointing out the error in the definition of the order parameter.
T.M. was supported by JSPS KAKENHI Grant Numbers 18K03315, 20KK0304.\\
\vspace{2mm}

\noindent
{\bf Data Availability}\\
 Data sharing not applicable to this article as no datasets were generated or analysed during
the current study.
\vspace{2mm}

\noindent
{\bf Financial interests}\\
 The author declares he has no financial interests.

\section{Basic setup and main theorem}\label{Sec2}

\subsection{Fundamental properties of the Fock spaces}
For each $B\subseteq \BbbZ^d$, we define
\begin{align}
\mathcal{N}_{B,  \rm e}=
\{0, 1\}^{B\times \{\up, \down\}},\quad
\mathcal{N}_{B,  \rm p}=\BbbZ_+^B,\quad
\mathcal{N}_B= \mathcal{N}_{B,  \rm e}\times \mathcal{N}_{B,  \rm p}, \label{DefNB}
\end{align}
where $\BbbZ_+$ denotes the set of non-negative integers.
The elements of $\mathcal{N}_{B, \rm e}$ are commonly referred to as {\it electron configurations} within the  region $B$, while the elements of $\mathcal{N}_{B, \rm p}$ are referred to as {\it phonon configurations} within the same region. On the other hand, the elements of $\mathcal{N}_B$ are commonly referred to as {\it electron-phonon configurations} within the region $B$.

For any given  $\bn_{\vLa,    \rm e}=(n_{x, \sigma})_{(x, \sigma)\in \vLa\times \{\up, \down\}}\in \mathcal{N}_{\vLa,  \rm e}$, we define the state $\ket{\bn_{\vLa, \rm e}}$ as
\be
\ket{\bn_{\vLa,   \rm e}}=\prod_{x\in \vLa}(c_{x, \sigma}^*)^{n_{x, \sigma}} \ket{\varnothing}_{\vLa, {\rm e}},
\ee
where $\ket{\varnothing}_{\vLa, {\rm e}}$ stands for the Fock vacuum in $\mathfrak{F}_{\rm e}(\vLa)$, and we understand that   $A^0=\one$.
In this manuscript, the product $\prod_{x\in \vLa}A_x$ is defined as follows. Firstly, we assign natural numbers as labels to the elements of $\BbbZ^d$, i.e., $\BbbZ^d=\{x_1, x_2, \dots\}$. Then, we establish a total order on $\BbbZ^d$ based on this labeling such that $x_1<x_2<\cdots$.\footnote{
The definition of $\mathcal{A}_{\vLa}$ in Subsection \ref{Sec1.2} also includes an ordering of $\BbbZ^d\times \{\up, \down\}$. We suppose that  the ordering introduced here is consistent with the ordering introduced in Subsection \ref{Sec1.2}.
} Given this ordering, if $\vLa$ is a subset of $\BbbZ^d$ that can be expressed as $\vLa=\{x_{i_1}, x_{i_2}, \dots, x_{i_{|\vLa|}} : i_1<i_2<\cdots<x_{i_{|\vLa|}}\}$, we define $\prod_{x\in \vLa}A_x$ as $A_{x_{i_1}}A_{x_{i_2}}\cdots A_{x_{i_{|\vLa|}}}$.
With this definition, the set  $\{\ket{\bn_{\vLa,   \rm e}} : \bn_{\vLa, \rm e} \in \mathcal{N}_{\vLa, \rm e}\}$ forms a complete orthonormal system(CONS) of $\mathfrak{F}_{\rm e}(\vLa)$.

Similarly, for each $\bn_{\vLa,   \rm p}=(n_x )_{x\in \vLa}\in \mathcal{N}_{\vLa,  \rm p}$,
we define the state $\ket{\bn_{\vLa, \rm p}}$ as
\be
\ket{\bn_{\vLa,  \rm p}}=\prod_{x\in \vLa} \frac{1}{\sqrt{n_x!}} (b_x^*)^{n_x}\ket{\varnothing}_{\vLa, \rm p}, 
\ee
where $\ket{\varnothing}_{\vLa, {\rm p}}$ represents the Fock vacuum in $\mathfrak{F}_{\rm p}(\vLa)$.
Then the set $\{\ket{\bn_{\vLa,   \rm p}} : \bn_{\vLa, \rm p} \in \mathcal{N}_{\vLa, \rm p}\}$ forms  a CONS of $\mathfrak{F}_{\rm p}(\vLa)$.
Moreover,  for a given $\bn_{\vLa}=(\bn_{\vLa,   \rm e}, \bn_{\vLa,   \rm p}) \in \mathcal{N}_{\vLa}$, we set 
\be
\ket{\bn_{\vLa}}=\ket{\bn_{\vLa,   \rm e}} \otimes \ket{\bn_{\vLa,   \rm p}}.
\ee
Thus,  the set $\{ \ket{\bn_{\vLa}} : \bn_{\vLa} \in \mathcal{N}_{\vLa}\}$ forms a CONS of $\mathfrak{H}_{\vLa}$.

If $\vLa=\vLa_1\sqcup \vLa_2$, the disjoint union of subsets $\vLa_1$ and $\vLa_2$, then we can establish an identification between the fermionic Fock spaces as follows:
\be
\mathfrak{F}_{\rm e}(\vLa)=\mathfrak{F}_{\rm e}(\vLa_1)\otimes \mathfrak{F}_{\rm e}(\vLa_2).
\ee
This identification is made possible by the unitary operator $\iota_{\rm e}: \mathfrak{F}_{\rm e}(\vLa)\to \mathfrak{F}_{\rm e}(\vLa_1)\otimes \mathfrak{F}_{\rm e}(\vLa_2)$,  defined by
\be
\iota_{\rm e} c_{x, \sigma}\iota_{\rm e}^{-1}=
\begin{cases}
c_{x, \sigma}\otimes \mathbbm{1} & (x\in \vLa_1)\\
(-1)^{N_{\vLa_1, \rm e}} \otimes c_{x, \sigma} & (x\in \vLa_2)
\end{cases},\quad \iota_{\rm e}\ket{\varnothing}_{\vLa, {\rm e}}=\ket{\varnothing}_{\vLa_1, {\rm e}}\otimes \ket{\varnothing}_{\vLa_1, {\rm e}},  \label{FacC}
\ee
where   $N_{\vLa, \rm e}$ denotes the total number operator for the electrons, which is defined as $N_{\vLa, \rm e}=\sum_{x\in \vLa}\sum_{\sigma} n_{x, \sigma}$.

In a similar manner, we can also identify the bosonic Fock spaces as follows:
\be
\mathfrak{F}_{\rm p}(\vLa)=\mathfrak{F}_{\rm p}(\vLa_1)\otimes \mathfrak{F}_{\rm p}(\vLa_2).
\ee
This identiﬁcation  is achieved by the unitary operator $\iota_{\rm p}: \mathfrak{F}_{\rm p}(\vLa)\to \mathfrak{F}_{\rm p}(\vLa_1)\otimes \mathfrak{F}_{\rm p}(\vLa_2)$,  defined by 
\be
\iota_{\rm p} b_{x}\iota_{\rm p}^{-1}=
\begin{cases}
b_x\otimes \mathbbm{1} & (x\in \vLa_1)\\
\mathbbm{1} \otimes b_x & (x\in \vLa_2)
\end{cases},\quad \iota_{\rm p}\ket{\varnothing}_{\vLa, {\rm p}}=\ket{\varnothing}_{\vLa_1, {\rm p}}\otimes \ket{\varnothing}_{\vLa_2, {\rm p}}.  \label{FacB}
\ee

Using these identifications, we can see that the Hilbert space $\mathfrak{H}_{\vLa}$ factorizes as
\be
\mathfrak{H}_{\vLa}=\mathfrak{H}_{\vLa_1} \otimes \mathfrak{H}_{\vLa_2}. \label{HilIdn}
\ee
As a result of this factorization, we can express the basis vector $\ket{\bn_{\vLa}}$ in terms of the basis vectors $\ket{\bn_{\vLa_1}}$ and $\ket{\bn_{\vLa_2}}$ as follows:
\be
\ket{\bn_{\vLa}}=\vartheta_{\bn_{\vLa_1}, \bn_{\vLa_2}}\ket{\bn_{\vLa_1}} \otimes \ket{\bn_{\vLa_2}}, \label{FacVec}
\ee
where $\vartheta_{\bn_{\vLa_1}, \bn_{\vLa_2}}\in \{-1, 1\}$.

Assume that $\bn_{\vLa}=(\bn_{\vLa, \rm e}, \bn_{\vLa, \rm p})\in \mathcal{N}_{\vLa}$ is given.
For $B\subseteq \vLa$, let $\mathfrak{A}_{B, {\rm e}}^{(0)}$ be the commutative ${\rm C}^*$-algebra generated by $\{\hn_{x, \sigma} : (x, \sigma)\in B \times \{\up, \down\}\}$.  Note that $\mathfrak{A}_{B, {\rm e}}^{(0)}$ is a $\mathrm{C}^*$-subalgebra of $\mathfrak{A}_{\rm e}$ defined in Subsection \ref{Sec1.2}.
Suppose $\vLa=\vLa_1\sqcup \vLa_2$.  
The restriction of $\bn_{\vLa}$ to $\vLa_i$ ($i=1, 2$) is denoted by $\bn_{\vLa_i}$. Specifically, this is accomplished by defining $\bn_{\vLa_i, \rm e}$ as $(n_{x, \sigma} )_{ (x, \sigma)\in \vLa_i\times \{\uparrow, \downarrow\}}\in \mathcal{N}_{\vLa_i, \rm e}$ and $\bn_{\vLa_i, \rm p}$ as $(n_{x} )_{ x\in \vLa_i}\in \mathcal{N}_{\vLa_i, \rm p}$. Finally, we define $\bn_{\vLa_i}$ as $\bn_{\vLa_i}=(\bn_{\vLa_i, \rm e}, \bn_{\vLa_i, \rm p})$.
For $X\in \mathfrak{A}_{\vLa, \rm e}^{(0)}$, we define the conditional expectation value $\Braket{\bn_{\vLa_1,   \rm e}|X| \bn_{\vLa_1,   \rm e}}\in \mathfrak{A}^{(0)}_{\vLa_2, \rm e}$ as follows:
\begin{align}
\Braket{\bn_{\vLa,   \rm e}|X| \bn_{\vLa,   \rm e}}
= \Braket{ \bn_{\vLa_2,   \rm e}|\Braket{\bn_{\vLa_1,   \rm e}|X| \bn_{\vLa_1,   \rm e}}|\bn_{\vLa_2,   \rm e}}.
\end{align}
We then  define $\mathfrak{A}_{B, {\rm p}}^{(0)}$ as the commutative $*$-algebra generated by $\{b_x^*b_x : x\in B \}$. Note that  operators in $\mathfrak{A}_{B, {\rm p}}^{(0)}$ are typically unbounded.  For $Y\in \mathfrak{A}_{\vLa, {\rm p}}^{(0)}$, we define the conditional expectation value $\Braket{\bn_{\vLa_1,   \rm p}|Y| \bn_{\vLa_1,   \rm p}}\in \mathfrak{A}^{(0)}_{\vLa_2, \rm p}$ as follows:
\begin{align}
\Braket{\bn_{\vLa,   \rm p}|Y| \bn_{\vLa,   \rm p}}
= \Braket{ \bn_{\vLa_2,   \rm p}|\Braket{\bn_{\vLa_1,   \rm p}|Y| \bn_{\vLa_1,   \rm p}}|\bn_{\vLa_2,   \rm p}}.
\end{align}
Finally, we define $\mathfrak{A}_{B}^{(0)}=\mathfrak{A}_{B, {\rm e}}^{(0)}\otimes \mathfrak{A}_{B, {\rm p}}^{(0)}$ where the right-hand side denotes the algebraic tensor product. For $Z\in \mathfrak{A}_{\vLa}^{(0)}$, we can similarly define the conditional expectation value $\Braket{\bn_{\vLa_1}|Z| \bn_{\vLa_1}}\in \mathfrak{A}^{(0)}_{\vLa_2}$.

\subsection{Definition of the model}
In this section, we define a model to describe the electron-phonon system considered in this paper. To do so, we introduce a Hamiltonian that describes a many-electron system:
\be
H_{ \vLa, {\rm e}}=\lambda\sum_{A\in \mathcal{A}_{ \vLa}} t_Ah_{A, {\rm e}}+\sum_{x\in \vLa} \varPhi_{x}(\bs{\hn}).
\ee
Here, we assume that $h_{A, \rm e}\in O_{\vLa}$. To ensure self-adjointness, we assume that
\be
\overline{t}_A=t_{A^*}, 
\ee
where, if we represent $A\in \mathcal{A}_0$ as $A=\{\tilde{x}_1, \dots, \tilde{x}_{2k}\}$, then $A^*$ is defined by
\be
A^*=(\tilde{x}_{2k}^*, \dots, \tilde{x}_{1}^*),\quad
(x, \sigma, \kappa)^*=(x, \sigma, -\kappa).
\ee
The parameter $\lambda \in \BbbR$ represents the strength of the quantum perturbation term.
For each $x\in \vLa$, a real-valued function $\varPhi_x(\bn_{\rm e})$ defined on $\mathcal{N}_{\BbbZ^d, \rm e}$ is given in advance. The operator $\varPhi_x(\bs \hn)$ is defined as replacing the number $n_{x,\sigma}$ with the operator $\hn_{x,\sigma}$.  We assume that there exists a non-negative number $R_0$ such that $\varPhi_x(\bs n)$ depends only on $n_{y,\sigma}$ for $y$ with $\mathrm{dist}(x,y)\le R_0$. We also assume that $t_A$ and $\varPhi_x$ depend on a vector parameter $\nu=(\nu_1,\dots,\nu_{k-1})\in \BbbR^{k-1}$.  For instance, in the case of the   Hubbard model given in \eqref{HamiHabbard}, we have $\nu=(\mu, U, W)$, and $t_A=t\ (A\in E_{\vLa})$ is independent of $\nu$.

Assuming that the Holstein interaction represents the interaction between electrons and phonons,  the Hamiltonian under investigation in this paper is given by
\begin{equation}
 \mathsf{ H}_{\vLa}=H_{\vLa, \rm e}+g\sum_{x\in \vLa}\hn_x(b_x+b_x^*)+\omega_0 N_{\vLa, \rm p}.\label{DefHamiHH}
  \end{equation}
Here, $N_{\vLa, \rm p}$ denotes the number operator for phonons in the region $\vLa$:
\begin{equation} N_{\vLa, \rm p}=\sum_{x\in \vLa} b_x^*b_x. \end{equation}
Using Kato-Rellich's theorem \cite[Theorem X.12]{Reed1975}, it can be proven that $\mathsf{H}_{\vLa}$ is self-adjoint on $\D(N_{\vLa, \rm p})$ and bounded below.
This paper focuses on the Hamiltonian given in \eqref{DefHamiHH}, which describes the interactions between electrons and phonons at the same site. However, it is worth noting that our results can be extended to more general electron-phonon interactions, as discussed in Remark \ref{GeInt}.

Including the Holstein interaction term in the Hamiltonian $\mathsf{H}_{\vLa}$ presents a challenge in analyzing the low-temperature phase as this operator is unbounded and does not conserve the number of phonons. To surmount this obstacle, we propose the {\it Lang--Firsov transformation} to convert $\mathsf{H}_{\vLa}$ into a more analytically tractable Hamiltonian. To this end, we define the self-adjoint operators, $p_x$ and $q_x$, for each $x\in \vLa$, as follows:
\begin{align}
p_x=\frac{\im }{\sqrt{2}}(\overline{b_x^*-b_x}), \ \ \  q_ x=\frac{1}{\sqrt{2}} (\overline{b_x^*+b_x}), \label{Defpq}
\end{align}
where $\overline{A}$ is the closure of $A$. 
It is known that these operators satisfy the standard commutation relation: 
$[q_x, p_y]=\im \delta_{x,y}$. 
Let
\begin{align}
L
=-\im  \alpha \sum_{x\in\vLa }\hn_xp_x\quad \left(\alpha=\frac{\sqrt2 g}{\omega_0}\right).
\end{align}
The Lang--Firsov transformation, which was first introduced in \cite{Lang1963}, is a unitary operator defined by  $\mathcal{U}=\ex^{\im \frac{\pi}{2}N_{\vLa, \rm p}}\ex^{L}$.
We  define the Hamiltonian $H_{\vLa}$ as the Lang--Firsov transformed Hamiltonian: 
$H_{\vLa}=\mathcal{U} \mathsf{H}_{\vLa} \mathcal{U}^{-1}$.
Using
the following formulas:
 \begin{align}
\ex^{\im \frac{\pi}{2}N_{\vLa, \rm p}}q_x\ex^{-\im \frac{\pi}{2}N_{\vLa, \rm p}}
&=p_x, \quad
\ex^{\im \frac{\pi}{2}N_{\vLa, \rm p}}p_x\ex^{-\im \frac{\pi}{2}N_{\vLa, \rm p}}
=-q_x, \label{LF1}\\
\ex^{L}c_{x,\sigma}\ex^{-L}
&=\ex^{ \im \alpha p_x} c_{x,\sigma}, \quad
\ex^{L}b_x \ex^{-L}
=b_x-\frac{\alpha}{\sqrt{2}} n_x,  \label{LF3}
\end{align} 
we can express $H_{\vLa}$ as:
\be
H_{\vLa}=\lambda \sum_{A\in \mathcal{A}_{\vLa}} t_A h_A
+\sum_{x\in \vLa} \varPhi_{{\rm eff}, x}(\bs{\hn})+\omega_0 N_{\vLa, \rm p},
\label{ELHami}
\ee
where 
\begin{align}
h_A&=\ex^{\im \vartheta_A}h_{A, {\rm e}},\\
\vartheta_A&=\sum_{\tilde{x} \in A} \vartheta_{\tilde{x}}, 
\quad
\vartheta_{\tilde{x}}=-\kappa \alpha q_x\quad(\tilde{x}=(x, \sigma, \kappa)); \label{DefTht}
\end{align}
and  the operator $\varPhi_{{\rm eff}, x}(\bs{\hn})$ is  defined  through the effective potential,  given by
\be
 \varPhi_{{\rm eff}, x}(\bs{n}_{\rm e})= \varPhi_{x}(\bs{n}_{\rm e})- \frac{\omega_0 \alpha^2}{2} (n_{x, \up}+n_{x, \down})^2\
\ee
for each $\bn_{\rm e}=\{n_{x, \sigma} : (x, \sigma)\in \BbbZ^d\times \{\up, \down\}\}\in \mathcal{N}_{\BbbZ^d, \rm e}$.
Note that $H_{\vLa}$ acts on the Hilbert space $\mathfrak{H}_{\overline{\vLa}}$. Here, $\overline{B}$ denotes the set obtained by thickening $B$ by a distance of $R_0$: 
\be
\overline{B}=\{x\in \vLa : \mathrm{dist}(x, B) \le R_0\}. \label{DefBarB}
\ee
 We will investigate $H_{\vLa}$ defined in this way in detail below. 
Here, we note that the Lang--Firsov transformation results in the vanishing of the electron-phonon interaction term in \eqref{DefHamiHH}, and instead, the real hopping term $h_{A, \rm e}$ becomes the operator $h_A$ containing a phase consisting of bosonic operators. The electron-phonon interaction in \eqref{DefHamiHH} is an unbounded operator, making it mathematically challenging to handle. However, through the Lang--Firsov transformation, this problem is reduced to the analysis of bounded hopping terms with bosonic phase operators. As we will see later, the transformed Hamiltonian is more amenable to mathematical analysis.

Introducing a new vector parameter $\overline{\nu}=(g, \nu_1, \dots, \nu_{k-1})$, we see that $\varPhi_{{\rm eff}, x}$ and $t_A$ depend on $\overline{\nu}\in \mathscr{U}$, where $\mathscr{U}$ is an open subset of $\BbbR^k$.
 Unless there is a risk of confusion, we will not explicitly indicate the dependence of $\varPhi_{{\rm eff}, x}$ and $t_A$ on $\overline{\nu}$.

The classical part $H_{\vLa}^{(0)}$ of $H_{\vLa}$ is given by
\begin{align}
H_{\vLa}^{(0)}&=\sum_{x\in \vLa}\varPhi_{{\rm eff}, x}(\bs{\hn})+\omega_0 N_{\vLa, \rm p}=\sum_{x\in \vLa}H_x^{(0)},\quad H_x^{(0)}:=\varPhi_{{\rm eff}, x}(\bs{\hn})+\omega_0 b_x^*b_x. \label{Defvtheta}
\end{align}
On the other hand, the quantum part of $H_{\vLa}$ is defined as follows:
\be
H^{(Q)}_{\vLa}=\lambda \sum_{A\in \mathcal{A}_{\vLa}} t_A h_A. \label{DefQPart}
\ee

\subsection{Basic assumptions}\label{Sec2.3}

Let us define several terms according to the textbook \cite{Friedli2017}. Consider the formal classical Hamiltonian: 
\be
H_{\rm e}^{(0)}(\bn_{\rm e})=\sum_{x\in \BbbZ^d} \varPhi_{{\rm eff}, x}(\bn_{\rm e}),\quad \bn_{\rm e}\in \mathcal{N}_{\BbbZ^d, \rm e}. \label{ClElHa}
\ee
Since the right-hand side is an infinite sum, it generally does not converge. How then should we define the ground state configurations of $H^{(0)}_{\rm e}(\bn_{\rm e})$?

Two configurations $\bn_{\rm e}$ and $\tilde{\bn}_{\rm e}\in \mathcal{N}_{\BbbZ^d, {\rm e}}$ are said to be  {\it equal at infinity}  if there exists a finite subset $\vLa\Subset \BbbZ^d$ such that $\bn_{\vLa^c, \rm e}=\tilde{\bn}_{\vLa^c, \rm e }$. Here, $\vLa \Subset \BbbZ^d$ means that $\vLa$ is a finite subset of $\BbbZ^d$. Moreover, for an electron configuration $\bn_{\rm e}=(n_{x, \sigma})_{ (x, \sigma)\in \BbbZ^d\times \{\up, \down\}}\in \mathcal{N}_{\BbbZ^d, \rm e}$ in $\BbbZ^d$, we denote its restriction to $B\subset \BbbZ^d$ by $\bn_{B, \rm e}$, that is, $\bn_{B, \rm e}:=(n_{x, \sigma} )_{ (x, \sigma)\in B\times \{\up, \down\}}\in \mathcal{N}_{B, \rm e}$.

When $\bn_{\rm e}$ and $\tilde{\bn}_{\rm e}$ are equal at infinity, we define the {\it relative Hamiltonian} as follows:
\be
H_{\rm e, rel}^{(0)}(\bn_{\rm e}\, |\, \tilde{\bn}_{\rm e})=\sum_{x\in \BbbZ^d}\left\{\varPhi_{{\rm eff}, x}(\bn_{\rm e})-\varPhi_{{\rm eff}, x}(\tilde{\bn}_{\rm e})\right\}.
\ee
Note that the right-hand side is a finite sum and hence convergent.
An electron configuration $\bs{g}_{\rm e}\in \mathcal{N}_{\BbbZ^d, \rm e}$ is called a {\it ground state configuration} of $H_{\rm e}^{(0)}(\cdot)$ if it satisfies the following condition for all configurations $\bn_{\rm e}$ that are equal at infinity with $\bs{g}_{\rm e}$:
\be
H_{\rm e, rel}^{(0)}(\bn_{\rm e }\, |\, \bs{g}_{\rm e}) \ge 0.
\ee

Hereafter, we assume that the set of ground state configurations of $H_{\rm e}^{(0)}(\cdot)$ consists solely of periodic configurations.
 Additionally, we assume that given any parameter $\overline{\nu}\in \mathscr{U}$, the set $G(\overline{\nu})$ of ground state configurations of $H_{\rm e}^{(0)}(\cdot)$ is a subset of  a collection of periodic configurations $\{\bs{g}^{(1)}_{\rm e}, \dots, \bs{g}^{(r)}_{\rm e}\}$. Consequently, we define the ground state energy of $H_{\rm e}^{(0)}(\cdot)$ as
\be
e_{\rm e}=\min_{\ell} e_{\ell}(\overline{\nu}),
\ee
where
\be
e_{\ell}(\overline{\nu})=\lim_{L\to \infty} \frac{1}{|\vLa|} \sum_{x\in \vLa}\varPhi_{{\rm eff}, x}(\bs{g}^{(\ell)}_{\rm e}).
\ee
Without loss of generality, we can assume that $\varPhi_{{\rm eff}, x}(\bs{g}^{(\ell)}_{\rm e})$ does not depend on $x\in \vLa$.

Let $\bn_{\rm e}$ be a given electron configuration. Then, a site $x$ is said to be {\it in a  configuration $\bs{g}_{\rm e}^{(\ell)}$} if 
$
\bs{g}^{(\ell)}_{U(x), \rm e}=\bn_{U(x), \rm e} 
$
holds, where 
\be
U(x)=\{y\in \BbbZ^d : \mathrm{dsit}(x, y)\le R_0\}. \label{DefU}
\ee
 On the other hand, if $x$ is not in any of the  configuration $\bs{g}^{(1)}_{\rm e}, \dots, \bs{g}^{(r)}_{\rm e}$, it is said to be {\it  in an excited configuration}.

We assume the following conditions for the ground state configurations of $H^{(0)}_{\rm e}$:
\begin{description}
\item[\hypertarget{A1}{(A. 1)}] There is a $\overline{\nu}_0\in \mathscr{U}$ such that $G(\overline{\nu}_0)=\{\bs{g}^{(1)}_{\rm e}, \dots, \bs{g}^{(r)}_{\rm e}\}$ holds.

\item[\hypertarget{A2}{(A. 2)}] The functions $e_{\ell}(\overline{\nu})$ are  $C^1$  in $\mathscr{U}$.
\item[\hypertarget{A3}{(A. 3)}] The matrix
\be
E=\left[
\frac{\partial e_{\ell}(\overline{\nu})}{\partial \overline{\nu}_i}
\right]_{\ell, i}
\ee
has rank $r-1$ for all $\overline{\nu}\in \mathscr{U}$,  and  the inverse matrices of corresponding submatrices are uniformly bounded from above.
\end{description}

 We assume the following conditions for $\varPhi_{{\rm eff}, x}$:
\begin{description}

\item[\hypertarget{A4}{(A. 4)}] 

Let $\bn_{\rm e}$ be a given electron configuration. If $x$ is in an excited configuration, then there exists a positive constant  $\gamma_{\rm e}$ such that
\be
\varPhi_{{\rm eff}, x}(\bn_{\rm e})\ge e_{\rm e}+\gamma_{\rm e}.
\ee
This condition is often referred to as the {\it Peierls condition}.
\item[\hypertarget{A5}{(A. 5)}] 
There exists a constant $C_0>0$ satisfying 
\be
\left|
\frac{\partial}{\partial \overline{\nu}_i} \varPhi_{{\rm eff}, x}(\bn_{\rm e})
\right| \le C_0
\ee
for all $i=1, \dots, k, \overline{\nu}\in \mathscr{U} $,  $x\in \BbbZ^d$ and $\bn_{\rm e} \in \mathcal{N}_{\BbbZ^d, \rm e}$.
\end{description}

We make the following assumptions regarding the quantum perturbation term:
\begin{description}

\item[\hypertarget{A6}{(A. 6)}] 
Let $\bs{t}=\{t_A: A\in \mathcal{A}_0\}$. We assume that $\bs{t}$ is translation invariant, i.e., $t_{\tau_x(A)}=t_A$ for all $x\in \BbbZ^d$ and $A\in \mathcal{A}_0$. 
Here, $\tau_x(A)$ is obtained by translating $A$ in the $x$ direction. More precisely, if we represent $A$ as $A=(\tilde{x}_1, \dots, \tilde{x}_n)\ (\tilde{x_i}=(x_i, \sigma_i, \kappa_i))$, then $\tau_x(A)$ is defined as $(\tilde{y}_1, \dots, \tilde{y}_n)$, where $\tilde{y}_i=(x_i+x, \sigma_i, \kappa_i)$.

Furthermore, the strength of the quantum perturbation is limited in the following sense: for $\gamma\ge 0$, we define the Sobolev norm as
\be
\|\bs{t}\|_{\gamma}:=\sum_{A : x\in \supp A}\left(
|t_A|+\sum_{i=1}^{k}\left|
\frac{\partial}{\partial \overline{\nu}_i} t_A
\right|
\right)\, \ex^{\gamma |\supp A|}.
\ee
Then, for sufficiently large $\gamma_Q>0$, we have $\|\bs{t}\|_{\gamma_Q} <\infty$.

\end{description}

In Subsection \ref{ExHH}, we confirm that the Holstein--Hubbard model discussed in Section \ref{Sec1} indeed satisfies the aforementioned conditions.

\subsection{Boundary conditions}
We consider the free Hamiltonian of the electron-phonon system,  given by
\be
H^{(0)}=H^{(0)}_{\rm e}+\omega_0 \sum_{x\in \BbbZ^d} b_x^*b_x,
\ee
where 
$H_{\mathrm{e}}^{(0)}$ is a formal Hamiltonian defined by replacing the electron configuration $\bn_{\rm e}$ with the family of number operators for electrons $\bs{\hn}=(\hn_{x, \sigma})_{x, \sigma}$ in the definition \eqref{ClElHa} of the classical Hamiltonian $H_{\rm e}^{(0)}(\cdot)$.
$H^{(0)}$ is a formal operator defined by replacing the electron configuration $\bn_{\rm e}$ and the phonon configuration $\bn_{\rm p}$ in the classical Hamiltonian
\be
H^{(0)}(\bn)=H^{(0)}_{\rm e}(\bn_{\rm e})+\omega_0 \sum_{x\in \BbbZ^d} n_{x, {\rm p}}
\quad (\bn=(\bn_{\rm e}, \bn_{\rm p})\in\mathcal{N}_{\BbbZ^d} )
\ee
with the families of number operators for electrons $\bs{\hn}$ and phonons $(b_x^*b_x)_x$, respectively. 
Here, we denote $\bn_{\rm p}=(n_{x, {\rm p}})_{x\in \BbbZ^d}$.
The ground state configurations of $H^{(0)}(\bn)$  are  given by $\bs{g}^{(\ell)}=(\bs{g}^{(\ell)}_{\rm e}, \bs{0})\in \mathcal{N}_{\BbbZ^d}$.\footnote{
To be more precise, firstly, we define the relative Hamiltonian $H^{(0)}_{\rm rel}(\cdot \, |\, \cdot)$ with respect to $H^{(0)}(\cdot)$ and conduct a discussion similar to that in Subsection \ref{Sec2.3}, which reveals that $\bs{g}^{(\ell)}$ corresponds to the ground state configuration of $H^{(0)}(\cdot)$.} Therefore, for $B\Subset \BbbZ^d$, we have $\ket{\bs{g}^{(\ell)}_B}=\ket{\bs{g}_{B, {\rm e}}^{(\ell)}}\otimes \ket{\varnothing}_{B, \rm p}$.
The Hamiltonian with  boundary conditions $\ell$ is defined as:
\be
H_{\ell , \vLa}=H^{(Q)}_{\vLa}+H_{\ell , \vLa}^{(0)}, \quad H_{\ell , \vLa}^{(0)}=\Braket{\bs{g}^{(\ell )}_{\partial \vLa}|H_{\vLa}^{(0)}|\bs{g}^{(\ell )}_{\partial \vLa} }, 
\ee
where, $\partial \vLa=\{x\in \BbbZ^d : x\notin \vLa,\, \mathrm{dist}(x, \vLa) \le R_0\}$ and $\Braket{\bs{g}^{(\ell )}_{\partial \vLa}|H_{\vLa}^{(0)}|\bs{g}^{(\ell )}_{\partial \vLa} }$ represents a  conditional expectation value. 
Recall that $H^{(Q)}_{\vLa}$ is defined in \eqref{DefQPart}.
 In this case, $H_{\ell , \vLa}$ acts on the Hilbert space $\mathfrak{H}_{\vLa}$.

For each $x\in \vLa$, we define
\begin{align}
\varPhi^{(\ell)}_{{\rm eff}, x}(\bn_{ \rm e})
=\begin{cases}
\varPhi_{{\rm eff}, x}(\bn_{U(x),  \rm e}) & \mbox{if $U(x) \setminus \vLa=\varnothing$}\\
\varPhi_{{\rm eff}, x}\left(\bn_{\bn_{U(x)\cap \vLa}, \rm e}\times \bs{g}^{(\ell)}_{U(x) \setminus \vLa, \rm e} \right) & \mbox{if $U(x) \setminus \vLa\ne \varnothing$}.
\end{cases}
\end{align}
Here, the product of two configurations is defined as follows: Given $A, B\Subset\BbbZ^d\ (A\cap B=\varnothing)$,  and two configurations $\bn_{A, \rm e}=(n^A_{x, \sigma} )_{ (x, \sigma)\in A\times \{\uparrow, \downarrow\}}\in \mathcal{N}_{A, \rm e}$ and $ \bn_{B, \rm e}=(n^B_{x, \sigma} )_{ (x, \sigma)\in B\times \{\uparrow, \downarrow\}}\in \mathcal{N}_{B, \rm e}$, we define their product $\bn_{A, \rm e}\times \bn_{B, \rm e}=(n^{A\sqcup B}_{x, \sigma})_{(x, \sigma)\in A\sqcup B\times \{\uparrow, \downarrow\}}\in \mathcal{N}_{A\sqcup B, \rm e}$ as
\begin{align}
n^{A\sqcup B}_{x, \sigma}=
\begin{cases}
n^A_{x, \sigma} & \mbox{ if $x\in A$}\\
n^B_{x, \sigma} &\mbox{ if $x\in B$.}
 \end{cases}\label{DefProdN}
\end{align}
Then, we can express the Hamiltonian with  boundary conditions $\ell$  as
\be
H_{\ell, \vLa}=H_{\ell, \vLa, \rm e}+\omega_0 N_{\vLa, \rm p},\quad
H_{\ell, \vLa, \rm e}=\sum_{x\in \vLa}\varPhi^{(\ell)}_{{\rm eff}, x}(\bs{\hat{n}}_{ \rm e}).
\ee
Note that $\ex^{-\beta H_{\ell, \vLa}}$ is a  trace class operator for every $\beta>0$.\footnote{
For example, this fact can be easily proven using the Golden--Thompson inequality \cite[Theorem 8.3]{simon2005}.
}
The thermal expectation value under  boundary conditions $\ell$ is defined as
\be
\la \vPsi\ra^{(\ell)}_{\beta, \vLa}=\frac{\Tr\left[\vPsi\, \ex^{-\beta H_{{\ell}, \vLa}}\right]}{Z_{{\ell}, \vLa}},\quad
Z_{{\ell}, \vLa}=\Tr\left[ \ex^{-\beta H_{{\ell}, \vLa}}\right],
\ee
where $\vPsi$ is a local observable.

Next, let us explain the  Hamiltonian with  periodic boundary conditions. We define the Hamiltonian on the torus $\vLa_{\rm P}=(\BbbZ/L\BbbZ)^d$ by
\be
H_{{\rm P}, \vLa}=\lambda\sum_{A : \supp A\subset \vLa_{\rm P}} t_A h_A
+\sum_{x\in \vLa_{\rm P}}H_x^{(0)}. \label{DefHamiP}
\ee
Here, the sum over the first term on the right-hand side is taken over $A$ such that $\supp A$ does not wrap around $\vLa_{\rm P}$. The thermal expectation value under periodic boundary conditions is defined as
\be
\la \vPsi\ra^{(\rm P)}_{\beta, \vLa}=\frac{\Tr\left[\vPsi\, \ex^{-\beta H_{{\rm P}, \vLa}}\right]}{Z_{{\rm P}, \vLa}},\quad
Z_{{\rm P}, \vLa}=\Tr\left[ \ex^{-\beta H_{{\rm P}, \vLa}}\right],
\ee
where  $\vPsi$ is a local observable.

\begin{Rem}\label{Remlamb}
\upshape
The parameter $\lambda$ representing the strength of quantum perturbation can be extended to complex numbers. Indeed, $\sum_{A\in \mathcal{A}_{\vLa}} t_A h_A$ is a bounded self-adjoint  operator, and both $H^{(0)}_{\ell, \vLa}$ and $\sum_{x\in \vLa_{\rm P}}H_x^{(0)}$ are self-adjoint operators bounded from below. Hence, for all $\lambda\in \BbbC$, $-H_{{\rm P}, \vLa}$ and $-H_{\ell, \vLa}$ are $m$-accretive. Therefore, by Lumer--Phillips theorem (see, for example, \cite[Theorem 6.12]{Schmdgen2012}), $\ex^{-\beta H_{{\rm P}, \vLa}}$ and $\ex^{-\beta H_{\ell, \vLa}}$ can be defined. Furthermore, detailed analysis in Sections \ref{Sec3} and \ref{Sec4} demonstrates that these operators are trace-class.
\end{Rem}

\subsection{Main theorem}
The main theorem of this paper is as follows.
\begin{Thm}\label{MainThm}
For $d\ge 2$, let us assume the conditions \hyperlink{A1}{\bf (A. 1)}--\hyperlink{A6}{\bf (A. 6)}. Then, there exist constants $0<\beta_0<\infty$, $0<\omega_*<\infty$ and $0<\gamma_Q<\infty$ such that, for all $\beta\ge \beta_0$, $\omega_0\ge\omega_*$, $\gamma\ge \gamma_Q$, and $\lambda\in \BbbC$ satisfying
\be
|\lambda| \le \lambda_0:=\frac{1}{\ex\, \beta_0 \|T\|_{\gamma}},
\ee
there exist positive constants $\xi_{\ell}$ and continuously differentiable functions $f_{\ell}(\overline{\nu})\ (\ell=1, \dots, r)$, such that as long as
\be
\mathrm{Re}f_{\ell}(\overline{\nu})-\min_m \mathrm{Re} f_m(\overline{\nu})=0,
\ee
the following {\rm (i)}-{\rm (viii)} hold:

\begin{itemize}
\item[\rm (i)] There exists an infinite volume free energy corresponding to $Z_{\ell, \vLa}$, which is equal to $f_{\ell}$:
\be
f_{\ell}(\overline{\nu})=-\frac{1}{\beta} \lim_{L\to \infty} \frac{1}{|\vLa|} \log Z_{\ell, \vLa}.
\ee
\item[\rm (ii)] For any local observable $\varPsi\in \mathfrak{A}_{\rm e}$, the infinite volume limit
\be
\la \varPsi\ra^{(\ell)}_{\beta}=\lim_{L\to \infty} \la \varPsi\ra^{(\ell)}_{\beta, \vLa}
\ee
exists.
\item[\rm (iii)]
For any local observables $\varPsi, \varPhi\in \mathfrak{A}_{\rm e}$, their two-point correlation function decays exponentially as follows:
\begin{align}
\left| \la \vPsi \varPhi\ra^{(\ell)}_{\beta} 
-\la \varPsi \ra^{(\ell)}_{\beta} \la \varPhi\ra^{(\ell)}_{\beta}
\right|
\le C_{\varPsi, \varPhi}
\exp\left\{
-\frac{\mathrm{dist}(\supp \vPsi,  \supp \varPhi)}{\xi_{\ell}}
\right\}.
\end{align}

\item[\rm (iv)]
For any given local observable $\vPsi\in \mathfrak{A}_{\rm e}$, we define its ground state expectation value as follows:
\be
\la \vPsi\ra_{\rm gs}^{(\ell)}=\lim_{L\to \infty} \lim_{\beta\to \infty} \la \vPsi\ra_{\beta, \vLa}^{(\ell)}.
\ee
It follows that the limit on the right-hand side exists. Additionally, we have the following inequality:
\be
\left|
\la \vPsi\ra^{(\ell)}_{\rm gs}-\la \vPsi\ra_{\beta, \vLa}^{(\ell)}
\right|
\le C_{\varPsi} \exp\left[- \zeta \min\left\{
 \frac{\beta}{2\beta_0},\ \mathrm{dist}(\supp \vPsi, \partial \vLa)
 \right\}
 \right],
\ee
where $C_{\varPsi}$ and $\zeta$ are positive constants.
\item[\rm (v)]
Let $P_x^{(m)}$ be the projection onto the local ground state configuration ${\bs g}_{U(x)}^{(m)}$, given by $P_x^{(m)}=\ket{\bs{g}_{U(x)}^{(m)}}\!\bra{\bs{g}_{U(x)}^{(m)}}$. There exists a  positive constant  $C>0$ such that  the following inequality holds:
\be
\left|
\la P_x^{(m)} \ra^{(\ell)}_{\beta}-\delta_{\ell, m}
\right|\le C\,  \ex^{-\gamma}.
\ee

\item[\rm (vi)]There exists a certain constant $C>0$ such that the following holds:
\begin{align}
|f_{\ell}(\overline{\nu})-e_{\ell}(\overline{\nu})| \le C\, \ex^{-\gamma},\quad
\left|
\frac{\partial}{\partial\overline{\nu}_i}\{f_{\ell}(\overline{\nu})-e_{\ell}(\overline{\nu})\}
\right|\le C \, \ex^{-\gamma}
\end{align}
for all $\overline{\nu}\in \mathscr{U}$ and $\ell\in \{1, \dots, r\}$.
\item[\rm (vii)]
For any local observable $\varPsi\in \mathfrak{A}_{\rm e}$, the following infinite volume limit exists under periodic boundary conditions:
\be
\la \vPsi\ra^{(\rm P)}_{\beta}=\lim_{L\to \infty} \la \vPsi\ra_{\beta, \vLa}^{(\rm P)}.
\ee
Furthermore, one obtains 
\be
\la \vPsi\ra^{(\rm P)}_{\beta}=\sum_{\ell\in \mathcal{Q}}\frac{1}{|\mathcal{Q}|} \la \vPsi\ra^{(\ell)}_{\beta}, 
\ee
where
$
\mathcal{Q}=\{\ell\in \{1, \dots, r\} : \mathrm{Re}f_{\ell}(\overline{\nu})-\min_m \mathrm{Re} f_m(\overline{\nu})=0\}
$.
\item[\rm (viii)] 
Suppose that the family of local observables $\{Q_{\vLa} : L\in \BbbN\}$ is extensive in the following sense:
\begin{itemize}
\item[\rm 1.] 
Each $Q_{\vLa}$ can be  expressed in the form $Q_{\vLa}=\sum_{x\in \vLa} Q_{x, \vLa}\quad (Q_{x, \vLa}\in \mathfrak{A}_{\vLa, \rm e})$. Furthermore, $\sup_{x, L}|\supp Q_{x, \vLa}|<\infty$ and $\sup_{x, L}\|Q_{x, \vLa}\|<\infty$.
\item[\rm 2.] $[Q_{\vLa}, \ex^{-\beta H_{\ell, \vLa}}]=0$ for all $\beta \ge 0$.
\item[\rm 3.] 
Each local   state $\ket{\bs{g}_{\vLa, \rm e}^{(\ell)}}$ is an eigenvector of $Q_{\vLa}$ and the corresponding eigenvalue is of the form $\rho_{\vLa}^{(\ell)} |\vLa|$. Furthermore, there exists an infinite volume limit: $\displaystyle \rho_{\rm cl}^{(\ell)}=\lim_{L\to \infty} \rho_{\vLa}^{(\ell)}$.
\end{itemize}

Then we have the following:
\begin{itemize}
\item[\rm (viii-a)]
Regarding the ground state expectations of $Q_{\vLa}$, the following holds:
\be
\rho_{\rm cl}^{(\ell)}=\lim_{L\to \infty} \frac{1}{|\vLa|} \la Q_{\vLa}\ra_{\rm gs}^{(\ell)}.
\ee
\item[\rm (viii-b)] For each  $\beta>0$, we define
\be
\rho^{(\ell)}(\beta)=\lim_{L\to \infty} \frac{1}{|\vLa|}\la Q_{\vLa}\ra_{\beta, \vLa}^{(\ell)}.
\ee
Then, for all $\beta \ge \beta_0$, there exists a positive constant $C$ such that the following holds:
\be
\left|
\rho^{(\ell)}(\beta)-\rho_{\rm cl}^{(\ell)}
\right|\le C \, \ex^{- c\beta}.
\ee

\end{itemize}
\end{itemize}
\end{Thm}

\begin{Rem}\label{GeInt}
\rm 
\begin{itemize}
\item
The constants $\beta_0$, $\gamma_Q$ and $\omega_*$  are selected to fulfill the conditions outlined in Propositions \ref{ActBd1} and \ref{DelRhoEst}, as well as Lemma \ref{PartialQ}.
\item
The method presented in this paper can also be applied to more generalized electron-phonon interactions, such as:
\be
\sum_{x, y\in \vLa} g_{x, y} \hn_x(b_y+b_y^*),
\ee
where $g_{x, y}\in \BbbR$ satisfies the following conditions: (i) translational invariance: $g_{x, y}=g_{x-y, o}$, (ii) $g_{x, y}$ is of finite range: $g_{x, y}=0$ for $x, y\in \BbbZ^d$ with $\|x-y\|>R_0$. Theorem \ref{MainThm} holds for such generalized interactions, although the proof becomes more complex.
\end{itemize}
\end{Rem}

\subsection{Example: Holstein--Hubbard Model}\label{ExHH}
For the reader's convenience, here we provide an overview of how the Holstein--Hubbard model discussed in Section \ref{Sec1} indeed satisfies the conditions \hyperlink{A1}{\bf (A. 1)}--\hyperlink{A6}{\bf (A. 6)}.
For $\bn_{\rm e}=(n_{x, \sigma})_{ x\in \vLa, \sigma\in \{\up, \down\} }\in \mathcal{N}_{\BbbZ^d, {\rm e}}$, the classical Hamiltonian $H_{\rm e}^{(0)}(\bn_{\rm e})$  is formally given by:
\be
H_{\rm e}^{(0)}(\bn_{\rm e})=
U_{\rm eff}\sum_{x\in \BbbZ^d}n_{x, \up} n_{x, \down}
+W \sum_{\la x, y\ra} n_x n_y
-\left(\mu+2dW+\frac{U_{\rm eff}}{2}\right)\sum_{x\in \BbbZ^d}n_x.
\ee
Here, $n_x=n_{x, \up}+n_{x, \down}$, and $U_{\rm eff}=U-\omega_0\alpha^2$.
Also, $\sum_{\la x, y\ra}$ denotes the sum over all edges of $\BbbZ^d$, i.e., $E_{\BbbZ^d}:=\{\la x, y\ra \in \BbbZ^d\times \BbbZ^d: \|x-y\|=1\}$.
To reveal the ground state configuration of $H_{\rm e}^{(0)}$, it is convenient to introduce classical spin variables $s_x=n_x-1$. Note that $s_x$ takes values $-1, 0, 1$. Rewriting $H_{\rm e}^{(0)}$ in terms of $s_x$, we get:
\be
H_{\rm e}^{(0)}(\bs{s})=\frac{U_{\rm eff}}{2}\sum_{x\in \BbbZ^d} s_x^2+W\sum_{\la x, y \ra} s_x s_y-\mu\sum_{x\in \BbbZ^d} s_x.
\ee
Ignoring the constant term diverging with $|\vLa|$\footnote{As the ground state configuration is determined by the relative Hamiltonian, such constant terms can be ignored.}, $H_{\rm e}^{(0)}$ can be expressed as:
\be
H_{\rm e}^{(0)}(\bs s)=\sum_{\la x, y\ra} h(s_x, s_y), 
\ee
where
\be
h(s_x, s_y)=Ws_xs_y+\frac{U_{\rm eff}}{4d}(s_x^2+s_y^2)-\frac{\mu}{2d}(s_x+s_y).
\ee
Now, let us define $ \varPhi_{{\rm eff}, x}(\bs{n}_{\rm e})$ as
\be
\varPhi_{{\rm eff}, x}(\bs{n}_{\rm e})=\frac{1}{2}\sum_{y: \|x-y\|=1} h(s_x, s_y),
\ee
and with this, we can express $H_{\rm e}^{(0)}=\sum_{x\in \BbbZ^d} \varPhi_{{\rm eff}, x}$. Here, the symbol $\sum_{y : \|x-y\|=1}$ represents the sum over the given $x$'s nearest neighboring lattice points.

In the paper \cite{Borgs_1996}, Borgs {\it et al.} elucidate the phase diagram concerning the ground state configurations of $H_{\rm e}^{(0)}$. The obtained results are depicted in Figure \ref{Phase}.
\begin{figure}[h]
 \begin{center}
  \begin{tikzpicture}[scale=1]

\draw[->,>=stealth,very thick] (-6,0)--(4.5,0) node[right]{$U_{\rm eff}/4d$}; 
\draw[->,>=stealth,very thick] (0,-4)--(0,4) node[left]{$\mu/2d$}; 
\draw [domain=0:1,smooth] plot (\x,-\x+2);
\draw [domain=0:1,smooth] plot (\x,\x-2);
\draw [domain=0:1,smooth] (1, -1)--(1, 1);
\draw [domain=0:1,smooth] (-5.5, 2)--(0, 2);
\draw [domain=0:1,smooth] (-5.5, -2)--(0, -2);
\draw [domain=0:2,smooth] plot (\x,\x+2);
\draw [domain=1:4,smooth] plot (\x,\x);
\draw [domain=0:2,smooth] plot (\x,-\x-2);
\draw [domain=1:4,smooth] plot (\x,-\x);
\draw (1,0)node[below right]{$W/2$};
\draw (0,2)node[above left]{$W$};
\draw (0,-2)node[below left]{$-W$};
\draw (0,1)node[left]{$W/2$};
\draw (0,-1)node[left]{$-W/2$};
\draw (-3,3)node[left]{\Large$H_2$};
\draw (-3,-3)node[left]{\Large$H_0$};
\draw (-3,-0.5)node[left]{\Large$S_{\rm ep, 0}$};
\draw (3, 3)node[left]{\Large$S_{\rm ep, +}$};
\draw (3, -3)node[left]{\Large$S_{\rm ep, -}$};
\draw (3.8, -0.5)node[left]{\Large$H_1$};
\end{tikzpicture}
  \caption{
  Phase diagram of the ground state configurations for
  $H_{\rm e}^{(0)}$.
  }\label{Phase}
        \end{center}
 \end{figure}
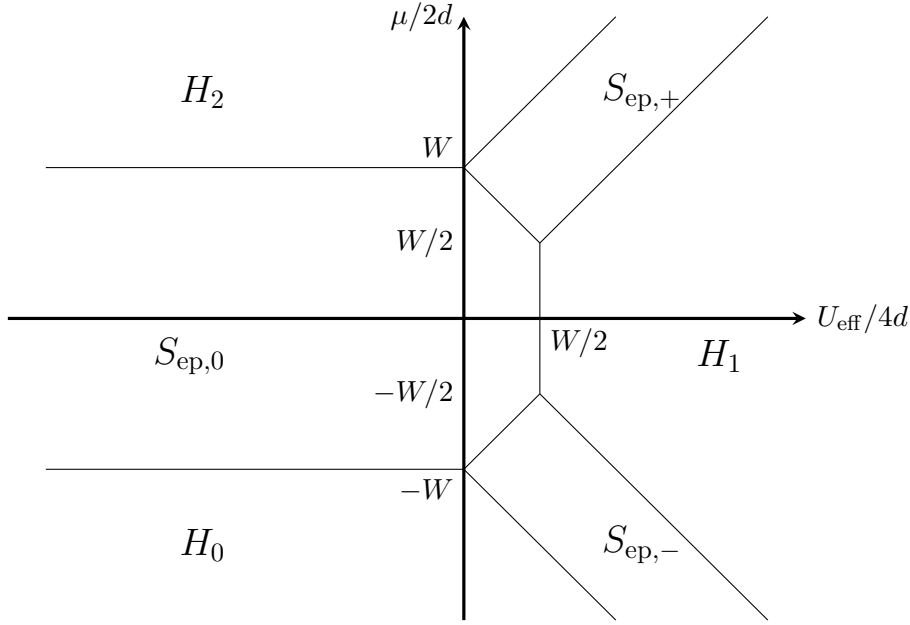
Let us  briefly explain the discussion in \cite{Borgs_1996} below.
Considering translational symmetry, the phase diagram can be constructed by examining pairs of spins $(s_x, s_y)$ that minimize the values of $h(s_x, s_y)$. We refer to $x\in \BbbZ^d$ as an even site (resp., odd site) when $\sum_{i=1}^d|x_i|$ is even (resp., odd). 
Let us  fix a pair of sites $(x, y)$, with $x$ being an even site and $y$ being an odd site. In this case, the possible combinations of $(s_x, s_y)$ are:
\be
(s_x, s_y)=(1, 1), (1, 0), (0, 0), (-1, 0), (1, -1), (-1, -1).
\ee
Note that pairs obtained by swapping the values of $s_x$ and $s_y$ are physically equivalent, so we do not distinguish them for a while.
The values of $h(s_x, s_y)$ corresponding to these configurations are as follows:
\begin{description}
\item$(s_x, s_y)=(1, 1)$:     $h(s_x, s_y)=W+\frac{U_{\rm eff}}{2d}-\frac{\mu}{d}$.
\item $(s_x, s_y)=(1, 0)$:     $h(s_x, s_y)=\frac{U_{\rm eff}}{4d}-\frac{\mu}{2d}$.
\item $(s_x, s_y)=(0, 0)$:    $h(s_x, s_y)=0$.
\item $(s_x, s_y)=(-1, 0)$:   $h(s_x, s_y)=\frac{U_{\rm eff}}{4d}+\frac{\mu}{2d}$.
\item $(s_x, s_y)=(1, -1)$:   $h(s_x, s_y)=-W+\frac{U_{\rm eff}}{2d}$.
\item $(s_x, s_y)=(-1, -1)$:  $h(s_x, s_y)=W+\frac{U_{\rm eff}}{2d}+\frac{\mu}{d}$.
\end{description}

By considering the minimum values of the obtained $h(s_x, s_y)$, we obtain Figure \ref{Phase}. In each region of Figure \ref{Phase}, the ground state configurations of $H_{\rm e}^{(0)}$ are obtained by periodically extending the pair configurations $(s_x, s_y)$ that minimize $h(s_x, s_y)$ over the entire $\BbbZ^d$.
The open regions 
$S_{\rm ep, +}, S_{\rm ep, -}, S_{\rm ep, 0}$ correspond to $(s_x, s_y)=(1, 0), (-1, 0), (-1,  1)$, respectively.
The open regions $H_0, H_1, H_2$ correspond to $(s_x, s_y)=(-1, -1), (0, 0), (1, 1)$, respectively. 

In $H_0, H_1, H_2$, the ground states of $H_{\rm e}^{(0)}$ do not exhibit charge density waves. On the other hand, in the open regions $S_{\rm ep, +}, S_{\rm ep, -}$, and $S_{\rm ep, 0}$, charge density waves appear in the ground state of $H_{\rm e}^{(0)}$.

In the pair spin configurations defining the regions $H_1, S_{\rm ep, +}, S_{\rm ep, -}$, it is essential to note that either $s_x=0$ or $s_y=0$. Returning to the definition of $s_x$, the case $s_x=0$ corresponds to two possibilities: $n_{x, \up}=1, n_{x, \down}=0$, or $n_{x, \up}=0, n_{x, \down}=1$. Of course, the same applies to $s_y=0$.
This observation indicates that the ground state configurations in these regions have infinite degeneracy. Consequently, the analysis within these regions falls outside the scope of the main theorem.

Let us delve into the region $\mathscr{U}=S_{\rm ep, \vepsilon}$  given by \eqref{Se0}.  Note that $S_{\rm ep, \vepsilon}$ is a region obtained by ``narrowing down" the region $S_{\rm ep, 0}$ by an amount $\vepsilon$.
 In this region,  one ground state configuration is such that $s_x=1$ for all even sites and $s_x=-1$ for all odd sites. Another configuration is obtained by swapping the values of $s_x$ on even and odd sites for the first configuration. Therefore, the ground state in this region is two-fold degenerate, hence $r=2$.
We note that our parameters are given by $\overline{\nu}=(g, \mu, U, W)$, and $e_{\rm e}=-W+\frac{U_{\rm eff}}{2d}$.
From this, conditions \hyperlink{A1}{\bf (A. 1)} to \hyperlink{A4}{\bf (A. 4)} are almost evident. For \hyperlink{A5}{\bf (A. 5)}, we note that 
$\left|
\frac{\partial}{\partial g} \varPhi_{{\rm eff}, x}(\bn_{\rm e})
\right|=c |g|$.
Therefore, when defining 
$S_{\rm ep, \vepsilon}$, it is necessary to fix $g_0>0$
 arbitrarily and restrict $|g|<g_0$.
Verifying the fulfillment of assumption \hyperlink{A6}{\bf (A. 6)} for every $\gamma_Q>0$ in the nearest-neighbor hopping of electrons, as given in \eqref{HamiHabbard}, is straightforward.
Furthermore, it is easily seen that $\|\bs{t}\|_{\gamma}=2d\,  \ex^{2\gamma}$ and $\lambda_0=(2d \beta_0\ex^{2\gamma_Q+1})^{-1}$.

 In regions $H_0$ and $H_2$, the ground state configurations are unique. Therefore, the main theorem can be applied to these regions, enabling the proof of the absence of charge density waves in the low-temperature phase.

\section{Connection to contour models}\label{Sec3}
\subsection{The  strategy for the proof of Theorem \ref{MainThm}}\label{StrMnTh}
Our approach to establish the main theorem involves the application of the Pirogov--Sinai theory to the Hamiltonian governing the electron-phonon interacting system as defined in \eqref{ELHami}.

In general, the Pirogov--Sinai theory is one of the few mathematical frameworks capable of constructing the phase diagram of the low-temperature phase when the phase diagram of the ground state of the model under consideration is well understood.
For a classical spin model on a $d$-dimensional hypercubic lattice, the Pirogov--Sinai theory proceeds in two steps:
\begin{description}
\item[\bf (C. 1)] Express the partition function of the model (or thermal expectations of physical observables) in terms of a contour model in $d$-dimensional space.
\item[\bf (C. 2)] Apply powerful analytical techniques, such as cluster expansions, to the obtained contour model to reveal the characteristics of the low-temperature phase.
\end{description}
Understanding of the Pirogov--Sinai theory in classical systems has significantly progressed, with an abundance of references available. Among them, the exposition in \cite{Friedli2017} is particularly well-suited for an introduction to this theory.

In the quantum systems on a $d$-dimensional hypercubic lattice, the Pirogov--Sinai theory is applicable when the quantum system Hamiltonian under consideration, as described in Section \ref{Sec2}, can be separated into a classical part and a quantum perturbation part.
The prerequisite for the application of this theory is a well-understood phase diagram of the ground state of the classical part of the Hamiltonian. In this scenario, even in the presence of quantum perturbations, it becomes feasible to construct the phase diagram of the low-temperature phase.
The Pirogov--Sinai theory for quantum systems consists of the following steps:
\begin{description}
\item[\hypertarget{Q1}{(Q. 1)}] Express the partition function of the Hamiltonian under consideration   in terms of a classical contour model in a $d+1$-dimensional space-time.
\item[\hypertarget{Q2}{(Q. 2)}] Apply and extend the existing Pirogov--Sinai theory for classical systems to the obtained space-time contour model, thereby analyzing the low-temperature phase.
\end{description}
Regarding step \hyperlink{Q2}{\bf (Q. 2)}, it is  worth noting that the foundation for this extension has already been established in \cite{Borgs1996}.

The objective of this section is to judiciously extend the methodology from the paper \cite{Borgs1996, Borgs2000} to achieve step \hyperlink{Q1}{\bf (Q. 1)} in the context of electron-phonon interacting systems. The electron-phonon interaction model considered in this paper has a more intricate quantum perturbation component compared to models analyzed in previous studies \cite{Borgs2000, Borgs1996, Datta1996, datta1996low, Ueltschi2002}. Consequently, it is essential to acknowledge the difficulty in constructing an appropriate contour model.
In Section \ref{Sec4}, we demonstrate that the space-time contour model constructed in this section satisfies the fundamental assumptions for applying the theory in \cite{Borgs1996}.

\subsection{Contour representation of the partition  function}\label{Sec3.1}

Let $M\in \BbbN$ and $\tilde{\beta}=\beta/M$. The partition function can be  expressed as follows:
\be
Z_{\ell, \vLa}=\Tr\left[T_{\vLa}^M\right],\quad T_{\vLa}=\ex^{-\tilde{\beta} H_{\ell , \vLa}}.\label{PartitionF}
\ee
We will choose the value of $M$ later so that it satisfies the appropriate conditions.
By using the Duhamel formula\footnote{By Remark \ref{Remlamb}, the formula \eqref{DefDuha} holds for complex numbers $\lambda$ as well.}, one obtains
\begin{align}
T_{\vLa}=\sum_{\bs{m}\in \vGa_{\vLa}} \frac{(-\lambda \bs{t})^{\bs{m}}}{\bs{m}!} \int d\bs{\tau}^{\bs{m}}
T_{\vLa}(\bs{\tau},  \bs{m}). \label{DefDuha}
\end{align}
Note that the right-hand side of \eqref{DefDuha} converges under the operator norm topology.
We shall provide explicit definitions for the symbols appearing in the above formula. Firstly, we define $\vGa_{\vLa}$ as $\vGa_{\vLa}=\BbbZ_+^{\mathcal{A}_{\vLa}}$. For every $\bs{m}$ belonging to $\vGa_{\vLa}$, we define the support of $\bs{m}$ as $\supp \bs{m}:=\{A\in\mathcal{A}_{\vLa} : m_A\neq 0\}$,
where $\mathcal{A}_{\vLa}$ is defined by \eqref{DefAL}. Additionally, we let $\mathfrak{S}_n$ denote the symmetric group on $n$ objects.
Expressing $\supp \bs{m}$ as $\mathcal{A}=\{A_1, \dots, A_k\}$ and $\bs{\tau}$ as $\{\tau_A^1, \dots, \tau_A^{m_A} : A\in \mathcal{A}\}$ with $\tau_A^i\in [0, \tilde{\beta}]$, we can define $T_{\vLa}(\bs{\tau}, \bs{m})$ as follows: Firstly, we select $\pi\in \mathfrak{S}_{|\bs{m}|}$ that satisfies the following conditions:
\begin{align}
(s_1, \dots, s_{|\bs{m}|})=\pi(\tau_{A_1}^1, \dots, \tau_{A_1}^{m_{A_1}},\dots, \tau_{A_k}^1, \dots, \tau_{A_k}^{m_{A_k}}), \quad
s_1\le s_2 \le \cdots \le s_{|\bs{m}|}. \label{PiS}
\end{align}
With $\pi$ selected in this manner, we define $\tilde{h}_i$ as
$
(\tilde{h}_1, \dots, \tilde{h}_{|\bs{m}|}):=\pi(h_{A_1}, \dots, h_{A_1}, \dots, h_{A_k}, \dots, h_{A_k})
$ and then define $T_{\vLa}(\bs{\tau}, \bs{m})$ as
\be
T_{\vLa}(\bs{\tau}, \bs{m})=\tilde{h}_1(s_1) \tilde{h}_2(s_2)\cdots \tilde{h}_{|\bs{m}|}(s_{|\bs{m}|})\, \ex^{-\tilde{\beta} H_{\ell, \vLa}^{(0)}},\quad \tilde{h}_i(s)=\ex^{-s H_{\ell , \vLa}^{(0)}}\tilde{h}_i\, \ex^{s H_{\ell , \vLa}^{(0)}}.
\ee
In addition, to simplify notation, we introduce the following symbols:
\be
\int d\bs{\tau}^{\bs m}=\prod_{A \in \supp \bs{m}} \int_0^{\tilde{\beta}} d\tau_{A}^{1}\cdots \int_0^{\tilde{\beta}} d\tau_{A}^{m_{A}}, \quad
\bs{m}!=\prod_{A \in \mathcal{A}_{\vLa}} m_{A}!,\quad
(-\lambda \bs{t})^{\bs m}=\prod_{A \in \mathcal{A}_{\vLa}} (-\lambda t_{A})^{m_{A}}.
\ee
Given the above setup, the partition function can be expressed as follows:
\begin{align}
Z_{\ell, \vLa}&= \sum_{\bs{n}_{\vLa}\in \N_{\vLa}}\Bigg[
 \prod_{t=1}^M \sum_{\bs{m}_t\in \vGa_{\vLa}}\frac{(-\lambda \bs{t})^{\bs{m}_t}}{\bs{m}_t!} \int d\bs{\tau}_t^{\bs{m}_t}\Bigg]
\bra{\bs{n}_{\vLa}} T_{\vLa}(\bs{\tau}_1,  \bs{m}_1) \cdots  T_{\vLa}(\bs{\tau}_M,  \bs{m}_M)
\ket{\bs{n}_{\vLa}}.\label{PartFu1}
\end{align}
As seen in Section \ref{Sec4}, the right-hand side of the above equation absolutely converges when the parameters satisfy appropriate conditions.
Therefore, the interchange of the sum and integral that was implicitly performed in deriving the equation \eqref{PartFu1} is justified.
To simplify the description further, we define the symbol: 
\be
\int_{\varXi_B} \mathscr{D}(\omega):=\sum_{{\mathcal{A}=\{A_1, \dots, A_k\}}\atop{\cup_{j=1}^k A_j=B}} \sum_{\bs{m}: \supp \bs{m}=\mathcal{A}}\Bigg[
\frac{(-\lambda \bs{t})^{\bs m}}{\bs{m}!} \int d\bs{\tau}
\Bigg], 
\ee
where 
\begin{align}
\varXi_B:=\left\{\omega=(\bs{\tau}, \bs{m}) : \bigcup_{A\in \supp \bs{m}}A=B, \ \mbox{$\bs{\tau}$ satisfies \eqref{PiS}}\right\}.
\end{align}
Then we can represent $T_{\vLa}$ in the following manner:
\be
T_{\vLa}=\sum_{B\subseteq \vLa} T_{\vLa}(B),\label{TLDec}
\ee
where 
\begin{align}
T_{\vLa}(B)
=\int_{\varXi_B} \mathscr{D}(\omega)T_{\vLa}(\omega).
\end{align}
It is worth noting that if $B, B\rq{}\subseteq\vLa$ satisfy  $\overline{B}\cap \overline{B}\rq{}=\varnothing$, then we have
\be
[T_{\vLa}(B), T_{\vLa}(B\rq{})]=0.
\ee
Here, recall that $\overline{B}$ is defined in \eqref{DefBarB}.

Given an electron-phonon configuration $\bn_{\vLa}=(\bn_{\rm e}, \bn_{\rm p})\in \mathcal{N}_{\vLa}$ in $\vLa$,   we denote the energy associated with the configuration $\bn_B$ for each region $B\subseteq \vLa$ as $E(\bs{n}_{B})$: 
\be
E(\bs{n}_{B})= \sum_{x\in B}\left\{ \varPhi^{(\ell)}_{{\rm eff}, x}(\bn_{{\rm e}})+\omega_0 n_{x, \rm p}\right\}.
\ee
Moreover, 
we define $H_{\ell, B}^{(0)}(\bn_{\overline{\partial B}})$  by means of the conditional expectation value as:
\be
H_{\ell, B}^{(0)}(\bn_{\overline{\partial B}})=\Braket{\bn_{\overline{\partial B}}|H_{\ell, \overline{B}}^{(0)}|\bn_{\overline{\partial B}}},\quad  H_{\ell, \overline{B}}^{(0)}=\sum_{x\in \overline{B}}\{ \varPhi_{{\rm eff}, x}^{(\ell)}(\bs{\hn}_{\rm e})+\omega_0 b_x^*b_x\}, 
\ee
 where 
  $\overline{\partial B}
:=\{x\notin B : \mathrm{dist}(x, B) \le 2R_0\}$. 
Note that 
$H_{\ell, B}^{(0)}(\bn_{\overline{\partial B}})
$ 
belongs to $\mathfrak{A}^{(0)}_B$.
 We readily confirm that 
\be
H_{\ell, \vLa}^{(0)}\ket{\bn_{\vLa}}=\left\{
E(\bn_{\vLa \setminus \overline{B}})+H_{\ell, B}^{(0)}(\bn_{\overline{\partial B}})
\right\}\ket{\bn_{\vLa}}.
\ee
For any given $\omega=(\bs{\tau}, \bs{m})\in \varXi_B$, we define
\begin{align}
\mathcal{T}_B(\omega,   \bn_{\overline{\partial B}})&=\tilde{h}_{1, \bn_{\overline{\partial B}}}(s_1) \tilde{h}_{2, \bn_{\overline{\partial B}}}(s_2)\cdots \tilde{h}_{|\bs{m}|, \bn_{\overline{\partial B}}}(s_{|\bs{m}|})\, \ex^{-\tilde{\beta} H_{\ell, B}^{(0)}(\bn_{\overline{\partial B}})}, \\
 \tilde{h}_{i, \bn_{\overline{\partial B}}}(s)&=\ex^{-s H_{\ell , B}^{(0)}(\bn_{\overline{\partial B}})}\tilde{h}_i\, \ex^{s H_{\ell , B}^{(0)}(\bn_{\overline{\partial B}})}.
\end{align}
Based on the above observation, we conclude that
\be
T_{\vLa}(\omega) \ket{\bn_{\vLa}}=\ex^{-\tilde{\beta}E(\bs{n}_{\vLa\setminus \overline{B}})}\mathcal{T}_B(\omega,   \bn_{\overline{\partial B}})\ket{\bn_{\vLa}} \quad (\omega\in \varXi_B).
\ee
Thus, by utilizing the identity  $\mathbbm{1}=\sum_{\bn_{\vLa}\in \mathcal{N}_{\vLa}} P_{\bn_{\vLa}}\ (P_{\bn_{\vLa}}=\ket{\bs{n}_{\vLa}}\!\bra{\bs{n}_{\vLa}})$, we obtain the following expression:
\begin{align}
T_{\vLa}(B)
=\sum_{\bs{n}_{\vLa}\in \N_{\vLa}} \ex^{-\tilde{\beta} E(\bs{n}_{\vLa\setminus \overline{B}})}
\mathcal{T}_B(\bn_{\overline{\partial B}}) P_{\bn_{\vLa}},   \label{TLB}
\end{align}
where
\be
\mathcal{T}_B(\bn_{\overline{\partial B}})=\int_{\vXi_B}\mathscr{D}(\omega) \mathcal{T}_B(\omega,   \bn_{\overline{\partial B}}). 
\ee

Thus, by setting $\varSigma=(B, \bs{n}_{\vLa})$, and using equations \eqref{TLDec} and \eqref{TLB}, we can express $T_{\vLa}$ as follows:
\begin{equation}
T_{\vLa}=\sum_{\varSigma} K_{\vLa}(\varSigma), \label{TSigma}
\end{equation}
where 
\be
K(\vSi)=
K(B, \bs{n}_{\vLa})=
\ex^{-\tilde{\beta}E(\bs{n}_{\vLa\setminus \overline{B}})}
\mathcal{T}_B(\bn_{\overline{\partial B}}) P_{\bn_{\vLa}}.
\ee
In particular, when $B=\varnothing$, we have
\begin{equation}
K(\varnothing, \bs{n}_{\vLa})=\ex^{-\tilde{\beta}E(\bs{n}_{\vLa})}P_{\bn_{\vLa}}
=\ex^{-\tilde{\beta} H_{\ell, \vLa}^{(0)}} P_{\bn_{\vLa}}.
\end{equation}
Therefore, using equation \eqref{PartFu1}, we obtain the following expression for the partition function:
\begin{equation}
Z_{\ell, \vLa}=\Bigg[\prod_{t=1}^M\sum_{\vSi_t}\Bigg] W(\vSi_1, \dots, \vSi_M),
\quad W(\vSi_1, \dots, \vSi_M)=\Tr\big[
K_{\vSi_1}\cdots K_{\vSi_M}
\big]. \label{ZSi}
\end{equation}

To construct the contour representation of the partition function $Z_{\ell, \vLa}$, 
we will now introduce the concept of a ``space-time lattice"  by $\mathbb{L}=\BbbZ^d \times \{1, \dots, M\}_{\rm P}$, where  $M$   coresponds  to  the temporal extent of the lattice. Additionally, we introduce a sub-lattice obtained by restricting the spatial component  of the space-time lattice $\mathbb{L}$ to $\vLa$: $\mathbb{L}_{\vLa}=\vLa \times \{1, \dots, M\}_{\rm P}$.

To extend the lattices to continuous space, we define $\mathbb{T}=\BbbR^d\times [0, M]_{\rm P}$, which represents the space-time continuum, and $\mathbb{T}_{\vLa}=\{x\in \BbbR^d : \mathrm{dist}(x, \vLa)\le 1/2\}\times [0, M]_{\rm P}$. Note that the subscript $\rm P$ indicates that we are imposing periodicity, i.e., $M+a \equiv a\ (a>0)$.

We denote the elementary cube centered at the site $(x,t)$ in the space-time lattice as $C(x,t)$. To be more precise, $C(x, t)$ is defined as 
 $C(x, t)=\{(y, s)\in\mathbb{T}_{\vLa} : 
|t-s|\le 1/2, |x_i-y_i|\le 1/2\ (i=1, \dots, d)
 \}$.  
Let $\mathsf{E}$ denote the collection of all sets obtained by assembling the elementary cubes centered at the space-time lattice points of $\mathbb{L}$.
For every $\bs{B}\in \mathsf{E}$, 
we denote the set of centers of elementary cubes contained in the set $\bs{B}$ as $[\bs{B}]$, that is,  
\be[\bs{B}] = \bs{B} \cap \mathbb{L}.\ee
The spatial and temporal projections of $[\bs{B}]$ are defined as follows:
\begin{align}
[\bs{B}]_{\rm S}&=\{x\in \BbbZ^d : \mbox{$(x, t)\in [\bs{B}]$ holds for some $t\in \mathbb{M}$}\}, \\
[\bs{B}]_{\rm T}&=\{t\in \mathbb{M} : \mbox{$(x, t)\in [\bs{B}]$ holds for some $x\in \BbbZ^d$}\},
\end{align}
where $\mathbb{M}=\{1, 2, \dots, M\}$.

The cross-sections of $[\bs{B}]$ in the space-time lattice are defined as follows:
\begin{align}
\mathsf{S}_{x_0}(\bs{B}) = \{(x_0, t)\in \mathbb{L} : (x_0, t) \in \bs{B}\}, \quad
\mathsf{T}_{t_0}(\bs{B}) = \{(x, t_0)\in \mathbb{L} : (x, t_0) \in \bs{B}\}.
\end{align}
Therefore, the projections of the cross-sections of $[\bs{B}]$ are given by:
\begin{align}
[\mathsf{S}_{x_0}(\bs{B})]_{\rm T} = \{t \in \mathbb{M} : (x_0, t) \in\mathsf{S}_{x_0}(\bs{B})\}, \quad
[\mathsf{T}_{t_0}(\bs{B})]_{\rm S} = \{x \in \BbbZ^d : (x, t_0) \in \mathsf{T}_{t_0}(\bs{B})\}.
\end{align}
For the reader's clarity, the symbols introduced so far are visually explained in Figure \ref{Grph1}.

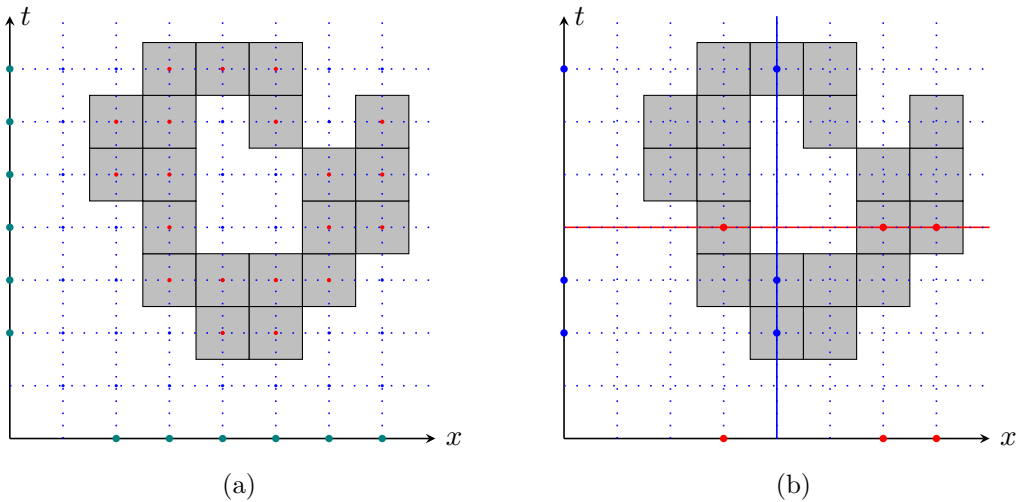
\begin{figure}[h]
  \begin{minipage}[b]{0.45\linewidth}
    \centering
     \begin{tikzpicture}[scale=0.7]
\coordinate(Z11)at(\Ct+1,\Ct-1);
\coordinate(Z12)at(\Ct+1,\Ct);
\coordinate(Z13)at(\Ct+2,\Ct);
\coordinate(Z14)at(\Ct+2,\Ct-1);
\filldraw[fill=lightgray](Z11)--(Z12)--(Z13)--(Z14)--cycle;
\coordinate(Z21)at(\Ct+2,\Ct-1);
\coordinate(Z22)at(\Ct+2,\Ct);
\coordinate(Z23)at(\Ct+3,\Ct);
\coordinate(Z24)at(\Ct+3,\Ct-1);
\filldraw[fill=lightgray](Z21)--(Z22)--(Z23)--(Z24)--cycle;

\coordinate(A11)at(\Ct,\Ct);
\coordinate(A12)at(\Ct,\Ct+1);
\coordinate(A13)at(\Ct+1,\Ct+1);
\coordinate(A14)at(\Ct+1,\Ct);
\filldraw[fill=lightgray](A11)--(A12)--(A13)--(A14)--cycle;
\coordinate(A21)at(\Ct+1,\Ct);
\coordinate(A22)at(\Ct+1,\Ct+1);
\coordinate(A23)at(\Ct+2,\Ct+1);
\coordinate(A24)at(\Ct+2,\Ct);
\filldraw[fill=lightgray](A21)--(A22)--(A23)--(A24)--cycle;
\coordinate(A31)at(\Ct+2,\Ct);
\coordinate(A32)at(\Ct+2,\Ct+1);
\coordinate(A33)at(\Ct+3,\Ct+1);
\coordinate(A34)at(\Ct+3,\Ct);
\filldraw[fill=lightgray](A31)--(A32)--(A33)--(A34)--cycle;
\coordinate(A41)at(\Ct+2+1,\Ct);
\coordinate(A42)at(\Ct+2+1,\Ct+1);
\coordinate(A43)at(\Ct+3+1,\Ct+1);
\coordinate(A44)at(\Ct+3+1,\Ct);
\filldraw[fill=lightgray](A41)--(A42)--(A43)--(A44)--cycle;

\coordinate(B11)at(\Ct,\Ct+1);
\coordinate(B12)at(\Ct,\Ct+2);
\coordinate(B13)at(\Ct+1,\Ct+2);
\coordinate(B14)at(\Ct+1,\Ct+1);
\filldraw[fill=lightgray](B11)--(B12)--(B13)--(B14)--cycle;
\coordinate(B31)at(\Ct+4,\Ct+1);
\coordinate(B32)at(\Ct+4,\Ct+2);
\coordinate(B33)at(\Ct+5,\Ct+2);
\coordinate(B34)at(\Ct+5,\Ct+1);
\filldraw[fill=lightgray](B31)--(B32)--(B33)--(B34)--cycle;
\coordinate(B41)at(\Ct+3,\Ct+1);
\coordinate(B42)at(\Ct+3,\Ct+2);
\coordinate(B43)at(\Ct+4,\Ct+2);
\coordinate(B44)at(\Ct+4,\Ct+1);
\filldraw[fill=lightgray](B41)--(B42)--(B43)--(B44)--cycle;

\coordinate(C11)at(\Ct-1,\Ct+2);
\coordinate(C12)at(\Ct-1,\Ct+3);
\coordinate(C13)at(\Ct,\Ct+3);
\coordinate(C14)at(\Ct,\Ct+2);
\filldraw[fill=lightgray](C11)--(C12)--(C13)--(C14)--cycle;
\coordinate(C21)at(\Ct,\Ct+2);
\coordinate(C22)at(\Ct,\Ct+3);
\coordinate(C23)at(\Ct+1,\Ct+3);
\coordinate(C24)at(\Ct+1,\Ct+2);
\filldraw[fill=lightgray](C21)--(C22)--(C23)--(C24)--cycle;
\coordinate(C31)at(\Ct+3,\Ct+2);
\coordinate(C32)at(\Ct+3,\Ct+3);
\coordinate(C33)at(\Ct+4,\Ct+3);
\coordinate(C34)at(\Ct+4,\Ct+2);
\filldraw[fill=lightgray](C31)--(C32)--(C33)--(C34)--cycle;
\coordinate(C41)at(\Ct+3+1,\Ct+2);
\coordinate(C42)at(\Ct+3+1,\Ct+3);
\coordinate(C43)at(\Ct+4+1,\Ct+3);
\coordinate(C44)at(\Ct+4+1,\Ct+2);
\filldraw[fill=lightgray](C41)--(C42)--(C43)--(C44)--cycle;

\coordinate(D11)at(\Ct-1,\Ct+3);
\coordinate(D12)at(\Ct-1,\Ct+4);
\coordinate(D13)at(\Ct,\Ct+4);
\coordinate(D14)at(\Ct,\Ct+3);
\filldraw[fill=lightgray](D11)--(D12)--(D13)--(D14)--cycle;
\coordinate(D21)at(\Ct+4-2,\Ct+3);
\coordinate(D22)at(\Ct+4-2,\Ct+4);
\coordinate(D23)at(\Ct+5-2,\Ct+4);
\coordinate(D24)at(\Ct+5-2,\Ct+3);
\filldraw[fill=lightgray](D21)--(D22)--(D23)--(D24)--cycle;
\coordinate(D31)at(\Ct+4,\Ct+3);
\coordinate(D32)at(\Ct+4,\Ct+4);
\coordinate(D33)at(\Ct+5,\Ct+4);
\coordinate(D34)at(\Ct+5,\Ct+3);
\filldraw[fill=lightgray](D31)--(D32)--(D33)--(D34)--cycle;
\coordinate(D41)at(\Ct,\Ct+3);
\coordinate(D42)at(\Ct,\Ct+4);
\coordinate(D43)at(\Ct+1,\Ct+4);
\coordinate(D44)at(\Ct+1,\Ct+3);
\filldraw[fill=lightgray](D41)--(D42)--(D43)--(D44)--cycle;

\coordinate(E11)at(\Ct,\Ct+3+1);
\coordinate(E12)at(\Ct,\Ct+4+1);
\coordinate(E13)at(\Ct+1,\Ct+4+1);
\coordinate(E14)at(\Ct+1,\Ct+3+1);
\filldraw[fill=lightgray](E11)--(E12)--(E13)--(E14)--cycle;
\coordinate(E21)at(\Ct+1,\Ct+3+1);
\coordinate(E22)at(\Ct+1,\Ct+4+1);
\coordinate(E23)at(\Ct+2,\Ct+4+1);
\coordinate(E24)at(\Ct+2,\Ct+3+1);
\filldraw[fill=lightgray](E21)--(E22)--(E23)--(E24)--cycle;
\coordinate(E41)at(\Ct+2,\Ct+3+1);
\coordinate(E42)at(\Ct+2,\Ct+4+1);
\coordinate(E43)at(\Ct+3,\Ct+4+1);
\coordinate(E44)at(\Ct+3,\Ct+3+1);
\filldraw[fill=lightgray](E41)--(E42)--(E43)--(E44)--cycle;

\draw[->,>=stealth,semithick] (0,0)--(\Ti,0)node[right]{$x$}; 
\draw[->,>=stealth,semithick] (0,0)--(0,\Sp)node[right]{$t$}; 
\draw[name path=X1,loosely dotted,  semithick, blue] (1,0) -- (1,\Ti);
\draw[name path=X2,loosely dotted, semithick, blue] (2,0) -- (2,\Ti);
\draw[name path=X3,loosely dotted, semithick, blue] (3,0) -- (3,\Ti);
\draw[name path=X4,loosely dotted, semithick, blue] (5,0) -- (5,\Ti);
\draw[name path=X5,loosely dotted, semithick, blue] (4,0) -- (4,\Ti);
\draw[name path=X6,loosely dotted, semithick, blue] (6,0) -- (6,\Ti);
\draw[name path=X7,loosely dotted, semithick, blue] (7,0) -- (7,\Ti);
\draw[name path=Y1,loosely dotted, semithick, blue] (0,1) -- (\Sp,1);
\draw[name path=Y2, loosely dotted, semithick, blue] (0,2) -- (\Sp,2);
\draw[name path=Y3, loosely dotted, semithick, blue] (0,3) -- (\Sp,3);
\draw[name path=Y4,loosely dotted, semithick, blue] (0,4) -- (\Sp,4);
\draw[name path=Y5,loosely dotted, semithick, blue] (0,5) -- (\Sp,5);
\draw[name path=Y6,loosely dotted, semithick, blue] (0,6) -- (\Sp,6);
\draw[name path=Y7,loosely dotted, semithick, blue] (0,7) -- (\Sp,7);

\foreach\Q in { 1,2,3, 4, 5, 6, 7}\foreach\P in { 1,2,3, 4, 5, 6, 7}\fill[blue](\P,\Q)circle(0.03);
\foreach\Q in { 5, 6}\fill[red](2,\Q)circle(0.04);
\foreach\Q in { 3,4, 5, 6, 7}\fill[red](3,\Q)circle(0.04);
\foreach\Q in { 2,3, 7}\fill[red](4,\Q)circle(0.04);
\foreach\Q in { 2,3, 6, 7}\fill[red](5,\Q)circle(0.04);
\foreach\Q in { 3, 4, 5}\fill[red](6,\Q)circle(0.04);
\foreach\Q in { 4, 5, 6}\fill[red](7,\Q)circle(0.04);

\foreach\P in { 2,3, 4, 5, 6, 7}\fill[teal](\P,0)circle(0.07);
\foreach\Q in { 2,3, 4, 5, 6, 7}\fill[teal](0,\Q)circle(0.07);

\end{tikzpicture}
    \subcaption{}\label{a}
  \end{minipage}
  \begin{minipage}[b]{0.45\linewidth}
    \centering
      \begin{tikzpicture}[scale=0.7]
\coordinate(Z11)at(\Ct+1,\Ct-1);
\coordinate(Z12)at(\Ct+1,\Ct);
\coordinate(Z13)at(\Ct+2,\Ct);
\coordinate(Z14)at(\Ct+2,\Ct-1);
\filldraw[fill=lightgray](Z11)--(Z12)--(Z13)--(Z14)--cycle;
\coordinate(Z21)at(\Ct+2,\Ct-1);
\coordinate(Z22)at(\Ct+2,\Ct);
\coordinate(Z23)at(\Ct+3,\Ct);
\coordinate(Z24)at(\Ct+3,\Ct-1);
\filldraw[fill=lightgray](Z21)--(Z22)--(Z23)--(Z24)--cycle;

\coordinate(A11)at(\Ct,\Ct);
\coordinate(A12)at(\Ct,\Ct+1);
\coordinate(A13)at(\Ct+1,\Ct+1);
\coordinate(A14)at(\Ct+1,\Ct);
\filldraw[fill=lightgray](A11)--(A12)--(A13)--(A14)--cycle;
\coordinate(A21)at(\Ct+1,\Ct);
\coordinate(A22)at(\Ct+1,\Ct+1);
\coordinate(A23)at(\Ct+2,\Ct+1);
\coordinate(A24)at(\Ct+2,\Ct);
\filldraw[fill=lightgray](A21)--(A22)--(A23)--(A24)--cycle;
\coordinate(A31)at(\Ct+2,\Ct);
\coordinate(A32)at(\Ct+2,\Ct+1);
\coordinate(A33)at(\Ct+3,\Ct+1);
\coordinate(A34)at(\Ct+3,\Ct);
\filldraw[fill=lightgray](A31)--(A32)--(A33)--(A34)--cycle;
\coordinate(A41)at(\Ct+2+1,\Ct);
\coordinate(A42)at(\Ct+2+1,\Ct+1);
\coordinate(A43)at(\Ct+3+1,\Ct+1);
\coordinate(A44)at(\Ct+3+1,\Ct);
\filldraw[fill=lightgray](A41)--(A42)--(A43)--(A44)--cycle;

\coordinate(B11)at(\Ct,\Ct+1);
\coordinate(B12)at(\Ct,\Ct+2);
\coordinate(B13)at(\Ct+1,\Ct+2);
\coordinate(B14)at(\Ct+1,\Ct+1);
\filldraw[fill=lightgray](B11)--(B12)--(B13)--(B14)--cycle;
\coordinate(B31)at(\Ct+4,\Ct+1);
\coordinate(B32)at(\Ct+4,\Ct+2);
\coordinate(B33)at(\Ct+5,\Ct+2);
\coordinate(B34)at(\Ct+5,\Ct+1);
\filldraw[fill=lightgray](B31)--(B32)--(B33)--(B34)--cycle;
\coordinate(B41)at(\Ct+3,\Ct+1);
\coordinate(B42)at(\Ct+3,\Ct+2);
\coordinate(B43)at(\Ct+4,\Ct+2);
\coordinate(B44)at(\Ct+4,\Ct+1);
\filldraw[fill=lightgray](B41)--(B42)--(B43)--(B44)--cycle;

\coordinate(C11)at(\Ct-1,\Ct+2);
\coordinate(C12)at(\Ct-1,\Ct+3);
\coordinate(C13)at(\Ct,\Ct+3);
\coordinate(C14)at(\Ct,\Ct+2);
\filldraw[fill=lightgray](C11)--(C12)--(C13)--(C14)--cycle;
\coordinate(C21)at(\Ct,\Ct+2);
\coordinate(C22)at(\Ct,\Ct+3);
\coordinate(C23)at(\Ct+1,\Ct+3);
\coordinate(C24)at(\Ct+1,\Ct+2);
\filldraw[fill=lightgray](C21)--(C22)--(C23)--(C24)--cycle;
\coordinate(C31)at(\Ct+3,\Ct+2);
\coordinate(C32)at(\Ct+3,\Ct+3);
\coordinate(C33)at(\Ct+4,\Ct+3);
\coordinate(C34)at(\Ct+4,\Ct+2);
\filldraw[fill=lightgray](C31)--(C32)--(C33)--(C34)--cycle;
\coordinate(C41)at(\Ct+3+1,\Ct+2);
\coordinate(C42)at(\Ct+3+1,\Ct+3);
\coordinate(C43)at(\Ct+4+1,\Ct+3);
\coordinate(C44)at(\Ct+4+1,\Ct+2);
\filldraw[fill=lightgray](C41)--(C42)--(C43)--(C44)--cycle;

\coordinate(D11)at(\Ct-1,\Ct+3);
\coordinate(D12)at(\Ct-1,\Ct+4);
\coordinate(D13)at(\Ct,\Ct+4);
\coordinate(D14)at(\Ct,\Ct+3);
\filldraw[fill=lightgray](D11)--(D12)--(D13)--(D14)--cycle;
\coordinate(D21)at(\Ct+4-2,\Ct+3);
\coordinate(D22)at(\Ct+4-2,\Ct+4);
\coordinate(D23)at(\Ct+5-2,\Ct+4);
\coordinate(D24)at(\Ct+5-2,\Ct+3);
\filldraw[fill=lightgray](D21)--(D22)--(D23)--(D24)--cycle;
\coordinate(D31)at(\Ct+4,\Ct+3);
\coordinate(D32)at(\Ct+4,\Ct+4);
\coordinate(D33)at(\Ct+5,\Ct+4);
\coordinate(D34)at(\Ct+5,\Ct+3);
\filldraw[fill=lightgray](D31)--(D32)--(D33)--(D34)--cycle;
\coordinate(D41)at(\Ct,\Ct+3);
\coordinate(D42)at(\Ct,\Ct+4);
\coordinate(D43)at(\Ct+1,\Ct+4);
\coordinate(D44)at(\Ct+1,\Ct+3);
\filldraw[fill=lightgray](D41)--(D42)--(D43)--(D44)--cycle;

\coordinate(E11)at(\Ct,\Ct+3+1);
\coordinate(E12)at(\Ct,\Ct+4+1);
\coordinate(E13)at(\Ct+1,\Ct+4+1);
\coordinate(E14)at(\Ct+1,\Ct+3+1);
\filldraw[fill=lightgray](E11)--(E12)--(E13)--(E14)--cycle;
\coordinate(E21)at(\Ct+1,\Ct+3+1);
\coordinate(E22)at(\Ct+1,\Ct+4+1);
\coordinate(E23)at(\Ct+2,\Ct+4+1);
\coordinate(E24)at(\Ct+2,\Ct+3+1);
\filldraw[fill=lightgray](E21)--(E22)--(E23)--(E24)--cycle;
\coordinate(E41)at(\Ct+2,\Ct+3+1);
\coordinate(E42)at(\Ct+2,\Ct+4+1);
\coordinate(E43)at(\Ct+3,\Ct+4+1);
\coordinate(E44)at(\Ct+3,\Ct+3+1);
\filldraw[fill=lightgray](E41)--(E42)--(E43)--(E44)--cycle;

\draw[->,>=stealth,semithick] (0,0)--(\Ti,0)node[right]{$x$}; 
\draw[->,>=stealth,semithick] (0,0)--(0,\Sp)node[right]{$t$}; 
\draw[semithick, blue] (4,0) -- (4,\Ti);
\draw[semithick, red] (0, 4) -- (\Sp,4);

\draw[name path=X1,loosely dotted,  semithick, blue] (1,0) -- (1,\Ti);
\draw[name path=X2,loosely dotted, semithick, blue] (2,0) -- (2,\Ti);
\draw[name path=X3,loosely dotted, semithick, blue] (3,0) -- (3,\Ti);
\draw[name path=X4,loosely dotted, semithick, blue] (5,0) -- (5,\Ti);
\draw[name path=X5,loosely dotted, semithick, blue] (4,0) -- (4,\Ti);
\draw[name path=X6,loosely dotted, semithick, blue] (6,0) -- (6,\Ti);
\draw[name path=X7,loosely dotted, semithick, blue] (7,0) -- (7,\Ti);
\draw[name path=Y1,loosely dotted, semithick, blue] (0,1) -- (\Sp,1);
\draw[name path=Y2, loosely dotted, semithick, blue] (0,2) -- (\Sp,2);
\draw[name path=Y3, loosely dotted, semithick, blue] (0,3) -- (\Sp,3);
\draw[name path=Y4,loosely dotted, semithick, blue] (0,4) -- (\Sp,4);
\draw[name path=Y5,loosely dotted, semithick, blue] (0,5) -- (\Sp,5);
\draw[name path=Y6,loosely dotted, semithick, blue] (0,6) -- (\Sp,6);
\draw[name path=Y7,loosely dotted, semithick, blue] (0,7) -- (\Sp,7);

\foreach\Q in { 2,3, 7}\fill[blue](4,\Q)circle(0.07);

\foreach\P in { 3, 6, 7}\fill[red](\P,4)circle(0.07);

\foreach\P in { 3, 6, 7}\fill[red](\P,0)circle(0.07);
\foreach\Q in { 2, 3,  7}\fill[blue](0,\Q)circle(0.07);

\end{tikzpicture}
    \subcaption{}\label{b}
  \end{minipage}
  \caption{(a) Example of $\bs B$ in $1+1$ dimensional space-time. Each gray box represents $C(x, t)$. Blue dots denote space-time lattice points, while red dots represent $[{\bs B}]$. Green dots along the $x$-axis (resp. $t$-axis) correspond to $[{\bs B}]_{\rm S}$ (resp. $[{\bs B}]_{\rm T}$).
  (b) 
The blue dots along the blue line represent the cross-section of $\bs B$ with the plane $x=4$, and the blue dots on the $t$-axis denote the projection $[\mathsf{S}_4({\bs B})]_{\rm T}$ onto the time axis of this cross-section. Similarly, the red dots along the red line represent the cross-section with the plane $t=4$, and the red dots on the $x$-axis denote the projection $[\mathsf{T}_4({\bs B})]_{\rm S}$ onto the $x$-axis of this cross-section.
  }\label{Grph1}
\end{figure}

Suppose that we are given  $\bs{\vSi}_{\mathbb{L}_{\vLa}}=(\vSi_1, \dots, \vSi_M)$,  where  $\vSi_t=(D_q^{(t)}, \bn_{\vLa}^{(t)})$.\footnote{
Using the symbols in \eqref{TSigma} and \eqref{ZSi}, one should write $\vSi_t=(B^{(t)}, \bn_{\vLa}^{(t)})$. However, for the convenience of the following description, the symbol $D_q^{(t)}$ will be used instead of $B^{(t)}$.
} Then, we define $\sigma{(x, t)}$ as follows:
\begin{align}
\sigma_{(x, t)}=
\begin{cases}
g^{(\ell)}_x & (x\notin \vLa)\\
n_{\vLa}^{(t)}(x) & (x\in \vLa).
\end{cases}
\end{align}
Here, we write  $\bn_{\vLa}^{(t)}$ as  $\bn_{\vLa}^{(t)}=\left( n_{\vLa}^{(t)}(x) \right)_{ x\in \vLa}$.

We say that an  elementary cube $C(x, t)$ is {\it in the $\ell$-ground state} if $\sigma_{(y, t)}$ coincides with  the  ground state configuration $\bs{g}^{(\ell)}$ for all $y\in U(x)$, where $U(x)$ is defined by \eqref{DefU}. 
Irrespective of the value of  $\ell$, an elementary cube being in the $\ell$-ground state is commonly referred to as an elementary cube in the ground state.  Conversely, if $C(x, t)$ is not in the ground state, we say that it is {\it in an excited state}  or simply an {\it excited cube}.

An excited cube $C(x,t)$  is said to be {\it in a quantum excitation state} if and only if $x$ belongs to $D_q^{(t)}$. Conversely, if an excited  cube $C(x,t)$  is not in a quantum excitation state, it is referred to as being {\it in a classical excitation state}.

We can  can decompose  $\mathbb{T}_{\vLa}$ into a union of two disjoint sets: the set $\mathbb{D}$ consisting of excited  cubes, and the set $\mathbb{G}$ consisting of elementary  cubes in the ground state. That is, $\mathbb{T}_{\vLa}=\mathbb{D}\sqcup \mathbb{G}$.

 Let $\mathcal{C}(\mathbb{D})$ denote the set of connected components of $\mathbb{D}$, so that we can write $\mathbb{D}=\bigsqcup_{\bs{D}\in \mathcal{C}(\mathbb{D})} \bs{D}$. 
 Each $\bs{D}\in \mathcal{C}(\mathbb{D})$ can be uniquely determined by  $(D^{(1)}, \dots, D^{(M)})$, where $D^{(t)}$ is defined by $D^{(t)}=[\mathsf{T}_t(\bs{D})]_{\rm S}$.
  Therefore, one can identify $\bs{D}$ with its time slice $\prod_{t\in \mathbb{M}}D^{(t)}$, which we denote as $\bs{D}\Eql \prod_{t\in \mathbb{M}}D^{(t)}$. For each $t\in \mathbb{M}$, we define $D_c^{(t)}=D^{(t)} \setminus D_q^{(t)}$. Let $\bs{D}_q\Eql\prod_{t\in \mathbb{M}}D_q^{(t)}$ and $\bs{D}_c\Eql\prod_{t\in \mathbb{M}}D_c^{(t)}$.
 $\bs{D}_q$ is the set consisting of elementary cubes in quantum excitation state; $\bs{D}_c$ is the set  consisting of elementary cubes in  classical excitation state.
  It is clear that $\bs{D}=\bs{D}_q\sqcup \bs{D}_c$. See Figure \ref{Sec3D}.
\begin{figure}[h]
 \begin{center}
  \begin{tikzpicture}[scale=0.7]
\coordinate(Z11)at(\Ct+1,\Ct-1);
\coordinate(Z12)at(\Ct+1,\Ct);
\coordinate(Z13)at(\Ct+2,\Ct);
\coordinate(Z14)at(\Ct+2,\Ct-1);
\filldraw[fill=gray](Z11)--(Z12)--(Z13)--(Z14)--cycle;
\coordinate(Z21)at(\Ct+2,\Ct-1);
\coordinate(Z22)at(\Ct+2,\Ct);
\coordinate(Z23)at(\Ct+3,\Ct);
\coordinate(Z24)at(\Ct+3,\Ct-1);
\filldraw[fill=gray](Z21)--(Z22)--(Z23)--(Z24)--cycle;

\coordinate(A11)at(\Ct,\Ct);
\coordinate(A12)at(\Ct,\Ct+1);
\coordinate(A13)at(\Ct+1,\Ct+1);
\coordinate(A14)at(\Ct+1,\Ct);
\filldraw[fill=gray](A11)--(A12)--(A13)--(A14)--cycle;
\coordinate(A21)at(\Ct+1,\Ct);
\coordinate(A22)at(\Ct+1,\Ct+1);
\coordinate(A23)at(\Ct+2,\Ct+1);
\coordinate(A24)at(\Ct+2,\Ct);
\filldraw[fill=gray](A21)--(A22)--(A23)--(A24)--cycle;
\coordinate(A31)at(\Ct+2,\Ct);
\coordinate(A32)at(\Ct+2,\Ct+1);
\coordinate(A33)at(\Ct+3,\Ct+1);
\coordinate(A34)at(\Ct+3,\Ct);
\filldraw[fill=gray](A31)--(A32)--(A33)--(A34)--cycle;
\coordinate(A41)at(\Ct+2+1,\Ct);
\coordinate(A42)at(\Ct+2+1,\Ct+1);
\coordinate(A43)at(\Ct+3+1,\Ct+1);
\coordinate(A44)at(\Ct+3+1,\Ct);
\filldraw[fill=lightgray](A41)--(A42)--(A43)--(A44)--cycle;

\coordinate(B11)at(\Ct,\Ct+1);
\coordinate(B12)at(\Ct,\Ct+2);
\coordinate(B13)at(\Ct+1,\Ct+2);
\coordinate(B14)at(\Ct+1,\Ct+1);
\filldraw[fill=gray](B11)--(B12)--(B13)--(B14)--cycle;
\coordinate(B21)at(\Ct+1,\Ct+1);
\coordinate(B22)at(\Ct+1,\Ct+2);
\coordinate(B23)at(\Ct+2,\Ct+2);
\coordinate(B24)at(\Ct+2,\Ct+1);
\filldraw[fill=gray](B21)--(B22)--(B23)--(B24)--cycle;
\coordinate(B31)at(\Ct+4,\Ct+1);
\coordinate(B32)at(\Ct+4,\Ct+2);
\coordinate(B33)at(\Ct+5,\Ct+2);
\coordinate(B34)at(\Ct+5,\Ct+1);
\filldraw[fill=lightgray](B31)--(B32)--(B33)--(B34)--cycle;
\coordinate(B41)at(\Ct+3,\Ct+1);
\coordinate(B42)at(\Ct+3,\Ct+2);
\coordinate(B43)at(\Ct+4,\Ct+2);
\coordinate(B44)at(\Ct+4,\Ct+1);
\filldraw[fill=lightgray](B41)--(B42)--(B43)--(B44)--cycle;
\coordinate(B51)at(\Ct+2,\Ct+1);
\coordinate(B52)at(\Ct+2,\Ct+2);
\coordinate(B53)at(\Ct+3,\Ct+2);
\coordinate(B54)at(\Ct+3,\Ct+1);
\filldraw[fill=gray](B51)--(B52)--(B53)--(B54)--cycle;

\coordinate(C11)at(\Ct-1,\Ct+2);
\coordinate(C12)at(\Ct-1,\Ct+3);
\coordinate(C13)at(\Ct,\Ct+3);
\coordinate(C14)at(\Ct,\Ct+2);
\filldraw[fill=gray](C11)--(C12)--(C13)--(C14)--cycle;
\coordinate(C21)at(\Ct,\Ct+2);
\coordinate(C22)at(\Ct,\Ct+3);
\coordinate(C23)at(\Ct+1,\Ct+3);
\coordinate(C24)at(\Ct+1,\Ct+2);
\filldraw[fill=gray](C21)--(C22)--(C23)--(C24)--cycle;
\coordinate(C31)at(\Ct+3,\Ct+2);
\coordinate(C32)at(\Ct+3,\Ct+3);
\coordinate(C33)at(\Ct+4,\Ct+3);
\coordinate(C34)at(\Ct+4,\Ct+2);
\filldraw[fill=lightgray](C31)--(C32)--(C33)--(C34)--cycle;
\coordinate(C41)at(\Ct+3+1,\Ct+2);
\coordinate(C42)at(\Ct+3+1,\Ct+3);
\coordinate(C43)at(\Ct+4+1,\Ct+3);
\coordinate(C44)at(\Ct+4+1,\Ct+2);
\filldraw[fill=lightgray](C41)--(C42)--(C43)--(C44)--cycle;
\coordinate(C51)at(\Ct+1,\Ct+2);
\coordinate(C52)at(\Ct+1,\Ct+3);
\coordinate(C53)at(\Ct+2,\Ct+3);
\coordinate(C54)at(\Ct+2,\Ct+2);
\filldraw[fill=gray](C51)--(C52)--(C53)--(C54)--cycle;
\coordinate(C61)at(\Ct+2,\Ct+2);
\coordinate(C62)at(\Ct+2,\Ct+3);
\coordinate(C63)at(\Ct+3,\Ct+3);
\coordinate(C64)at(\Ct+3,\Ct+2);
\filldraw[fill=lightgray](C61)--(C62)--(C63)--(C64)--cycle;

\coordinate(D11)at(\Ct-1,\Ct+3);
\coordinate(D12)at(\Ct-1,\Ct+4);
\coordinate(D13)at(\Ct,\Ct+4);
\coordinate(D14)at(\Ct,\Ct+3);
\filldraw[fill=gray](D11)--(D12)--(D13)--(D14)--cycle;
\coordinate(D21)at(\Ct+4-2,\Ct+3);
\coordinate(D22)at(\Ct+4-2,\Ct+4);
\coordinate(D23)at(\Ct+5-2,\Ct+4);
\coordinate(D24)at(\Ct+5-2,\Ct+3);
\filldraw[fill=lightgray](D21)--(D22)--(D23)--(D24)--cycle;

\coordinate(D31)at(\Ct+1,\Ct+3);
\coordinate(D32)at(\Ct+1,\Ct+4);
\coordinate(D33)at(\Ct+2,\Ct+4);
\coordinate(D34)at(\Ct+2,\Ct+3);
\filldraw[fill=gray](D31)--(D32)--(D33)--(D34)--cycle;
\coordinate(D41)at(\Ct,\Ct+3);
\coordinate(D42)at(\Ct,\Ct+4);
\coordinate(D43)at(\Ct+1,\Ct+4);
\coordinate(D44)at(\Ct+1,\Ct+3);
\filldraw[fill=gray](D41)--(D42)--(D43)--(D44)--cycle;

\coordinate(D51)at(\Ct+4-1,\Ct+3);
\coordinate(D52)at(\Ct+4-1,\Ct+4);
\coordinate(D53)at(\Ct+5-1,\Ct+4);
\coordinate(D54)at(\Ct+5-1,\Ct+3);
\filldraw[fill=lightgray](D51)--(D52)--(D53)--(D54)--cycle;

\coordinate(E11)at(\Ct,\Ct+3+1);
\coordinate(E12)at(\Ct,\Ct+4+1);
\coordinate(E13)at(\Ct+1,\Ct+4+1);
\coordinate(E14)at(\Ct+1,\Ct+3+1);
\filldraw[fill=gray](E11)--(E12)--(E13)--(E14)--cycle;
\coordinate(E21)at(\Ct+1,\Ct+3+1);
\coordinate(E22)at(\Ct+1,\Ct+4+1);
\coordinate(E23)at(\Ct+2,\Ct+4+1);
\coordinate(E24)at(\Ct+2,\Ct+3+1);
\filldraw[fill=gray](E21)--(E22)--(E23)--(E24)--cycle;
\coordinate(E41)at(\Ct+2,\Ct+3+1);
\coordinate(E42)at(\Ct+2,\Ct+4+1);
\coordinate(E43)at(\Ct+3,\Ct+4+1);
\coordinate(E44)at(\Ct+3,\Ct+3+1);
\filldraw[fill=lightgray](E41)--(E42)--(E43)--(E44)--cycle;



\draw[-, semithick, blue] (3.5,2)--(5.5,2); 
\draw[-, semithick, blue] (2.5,3)--(6.5,3); 
\draw[-, semithick, blue] (2.5,4)--(7.5,4);
\draw[-, semithick, blue] (1.5,5)--(7.5,5);
\draw[-, semithick, blue] (1.5,6)--(6.5,6);
\draw[-, semithick, blue] (2.5,7)--(5.5,7);

\draw[->,>=stealth,semithick] (0,0)--(\Ti,0)node[right]{$x$}; 
\draw[->,>=stealth,semithick] (0,0)--(0,\Sp)node[right]{$t$}; 
\draw[name path=X1,loosely dotted,  semithick, blue] (1,0) -- (1,\Ti);
\draw[name path=X2,loosely dotted, semithick, blue] (2,0) -- (2,\Ti);
\draw[name path=X3,loosely dotted, semithick, blue] (3,0) -- (3,\Ti);
\draw[name path=X4,loosely dotted, semithick, blue] (5,0) -- (5,\Ti);
\draw[name path=X5,loosely dotted, semithick, blue] (4,0) -- (4,\Ti);
\draw[name path=X6,loosely dotted, semithick, blue] (6,0) -- (6,\Ti);
\draw[name path=X7,loosely dotted, semithick, blue] (7,0) -- (7,\Ti);
\draw[name path=Y1,loosely dotted, semithick, blue] (0,1) -- (\Sp,1);
\draw[name path=Y2, loosely dotted, semithick, blue] (0,2) -- (\Sp,2);
\draw[name path=Y3, loosely dotted, semithick, blue] (0,3) -- (\Sp,3);
\draw[name path=Y4,loosely dotted, semithick, blue] (0,4) -- (\Sp,4);
\draw[name path=Y5,loosely dotted, semithick, blue] (0,5) -- (\Sp,5);
\draw[name path=Y6,loosely dotted, semithick, blue] (0,6) -- (\Sp,6);
\draw[name path=Y7,loosely dotted, semithick, blue] (0,7) -- (\Sp,7);

\foreach\Q in { 1,2,3, 4, 5, 6, 7}\foreach\P in { 1,2,3, 4, 5, 6, 7}\fill[blue](\P,\Q)circle(0.03);
\foreach\Q in { 5, 6}\fill[red](2,\Q)circle(0.04);
\foreach\Q in { 3,4, 5, 6, 7}\fill[red](3,\Q)circle(0.04);
\foreach\Q in { 2,3, 4,5,6, 7}\fill[red](4,\Q)circle(0.04);
\foreach\Q in { 2,3, 4}\fill[red](5,\Q)circle(0.04);
\foreach\Q in {  5, 6, 7}\fill[green](5,\Q)circle(0.04);
\foreach\Q in { 3, 4, 5, 6}\fill[green](6,\Q)circle(0.04);
\foreach\Q in { 4, 5}\fill[green](7,\Q)circle(0.04);


\end{tikzpicture}
  \caption{The colored boxes represent $\bs D$. The light gray region represents $\bs D_{c}$, and the dark gray represents $\bs D_q$. The red and green lattice points correspond to lattice points $[\bs D_{q}]$ and $[\bs D_c]$, respectively. The lattice points on the cross-section of $\bs D$ by the horizontal line represent $D^{(t)}$.
  }\label{Sec3D}
        \end{center}
 \end{figure}
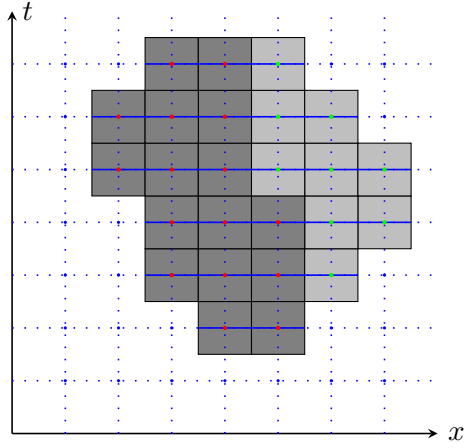

For each  value of $\ell$,  we define the set  consisting of elementary cubes in the $\ell$-ground state  $\bs{g}^{(\ell)}$ as $V_{\ell}$. Thus, we have $\mathbb{G}=\sqcup_{\ell} V_{\ell}$.  Considering each connected component $\bs{D}$ of excited cubes, it can be observed that every boundary face $F$ of $\bs{D}$ is in contact with cubes in a specific $\ell^{\prime}$-ground state, which is represented by $\alpha_{\bs{D}}(F)=\ell^{\prime}$. The function $\alpha_{\bs{D}}$ maps the set of boundary faces of $\bs{D}$ into $\{1, \dots, r\}$, and is termed as the {\it label} of the boundary of $\bs{D}$.

A {\it contour} $Y=(\bs{D}, \alpha)$ is defined as a pair comprising of the connected component of excited cubes $\bs{D}\in \mathcal{C}(\mathbb{D})$, and the label of the boundary of $\bs{D}$, which is denoted as $\alpha(\cdot)=\alpha_{\bs{D}}(\cdot)$. The connected component $\bs{D}$ of the contour $Y$ is denoted as $\supp Y$. 
The interior $\inte Y$ of a contour $Y$ is defined as the union of finite connected components of $\mathbb{T}\setminus \supp Y$.
We say that $Y$ is an {\it $\ell$-contour} if the boundary of $V(Y)=\supp Y\cup \inte Y$ contacts cubes in the $\ell$- ground state, i.e., $\alpha_Y(F_{ V(Y)})=\ell$ is satisfied, where $F_{V(Y)}$ represents the boundary face of $V(Y)$.

Two contours $Y_1$ and $Y_2$ are said to be {\it compatible} if $\supp Y_1 \cap \supp Y_2=\varnothing$.
Let $\mathbb{Y}=\{Y_1, \dots, Y_n\}$ be a set of contours that are mutually compatible. 
We define $Y_i$ as an {\it external contour}  if $\supp Y_i\cap V(Y_j)=\varnothing$ holds for all $j\neq i$.
We say that $\mathbb{Y}$ is {\it  matching}  if the label $\alpha$ is constant on each boundary face of every  connected component of $\mathbb{T}_{\vLa}\setminus (\supp Y_1\cup  \dots \cup \supp Y_n)$.
So far, various symbols related to contours have been introduced. By revisiting the definitions while examining Figure \ref{Grph3} and its caption, readers will likely gain a clear understanding of their meanings.\footnote{
In this paper, considering contour models in space-time with dimensions $d+1$ where $d$ is $2$ or above, it is not straightforward to draw contour diagrams easily. However, if one has a good understanding of the case in $1+1$ dimensions, subsequent discussions should be easily comprehensible.}

For each $\bs{\vSi}_{\mathbb{L}_{\vLa}}$, there corresponds a set of mutually compatible and matching contours $\mathbb{Y}=\{Y_1, \dots, Y_n\}$.
Conversely, given a set of mutually compatible and matching contours $\mathbb{Y}=\{Y_1, \dots, Y_n\}$, there may exist multiple $\bs{\vSi}_{\mathbb{L}_{\vLa}}$ that correspond to $\mathbb{Y}$. We denote the collection of such $\bs{\vSi}_{\mathbb{L}_{\vLa}}$ by $\mathscr{S}(\mathbb{Y})$.

\begin{figure}[h]
 \begin{center}
  \includegraphics[scale=0.7]{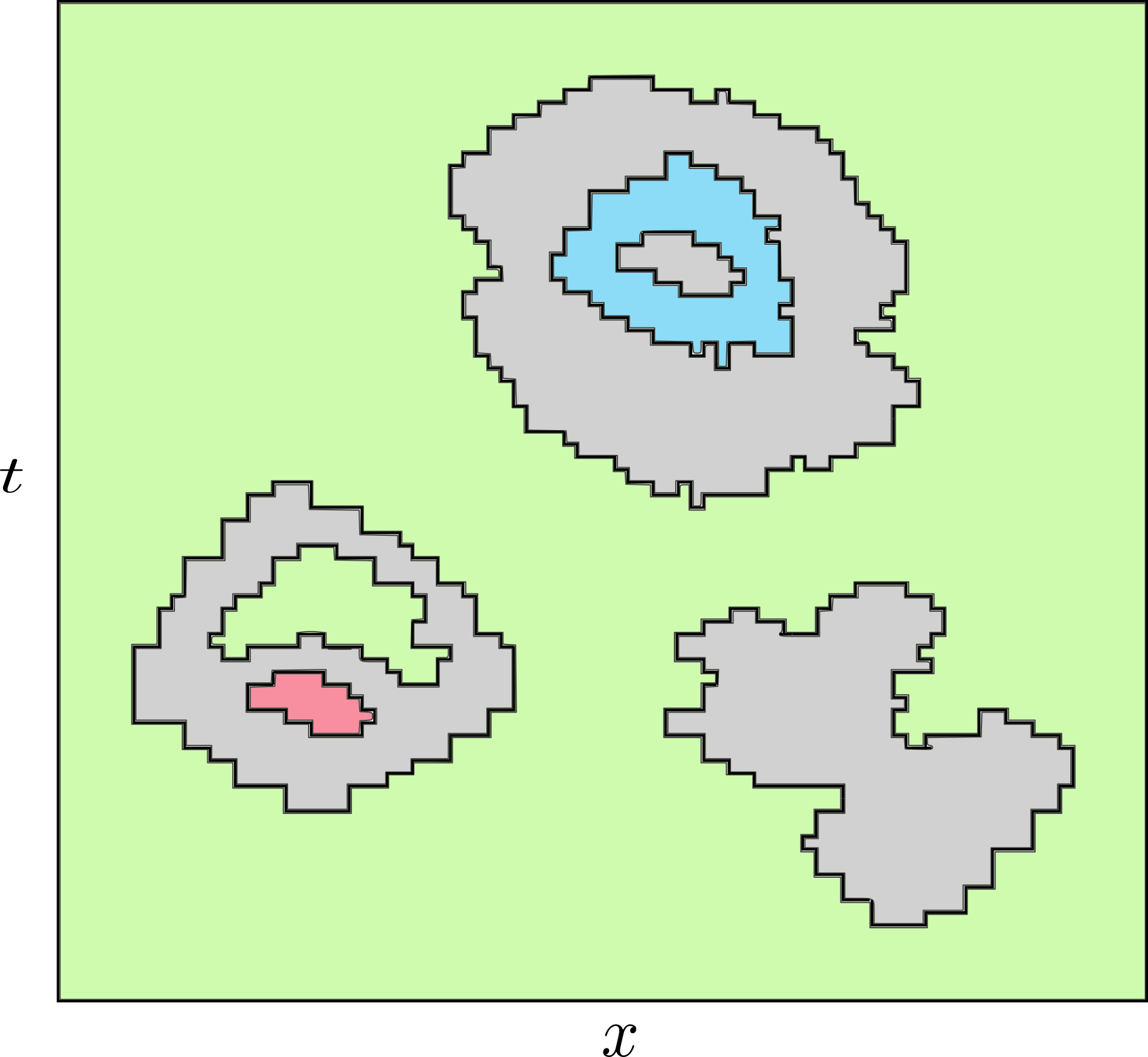}
  \caption{
An example of $\mathbb{Y}=\{Y_1, Y_2, Y_3, Y_4\}$  in $1+1$-dimensional space-time. The gray regions represent $\mathrm{supp}Y_i\ (i=1, \dots, 4)$  in $1+1$-dimensional space-time. The red, blue, and green regions represent the ground state configurations of the classical part $H^{(0)}$. The narrower green region along with the red and blue regions represents $\mathrm{int} Y_i$. Each of the red, blue, and green regions is numbered as 1, 2, and 3, respectively.
When the contour corresponding to the region floating in the blue region is denoted as $Y_1$, then $Y_2, Y_3, Y_4$ are all $3$-contours, and $\mathbb{Y}$ is a set of mutually compatible and matching contours. 
  }\label{Grph3}
        \end{center}
 \end{figure}

For every $\bs{\vSi}_{\mathbb{L}_{\vLa}}\in \mathscr{S}(\mathbb{Y})$, the contribution from $\mathbb{G}$ in $W(\vSi_1, \dots, \vSi_M)$ is given by $\ex^{-\tilde{\beta} \sum_{m}e_{m}|V_{m}|}$, i.e.,
\be
W(\vSi_1, \dots, \vSi_M)=\ex^{-\tilde{\beta}\sum_{m}e_{m}|V_{m}|} w(\bs{\vSi}_{\mathbb{L}_{\vLa}}), 
\ee
where
$|V_m|$ denotes  the number of  elementary cubes contained in $V_m$, and 
 $w(\bs{\vSi}_{\mathbb{L}_{\vLa}})$ represents the remaining contribution.
Hence we arrive at  the following expression of the partition function:
\begin{align}
Z_{\ell, \vLa}=\sum_{\mathbb{Y}}\ex^{-\tilde{\beta}\sum_{m}e_{m}|V_{m}|} \rho(\mathbb{Y}),\quad
\rho(\mathbb{Y})=\sum_{\bs{\vSi}_{\mathbb{L}_{\vLa}}\in \mathscr{S}(\mathbb{Y})}w(\bs{\vSi}_{\mathbb{L}_{\vLa}}),  
\end{align}
where $\sum_{\mathbb{Y}}$ denotes the sum over a collection of contours that are mutually compatible, matching, and such that all the external contours are $\ell$-contours.
The quantity $\rho(\mathbb{Y})$ associated with a collection of contours $\mathbb{Y}$ is referred to as the {\it activity of $\mathbb{Y}$}.

For each 
$\bs{D} \Eql \prod_{t\in \mathbb{M}} D^{(t)}  \in \mathsf{E}$, we define
\be
\bs{\mathcal{N}}_{\bs{D}_{ \rm \sharp} }=\prod_{t\in \mathbb{M}} \mathcal{N}_{D^{(t)}_{\rm \sharp}}\quad(\sharp=c, q),\quad
\bs{\mathcal{N}}_{\bs{D} }=\prod_{t\in \mathbb{M}} \mathcal{N}_{D^{(t)}}.
\ee
Here, $\mathcal{N}_{D^{(t)}_{\rm \sharp}}$ is defined by choosing $B=D_{\sharp}^{(t)}$ in the definition \eqref{DefNB} of $\mathcal{N}_B$.
We represent an element of $\bs{\mathcal{N}}_{\bs{D}}$ as $\bn_{\bs{D}}=(\bn_{D^{(t)}})_{t\in \mathbb{M}}$.
Elements of $\bs{\mathcal{N}}_{\bs{D}_{\rm \sharp }}$ are also represented in the same way.
 Each connected component of $\mathbb{G}$ in $\mathbb{T}$ can be expressed in the form of $\bs{g}^{(\ell)}_{\bs B}=\left(\bs{g}^{(\ell)}_{B^{(t)}}\right)_{t\in \mathbb{M}}$. Here, $\bs{B} $ belongs to $\mathsf{E}$, and we identify it with $\prod_{t\in \mathbb{M}}B^{(t)}$, where $B^{(t)}=[\mathsf{T}_t(\bs{B})]_{\rm S}$.

\subsection{Factorization property of the contour activities}

Suppose that only a single $\ell$-contour $Y=(\bs{D}, \alpha)$ is present in $\mathbb{T}_{\vLa}$. Let $\inte_m Y$ be the union of connected components of $\inte Y$ that are in the $m$-ground state. Then, we have $\inte Y=\sqcup_m \inte_m Y$. As a result, we can represent the ground state configuration in $\mathbb{T}_{\vLa}\setminus \supp Y$ as follows:
\be
\bs{g}^{(*)}_{\mathbb{T}_{\vLa} \setminus \bs{D}}
:=
\bs{g}_{\mathbb{T}_{\vLa} \setminus V(Y)}^{(\ell)} \times
\left(
\prod_m \bs{g}_{\inte_m Y}^{(m)}
\right).
\ee
Here, we employ the identification $\mathbb{T}_{\vLa} \setminus \bs{D} \Eql  \prod_{t \in \mathbb{M}} \vLa \setminus D^{(t)}$. Moreover, the product of configurations is defined as in \eqref{DefProdN}. Let us now consider the corresponding $\mathscr{S}(\{Y\})$. 
For any $\bs{\vSi}_{\mathbb{L}_{\vLa}}=\left( (D_q^{(t)}, \bn_{\vLa}^{(t)})\right)_{ t\in \mathbb{M}}\in \mathscr{S}(\{Y\})$, we express $\bn_{\bs D}\in \bs{\mathcal{N}}_{\bs D}$ and $\bn_{\bs{D}_{\sharp}}\in \bs{\mathcal{N}}_{\bs{D}_{\sharp}}\, (\sharp=c, q)$ as follows:
\begin{align}
\bn_{\bs{D}}=\left(
n_{D^{(t)}}^{(t)}(x) \right)_{ (x, t)\in [\bs{D}]
}, \quad
\bn_{\bs{D}_{\sharp}}=\left(
n_{D_{\sharp}^{(t)}}^{(t)}(x) \right)_{ (x, t)\in [\bs{D}_{\sharp}]}.
\end{align}

In addition, we define $\bE_{\bs{D}_c}$ as the set of configurations $\bn_{\bs{D}_c}$ that are determined by $\bs{\vSi}_{\mathbb{L}_{\vLa}}\in \mathscr{S}(\{Y\})$ such that the union of elementary cubes in a classical excitation state contained in $\supp Y$ is equal to $\bs{D}_c$.
Then, $\rho(Y):=\rho(\{Y\})$ is given by
\begin{align}
\rho(Y)
=&
\sum_{\bs{D}_q \subseteq \bs{D}} \sum_{\bn_{\bs{D}_q}\in \bN_{\bs{D}_q}} \sum_{\bn_{\bs{D}_c}\in \bE_{\bs{D}_c}} \ex^{-\tilde{\beta}  \bs{E}(\bn_{\bs{D}_c })}\times  \no
&\times \prod_{t=1}^{M}\Braket{\bs{g}^{(*)}_{\vLa\setminus D^{(t)}}\times   \bn_{D^{(t)}}^{(t)}| \mathcal{T}_{D_q^{(t+1)}}\left(\bn_{\overline{\partial D}_q^{(t+1)}}\right) |\bs{g}^{(*)}_{\vLa\setminus D^{(t+1)}}\times  \bn_{D^{(t+1)}}^{(t+1)}},  \label{DefRho1}
\end{align}
where $\bs{E}(\bn_{\bs{D}_c })=\sum_{t=1}^M E\left(\bn_{D^{(t)} \setminus \overline{D}_q^{(t)} }\right)$.
We call $\rho(Y)$ the {\it  activity of the contour $Y$}.

\begin{Prop}
For each collection of mutually compatible and matching contours $\mathbb{Y}=\{Y_1, \dots, Y_n\}$,  its   activity satisfies the following factorization property:
\be
\rho(Y_1, \dots, Y_n)=\rho(Y_1) \cdots \rho(Y_n).
\ee
Therefore, we can express the partition function $Z_{\ell, \vLa}$ as
\be
Z_{\ell, \vLa}=\sum_{\mathbb{Y}}\ex^{-\tilde{\beta} \sum_{m}e_{m}|V_{m}|} \rho(Y_1) \cdots \rho(Y_n).
\ee
\end{Prop}
\begin{proof}
 We can prove the assertion in the proposition by utilizing \eqref{HilIdn}, which states that 
$
\mathfrak{H}_{\vLa}=\mathfrak{H}_{\vLa_1}\otimes \mathfrak{H}_{\vLa_2}. 
$
For simplicity, we consider the case where $n=2$: 
$\mathbb{Y}=\{Y_1, Y_2\}$ with 
$Y_i=(\bs{D}_i, \alpha_{\bs{ D}_i})\ (i=1, 2)$.
Fix $\bs{\vSi}_{\mathbb{L}_{\vLa}}\in \mathscr{S}(\mathbb{Y})$ arbitrarily.
Corresponding to $\bs{\vSi}_{\mathbb{L}_{\vLa}}$, we can determine the pairs $(\bs{D}_{q, 1}, \bs{D}_{c, 1})$ and $(\bs{D}_{q, 2}, \bs{D}_{c, 2})$.
Here, $\bs{D}_{q, i}$ (resp. $\bs{D}_{c, i}$) is the union of elementary cubes in $\bs{D}_i$ associated with quantum (resp. classical) excitations.

Let $\bn_{\vLa}\in \mathcal{N}_{\vLa}$. 
By applying \eqref{FacVec}, we obtain 
$
P_{\bn_{\vLa}}=P_{\bn_{ B}}\otimes P_{\bn_{\vLa\setminus B}}
$
 for each $B\subseteq \vLa$
which implies that 
\be
K(\vSi)=
K(B, \bs{n}_{\vLa})=\ex^{-\tilde{\beta}E(\bs{n}_{\vLa\setminus \overline{B}})}
\mathcal{T}_B(\bn_{\overline{\partial B}})P_{\bn_B} \otimes P_{\bs{n}_{\vLa\setminus B}}.
\ee
Using  \eqref{FacC} and \eqref{FacB}, we can identify the operator $\mathcal{T}_B(\bn_{\overline{\partial B}})$ acting on $\mathfrak{H}_{\vLa}$ with $\mathcal{T}_B(\bn_{\overline{\partial B}})\otimes \one_{\vLa\setminus B}$ under the identification $\mathfrak{H}_{\vLa}=\mathfrak{H}_B\otimes \mathfrak{H}_{\vLa\setminus B}$. Here, $\one_{\vLa\setminus B}$ denotes the identity operator on $\mathfrak{H}_{\vLa\setminus B}$.
Therefore, analogous arguments to those used for ordinary quantum spin systems are applicable.

Let $V_m$ be a collection of cubes in the $m$-ground state, and assume that its time slice representation $V_m\Eql \prod_{t=1}^M V_m^{(t)}$ is given.
For each $t\in \mathbb{M}$, $\vLa$ can be decomposed as  $\vLa=D^{(t)}_1\sqcup D^{(t)}_2\sqcup \left(\sqcup_{m}V_{m}^{(t)}\right)$, which implies the following expression of $K(\vSi^{(t)})$:
\begin{align}
K(\vSi^{(t)})
=&\ex^{-\tilde{\beta} E\left(\bn_{(D_{1}^{(t)}\setminus \overline{D}_{q, 1}^{(t)})\sqcup (D_{2}^{(t)}\setminus D_{q, 2}^{(t)})}\right) }
  \ex^{-\tilde{\beta} \sum_{m} e_{m} |V_{m}^{(t)}|}\times\no
&\times \left[ \mathcal{T}_{D_{q, 1}^{(t)}}\left(\bn_{\overline{\partial D}_{q, 1}^{(t)}}\right) P_{\bn_{D_{q, 1}^{(t)}}}\otimes P_{\bn_{D_{c, 1}^{(t)}}}
\right] \otimes  \left[\mathcal{T}_{D_{q, 2}^{(t)}}\left(\bn_{\overline{\partial D}_{q, 2}^{(t)}}\right) P_{\bn_{D_{q, 2}^{(t)}}}\otimes  P_{\bn_{D_{c, 2}^{(t)}}}
\right]\otimes
\no
&\otimes
\left[ P_{g^{(1)}_{V_{1}^{(t)}}}\otimes \cdots \otimes P_{g^{(r)}_{V_{r}^{(t)}}}
\right]
\end{align}
Hence, we obtain the desired factorization:
$
\rho(Y_1, Y_2)=\rho(Y_1) \rho(Y_2)
$.
\end{proof}

\subsection{A useful formula for the contour activities}
In this subsection,  we derive a useful representation of $\rho(Y)$ given by \eqref{DefRho1}. For this purpose, we introduce some symbols. The classical part of the Hamiltonian, $H_{\ell,\vLa}^{(0)}$, can be decomposed as the sum of the electron interaction term, $H_{\ell,\vLa,\rm e}^{(0)}$, and the phonon energy term,  $H_{\ell, \vLa, \rm p}^{(0)}=\omega_0N_{\vLa, \rm p}$: $H_{\ell, \vLa}^{(0)}=H_{\ell, \vLa, \rm e}^{(0)}+H_{\ell, \vLa, \rm p}^{(0)}$.  
Correspondingly, $H_{\ell, B}^{(0)}(\bn_{\overline{\partial B}})$ can be also split as follows:
\be
H_{\ell, B}^{(0)}(\bn_{\overline{\partial B}})=H_{\ell, B, \rm e}^{(0)}(\bn_{\overline{\partial B}, \rm e})+H_{\ell, B, \rm p}^{(0)}(\bn_{\overline{\partial B}, \rm p}), \label{DecH0}
\ee
where  we define $H_{\ell, B, \rm e}^{(0)}(\bn_{\overline{\partial B}, \rm e}) \in \mathfrak{A}_{B, \rm e}^{(0)}$ using the conditional expectation value as follows:
\be
H_{\ell, B, \rm e}^{(0)}(\bn_{\overline{\partial B}, \rm e})=\big\la \bn_{\overline{\partial B}, \rm e}\big|
H_{\ell, \overline{B}, \rm e}^{(0)} \big|\bn_{\overline{\partial B}, \rm e}\big\ra.
\ee
Due to the absence of interactions between different sites, the conditional expectation value  of the phonon energy term takes a simple form as follows:
\be
H_{\ell, B, \rm p}^{(0)}(\bn_{\overline{\partial B}, \rm p})=\omega_0 N_{B, \rm p}+\omega_0 \sum_{x\in \partial B}n_{x, \rm p}.
\ee

Utilizing \eqref{DecH0}, we can represent $\mathcal{T}_B(\omega, \bn_{\overline{\partial B}})$ for every $\omega=(\bs{\tau}, \bs{m})\in \vXi_B$ in the following manner:
\begin{align}
&\mathcal{T}_B(\omega, \bn_{\overline{\partial B}})\no
 =& \tilde{h}_{1, {\rm e},  \bn_{\overline{\partial B}, \rm e}}(s_1)\cdots \tilde{h}_{|\bs{m}|, {\rm e},  \bn_{\overline{\partial B}, \rm e}}(s_{|\bs{m}|}) \ex^{-\tilde{\beta} H_{\ell , B,  \rm e}^{(0)}(\bn_{\overline{\partial B}, \rm e})}
\otimes
\tilde{\varTheta}_1(s_1)\cdots \tilde{\varTheta}_{|\bs{m}|}(s_{|\bs{m}|})\ex^{-\tilde{\beta} H_{\ell , B,  \rm p}^{(0)}(\bn_{\overline{\partial B}, \rm p})}
\no
=:&\mathcal{T}_{B, \rm e}(\omega, \bn_{\overline{\partial B}, \rm e}) \otimes \mathcal{T}_{B, \rm p}(\omega, \bn_{\overline{\partial B}, \rm p}).
\end{align}
Here, if we let $\varTheta_i=\ex^{\im \varPhi_{A_i}}$, then $\tilde{\varTheta}_i$ is defined as $(\tilde{\varTheta}_1, \dots, \tilde{\varTheta}_{|\bs{m}|})=\pi(\varTheta_1, \dots, \varTheta_1, \dots, \varTheta_{|\bs{m}|}, \dots, \varTheta_{|\bs{m}|})$, where the definitions of $s_1, \dots, s_{|\bs{m}|}$ and $\pi$ can be found in \eqref{PiS}. The same applies for the definition of $\tilde{h}_{{\rm e}, i}$.
In addition, the following definitions are employed:
\be
\tilde{h}_{i, {\rm e}, \bn_{\overline{\partial B}, \rm e}}(s)=\ex^{-s H_{\ell , B,  \rm e}^{(0)}(\bn_{\overline{\partial B}, \rm e})}\tilde{h}_{i, {\rm e}}\ex^{s H_{\ell , B,  \rm e}^{(0)}(\bn_{\overline{\partial B}, \rm e})},\quad
\tilde{\varTheta}_i(s)= \ex^{-s H_{\ell , B,  \rm p}^{(0)}(\bn_{\overline{\partial B}, \rm p})} \tilde{\varTheta}_i\ex^{s H_{\ell , B,  \rm p}^{(0)}(\bn_{\overline{\partial B}, \rm p})}.
\ee

We express the electron-phonon ground state configuration of $\bs{g}_B^{(*)}$ as $\bs{g}_B^{(*)}=(\bs{g}_{B, {\rm e}}^{(*)}, \bs{g}^{(*)}_{B, {\rm p}})$, where the configuration $\bs{g}^{(*)}_{B, {\rm p}}$  corresponds to the Fock vacuum: $\ket{\bs{g}^{(*)}_{B, {\rm p}}}=\ket{\varnothing}_{B, \rm p}$. Then  we have 
\begin{align}
&\Braket{\bs{g}^{(*)}_{\vLa\setminus D^{(t)}}\times   \bn_{D^{(t)}}^{(t)}| \mathcal{T}_{D_q^{(t+1)}}\left(\bn_{\overline{\partial D}_q^{(t+1)}}\right) |\bs{g}^{(*)}_{\vLa\setminus D^{(t+1)}}\times  \bn_{D^{(t+1)}}^{(t+1)}}\no
=&\int_{\varXi_{D_q^{(t+1)}}}\mathscr{D}(\omega^{(t+1)})\Braket{\bs{g}^{(*)}_{\vLa\setminus D^{(t)}} \times  \bn_{D^{(t)}}^{(t)}| \mathcal{T}_{D_q^{(t+1)}}\left(\omega^{(t+1)},  \bn_{\overline{\partial D}_q^{(t+1)}}\right)
 |\bs{g}^{(*)}_{\vLa\setminus D^{(t+1)}}\times  \bn_{D^{(t+1)}}^{(t+1)}} \no
=&\int_{\varXi_{D_q^{(t+1)}}}\mathscr{D}(\omega^{(t+1)})\Braket{\bs{g}^{(*)}_{\vLa\setminus D^{(t)},  \rm e} \times  \bn_{D^{(t)},  \rm e}^{(t)}| 
\mathcal{T}_{D_q^{(t+1)}, \rm e}\left(\omega^{(t+1)},  \bn_{\overline{\partial D}_q^{(t+1)}, \rm e}\right)
 |\bs{g}^{(*)}_{\vLa\setminus D^{(t+1)},  \rm e}\times  \bn_{D^{(t+1)},  \rm e}^{(t+1)}}\times\no
 &\times \Braket{\bs{g}^{(*)}_{\vLa\setminus D^{(t)},  \rm p}\times  \bn_{D^{(t)},  \rm p}^{(t)}| 
 \mathcal{T}_{D_q^{(t+1)}, \rm p}\left(\omega^{(t+1)}, \bn_{\overline{\partial D}_q^{(t+1)}, \rm p}\right)
 |\bs{g}^{(*)}_{\vLa\setminus D^{(t+1)},  \rm p}\times  \bn_{D^{(t+1)},  \rm p}^{(t+1)}}.
\end{align}
For the sake of simplicity of notation, we define
\be
\int_{\bs{\varXi}_{\bs{D}_q}} \mathscr{D}(\bome) := \prod_{t=1}^M \int_{\varXi_{D_q^{(t)}}} \mathscr{D}(\omega^{(t)}),
\ee
where $\bs{\varXi}_{\bs{D}_q}=\prod_{t=1}^M \varXi_{D_q^{(t)}}$ and $\bome=(\omega^{(t)})_{t=1}^M$.
Using this notation, we can write $\rho(Y)$ as follows:
\begin{align}
\rho(Y)
=& \sum_{\bs{D}_q \subseteq \bs{D}} \sum_{\bn_{\bs{D}_q}} \sum_{\bn_{\bs{D}_c}} \int_{\bs{\varXi}_{\bs{D}_q}} \mathscr{D}(\bome)\,  \ex^{-\beta \bs{E}(\bn_{\bs{D}_ c})}\times \no
&\times \prod_{t=1}^M \Braket{\bs{g}^{(*)}_{\vLa\setminus D^{(t)},  \rm e} \times  \bn_{D^{(t)},  \rm e}^{(t)}| 
\mathcal{T}_{D_q^{(t+1)}, \rm e}\left(\omega^{(t+1)},  \bn_{\overline{\partial D}_q^{(t+1)}, \rm e}\right)
 |\bs{g}^{(*)}_{\vLa\setminus D^{(t+1)},  \rm e}\times  \bn_{D^{(t+1)},  \rm e}^{(t+1)}}\times\no
 &\times \Braket{\bs{g}^{(*)}_{\vLa\setminus D^{(t)},  \rm p}\times  \bn_{D^{(t)},  \rm p}^{(t)}| 
 \mathcal{T}_{D_q^{(t+1)}, \rm p}\left(\omega^{(t+1)},  \bn_{\overline{\partial D}_q^{(t+1)}, \rm p}\right)
 |\bs{g}^{(*)}_{\vLa\setminus D^{(t+1)},  \rm p}\times  \bn_{D^{(t+1)},  \rm p}^{(t+1)}}
\no
=&\int_{\bs{\varPi}_Y} \bs{\mathscr{D}}(\bs{\vphi})\, \ex^{-\tilde{\beta}\bs{E}(\bn_{\bs{D}_ c})} S_{\rm e}(\bs{\vphi}_{\rm e})S_{\rm p}(\bs{\vphi}_{\rm p}).
\label{RhoEP}
\end{align}
Let us explain the symbols appearing in \eqref{RhoEP}. Firstly, 
\be
\int_{\bs{\varPi}_Y} \bs{\mathscr{D}}(\bs{\vphi}):=\sum_{\bs{D}_q \subseteq \bs{D}} \sum_{\bn_{\bs{D}_q}} \sum_{\bn_{\bs{D}_c}} \int_{\bs{\varXi}_{\bs{D}_q}} \mathscr{D}(\bome).
\ee
Secondly, $S_{\flat }(\bphi_{\flat })\ (\flat =\rm e, p)$ are given by 
\begin{align}
S_{\flat }(\bphi_{\flat })&=\prod_{t=1}^M 
\Braket{\bs{g}^{(*)}_{\vLa\setminus D^{(t)},  \rm \flat } \times  \bn_{D^{(t)},  \rm \flat }^{(t)}| \mathcal{T}_{D_q^{(t+1)}, \flat}\left(\omega^{(t+1)},  \bn_{\overline{\partial D}_q^{(t+1)}, \flat}\right)
 |\bs{g}^{(*)}_{\vLa\setminus D^{(t+1)},  \rm \flat }\times  \bn_{D^{(t+1)},  \rm \flat }^{(t+1)}}.
\end{align}
Thirdly, we define
\be
\bs{\varPi}_Y=\left\{\bphi=(\bs{D}_q, \bs{D}_c, \bome, \bn_{\bs D}) :  \mbox{$(\bs{D}, \bn_{\bs D})\in \mathscr{S}(\{Y\})$ and $\bome\in \bs{\varXi}_{\bs{D}_q}$}\right\}, 
\ee
and, finally,  for each $\bphi\in \bs{\varPi}_Y$, $\bphi_{\rm e}$ and $\bphi_{\rm p}$ are defined by 
\be
\bphi_{\flat}=(\bs{D}_{q}, \bs{D}_{c}, \bome, \bn_{\bs D, \flat})\quad (\flat={\rm e, p}).
\ee
Here, for each $\bn_{\bs D}=\left((\bn^{(t)}_{D^{(t)}, \rm e}, \bn^{(t)}_{D^{(t)}, \rm p})\right)_{t=1}^M \in \bs{\mathcal{N}_D}$, we define $\bn_{\bs D, \flat}=(\bn^{(t)}_{D^{(t)}, \flat})_{t=1}^M\ (\flat={\rm e, p})$ with $\bn^{(t)}_{D^{(t)}, \flat}\in \mathcal{N}_{D^{(t)}, \flat}$.
In the subsequent sections, readers will find the expression \eqref{RhoEP} useful when evaluating contour activities.

\section{Estimates of the contour activities}\label{Sec4}

\subsection{The objective and outline of Section \ref{Sec4}}\label{IntSec4}
As explained in Subsection \ref{StrMnTh}, we prove the main theorem (Theorem \ref{MainThm}) by accomplishing steps \hyperlink{Q1}{\bf (Q. 1)} and \hyperlink{Q2}{\bf (Q. 2)}. In Section \ref{Sec3}, we have already achieved \hyperlink{Q1}{\bf (Q. 1)}. In this section, we complete \hyperlink{Q2}{\bf (Q. 2)}. As mentioned in the previous section, the foundation of the Pirogov--Sinai theory for the space-time contour model has been laid in \cite{Borgs1996}. The primary goal of this section is to verify that several technical conditions concerning the contour model, necessary for applying the results of \cite{Borgs1996}, are indeed satisfied. Refer to Propositions \ref{ActBd1}, \ref{RhoEst2}, \ref{DelRhoEst}, \ref{ExpRhoPsi} for details. Once these propositions are incorporated, the main theorem can be proven by applying various theorems from \cite{Borgs1996}.

In order to establish these propositions, a detailed analysis of the quantum perturbation arising from electron-phonon interaction is essential.
The quantum perturbation term 
$H_{\vLa}^{\rm (Q)}$
  incorporates unitary operators, such as 
$\ex^{\im \vartheta_A}$, determined by the bosonic field operators. 
 As these factors are unitary operators, one might believe that the analytical methods employed for the Hubbard model in \cite{Borgs_1996} can be readily applied. However, the bosonic Fock space associated with each site is an infinite-dimensional Hilbert space, and various operators describing bosons are unbounded. Therefore, using such a simplistic approach, it is not possible to rigorously analyze the electron-phonon interacting system. A meticulous analysis of vacuum expectations and thermal expectations for the bosonic field operators given in \eqref{DefU} is imperative.
 
Fortuitously, extensive analyses pertaining to these expectations of  bosonic field operators have been conducted in the realm of quantum field theory \cite{Arai2022}. By leveraging this wealth of knowledge, we can accomplish \hyperlink{Q2}{\bf (Q. 2)}. The insights into the expectations of the bosonic field operators utilized in this paper are succinctly summarized in Appendix \ref{AppCor}.

\subsection{The case where $Y$ does not wind around $\mathbb{T}_{\vLa}$ }\label{Sec4.2}
In this subsection, we consider the case where  the contour $Y$ does not wind around the thermal torus $\mathbb{T}_{\vLa}$.
Refer to Figure \ref{Grph4}.
When $Y$ winds around $\mathbb{T}_{\vLa}$, the mathematical treatment is somewhat different, and this will be discussed further in Subsection \ref{Sect4.3}.
\begin{figure}[htbp]
    \centering
     \includegraphics[scale=0.6]{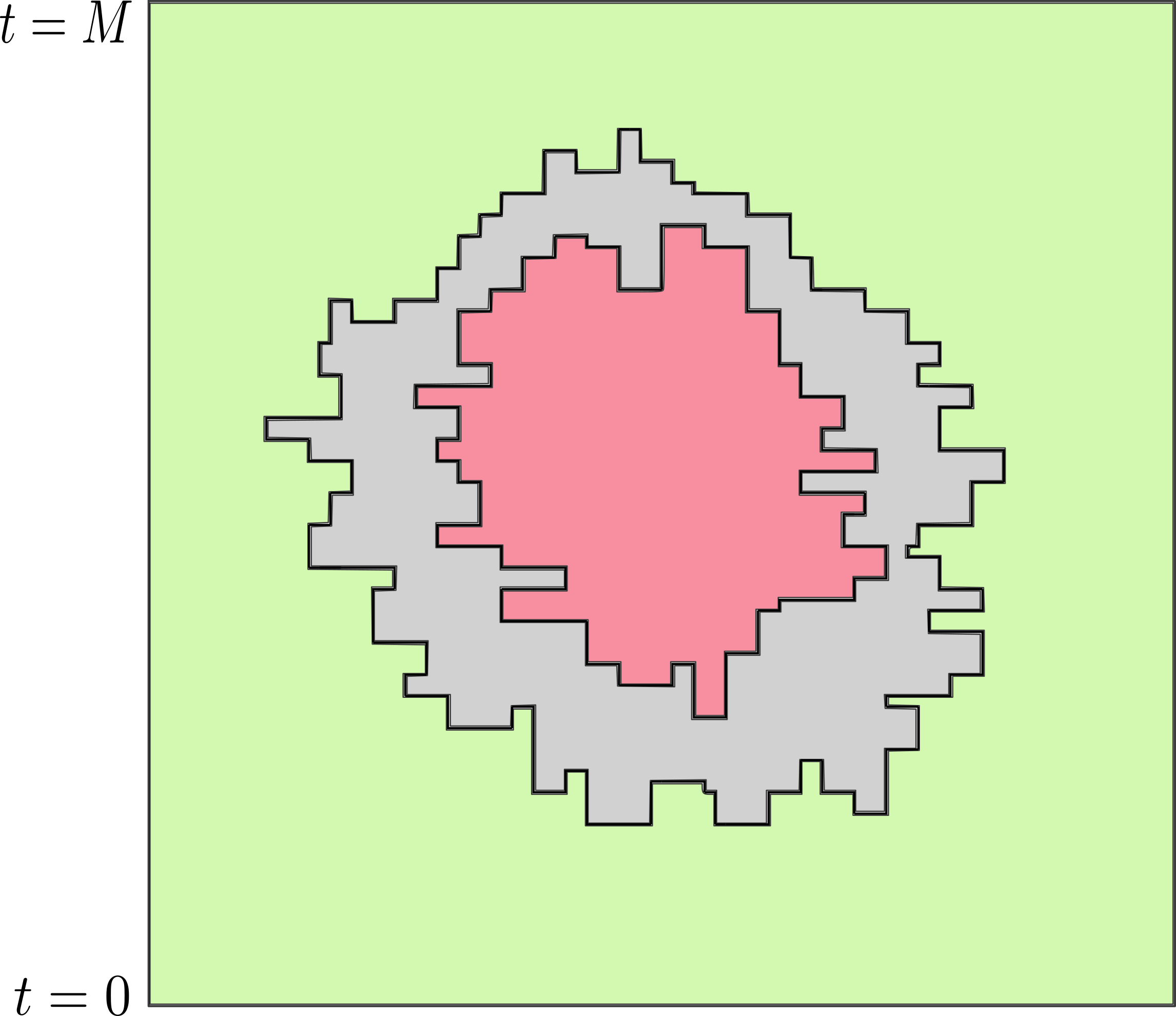}
    \caption{Image depicting a contour not winding around $\mathbb{T}_{\vLa}$. (To be precise,  
    section of the support of a contour not winding around $\mathbb{T}_{\vLa}$ on a specific plane defined by $x_i = \rm const.$)
    }\label{Grph4}
\end{figure}

Fix $\bphi=(\bs{D}_q, \bs{D}_c, \bome, \bn_{\bs D})\in \bs{\varPi}_Y$, arbitrarily. Then, for each $x\in [\bs{D}]_{\mathsf{S}} $, we define 
$\bphi(x)=\left(\bs{D}_q(x), \bs{D}_c(x), \bome(x), \bn_{\bs{D}}(x)\right)$ as follows: 
\begin{align}
\bs{D}_{\sharp}(x)&=
\left( \bigcup_{t\in \mathbb{M}} C(x, t) \right) \cap \bs{D}_{\sharp}
 \ (\sharp=c, q), \\
\bn_{\bs{D}}(x)&=\left(
n_{D^{(t)}}^{(t)}(x) \right)_{ t\in [\mathsf{S}_{x}(\bs{D})]_{\rm T}
},\\
\bome(x)&=\left(\omega^{(t)}(x) \right)_{
t\in [\mathsf{S}_x(\bs{D}_{q})]_{\rm T}
 },\quad
\omega^{(t)}(x) =\left\{ (\bs{m}_A^{(t)},  \bs{\tau}^{(t)}_A) :x\in \supp A\right\}.
\end{align}
Depending on $x$, $\bs{D}_{\sharp}(x)$ could be an empty set. If $\bs{D}_{q}(x)=\varnothing$, then we set $\bome(x)=\varnothing$.
See Figure \ref{PicDs}.
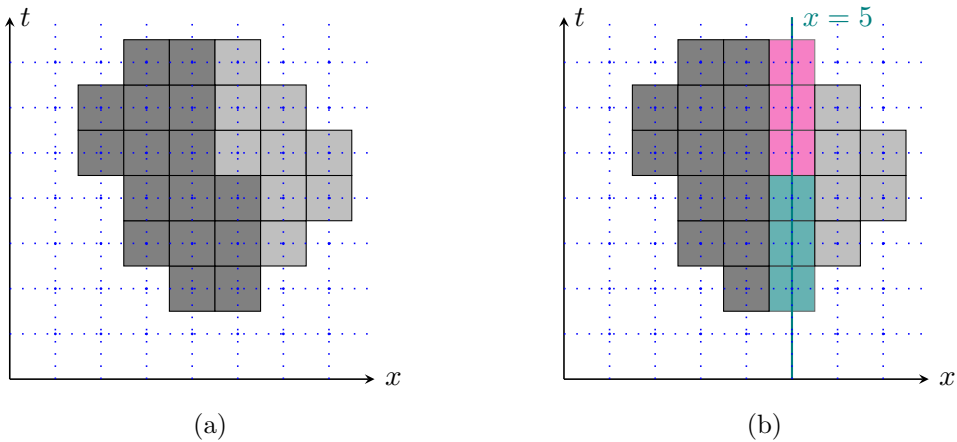
\begin{figure}[htbp]
  \begin{minipage}[b]{0.45\linewidth}
    \centering
     \begin{tikzpicture}[scale=0.6]
\coordinate(Z11)at(\Ct+1,\Ct-1);
\coordinate(Z12)at(\Ct+1,\Ct);
\coordinate(Z13)at(\Ct+2,\Ct);
\coordinate(Z14)at(\Ct+2,\Ct-1);
\filldraw[fill=gray](Z11)--(Z12)--(Z13)--(Z14)--cycle;
\coordinate(Z21)at(\Ct+2,\Ct-1);
\coordinate(Z22)at(\Ct+2,\Ct);
\coordinate(Z23)at(\Ct+3,\Ct);
\coordinate(Z24)at(\Ct+3,\Ct-1);
\filldraw[fill=gray](Z21)--(Z22)--(Z23)--(Z24)--cycle;

\coordinate(A11)at(\Ct,\Ct);
\coordinate(A12)at(\Ct,\Ct+1);
\coordinate(A13)at(\Ct+1,\Ct+1);
\coordinate(A14)at(\Ct+1,\Ct);
\filldraw[fill=gray](A11)--(A12)--(A13)--(A14)--cycle;
\coordinate(A21)at(\Ct+1,\Ct);
\coordinate(A22)at(\Ct+1,\Ct+1);
\coordinate(A23)at(\Ct+2,\Ct+1);
\coordinate(A24)at(\Ct+2,\Ct);
\filldraw[fill=gray](A21)--(A22)--(A23)--(A24)--cycle;
\coordinate(A31)at(\Ct+2,\Ct);
\coordinate(A32)at(\Ct+2,\Ct+1);
\coordinate(A33)at(\Ct+3,\Ct+1);
\coordinate(A34)at(\Ct+3,\Ct);
\filldraw[fill=gray](A31)--(A32)--(A33)--(A34)--cycle;
\coordinate(A41)at(\Ct+2+1,\Ct);
\coordinate(A42)at(\Ct+2+1,\Ct+1);
\coordinate(A43)at(\Ct+3+1,\Ct+1);
\coordinate(A44)at(\Ct+3+1,\Ct);
\filldraw[fill=lightgray](A41)--(A42)--(A43)--(A44)--cycle;

\coordinate(B11)at(\Ct,\Ct+1);
\coordinate(B12)at(\Ct,\Ct+2);
\coordinate(B13)at(\Ct+1,\Ct+2);
\coordinate(B14)at(\Ct+1,\Ct+1);
\filldraw[fill=gray](B11)--(B12)--(B13)--(B14)--cycle;
\coordinate(B21)at(\Ct+1,\Ct+1);
\coordinate(B22)at(\Ct+1,\Ct+2);
\coordinate(B23)at(\Ct+2,\Ct+2);
\coordinate(B24)at(\Ct+2,\Ct+1);
\filldraw[fill=gray](B21)--(B22)--(B23)--(B24)--cycle;
\coordinate(B31)at(\Ct+4,\Ct+1);
\coordinate(B32)at(\Ct+4,\Ct+2);
\coordinate(B33)at(\Ct+5,\Ct+2);
\coordinate(B34)at(\Ct+5,\Ct+1);
\filldraw[fill=lightgray](B31)--(B32)--(B33)--(B34)--cycle;
\coordinate(B41)at(\Ct+3,\Ct+1);
\coordinate(B42)at(\Ct+3,\Ct+2);
\coordinate(B43)at(\Ct+4,\Ct+2);
\coordinate(B44)at(\Ct+4,\Ct+1);
\filldraw[fill=lightgray](B41)--(B42)--(B43)--(B44)--cycle;
\coordinate(B51)at(\Ct+2,\Ct+1);
\coordinate(B52)at(\Ct+2,\Ct+2);
\coordinate(B53)at(\Ct+3,\Ct+2);
\coordinate(B54)at(\Ct+3,\Ct+1);
\filldraw[fill=gray](B51)--(B52)--(B53)--(B54)--cycle;

\coordinate(C11)at(\Ct-1,\Ct+2);
\coordinate(C12)at(\Ct-1,\Ct+3);
\coordinate(C13)at(\Ct,\Ct+3);
\coordinate(C14)at(\Ct,\Ct+2);
\filldraw[fill=gray](C11)--(C12)--(C13)--(C14)--cycle;
\coordinate(C21)at(\Ct,\Ct+2);
\coordinate(C22)at(\Ct,\Ct+3);
\coordinate(C23)at(\Ct+1,\Ct+3);
\coordinate(C24)at(\Ct+1,\Ct+2);
\filldraw[fill=gray](C21)--(C22)--(C23)--(C24)--cycle;
\coordinate(C31)at(\Ct+3,\Ct+2);
\coordinate(C32)at(\Ct+3,\Ct+3);
\coordinate(C33)at(\Ct+4,\Ct+3);
\coordinate(C34)at(\Ct+4,\Ct+2);
\filldraw[fill=lightgray](C31)--(C32)--(C33)--(C34)--cycle;
\coordinate(C41)at(\Ct+3+1,\Ct+2);
\coordinate(C42)at(\Ct+3+1,\Ct+3);
\coordinate(C43)at(\Ct+4+1,\Ct+3);
\coordinate(C44)at(\Ct+4+1,\Ct+2);
\filldraw[fill=lightgray](C41)--(C42)--(C43)--(C44)--cycle;
\coordinate(C51)at(\Ct+1,\Ct+2);
\coordinate(C52)at(\Ct+1,\Ct+3);
\coordinate(C53)at(\Ct+2,\Ct+3);
\coordinate(C54)at(\Ct+2,\Ct+2);
\filldraw[fill=gray](C51)--(C52)--(C53)--(C54)--cycle;
\coordinate(C61)at(\Ct+2,\Ct+2);
\coordinate(C62)at(\Ct+2,\Ct+3);
\coordinate(C63)at(\Ct+3,\Ct+3);
\coordinate(C64)at(\Ct+3,\Ct+2);
\filldraw[fill=lightgray](C61)--(C62)--(C63)--(C64)--cycle;

\coordinate(D11)at(\Ct-1,\Ct+3);
\coordinate(D12)at(\Ct-1,\Ct+4);
\coordinate(D13)at(\Ct,\Ct+4);
\coordinate(D14)at(\Ct,\Ct+3);
\filldraw[fill=gray](D11)--(D12)--(D13)--(D14)--cycle;
\coordinate(D21)at(\Ct+4-2,\Ct+3);
\coordinate(D22)at(\Ct+4-2,\Ct+4);
\coordinate(D23)at(\Ct+5-2,\Ct+4);
\coordinate(D24)at(\Ct+5-2,\Ct+3);
\filldraw[fill=lightgray](D21)--(D22)--(D23)--(D24)--cycle;

\coordinate(D31)at(\Ct+1,\Ct+3);
\coordinate(D32)at(\Ct+1,\Ct+4);
\coordinate(D33)at(\Ct+2,\Ct+4);
\coordinate(D34)at(\Ct+2,\Ct+3);
\filldraw[fill=gray](D31)--(D32)--(D33)--(D34)--cycle;
\coordinate(D41)at(\Ct,\Ct+3);
\coordinate(D42)at(\Ct,\Ct+4);
\coordinate(D43)at(\Ct+1,\Ct+4);
\coordinate(D44)at(\Ct+1,\Ct+3);
\filldraw[fill=gray](D41)--(D42)--(D43)--(D44)--cycle;

\coordinate(D51)at(\Ct+4-1,\Ct+3);
\coordinate(D52)at(\Ct+4-1,\Ct+4);
\coordinate(D53)at(\Ct+5-1,\Ct+4);
\coordinate(D54)at(\Ct+5-1,\Ct+3);
\filldraw[fill=lightgray](D51)--(D52)--(D53)--(D54)--cycle;

\coordinate(E11)at(\Ct,\Ct+3+1);
\coordinate(E12)at(\Ct,\Ct+4+1);
\coordinate(E13)at(\Ct+1,\Ct+4+1);
\coordinate(E14)at(\Ct+1,\Ct+3+1);
\filldraw[fill=gray](E11)--(E12)--(E13)--(E14)--cycle;
\coordinate(E21)at(\Ct+1,\Ct+3+1);
\coordinate(E22)at(\Ct+1,\Ct+4+1);
\coordinate(E23)at(\Ct+2,\Ct+4+1);
\coordinate(E24)at(\Ct+2,\Ct+3+1);
\filldraw[fill=gray](E21)--(E22)--(E23)--(E24)--cycle;
\coordinate(E41)at(\Ct+2,\Ct+3+1);
\coordinate(E42)at(\Ct+2,\Ct+4+1);
\coordinate(E43)at(\Ct+3,\Ct+4+1);
\coordinate(E44)at(\Ct+3,\Ct+3+1);
\filldraw[fill=lightgray](E41)--(E42)--(E43)--(E44)--cycle;




\draw[->,>=stealth,semithick] (0,0)--(\Ti,0)node[right]{$x$}; 
\draw[->,>=stealth,semithick] (0,0)--(0,\Sp)node[right]{$t$}; 
\draw[name path=X1,loosely dotted,  semithick, blue] (1,0) -- (1,\Ti);
\draw[name path=X2,loosely dotted, semithick, blue] (2,0) -- (2,\Ti);
\draw[name path=X3,loosely dotted, semithick, blue] (3,0) -- (3,\Ti);
\draw[name path=X4,loosely dotted, semithick, blue] (5,0) -- (5,\Ti);
\draw[name path=X5,loosely dotted, semithick, blue] (4,0) -- (4,\Ti);
\draw[name path=X6,loosely dotted, semithick, blue] (6,0) -- (6,\Ti);
\draw[name path=X7,loosely dotted, semithick, blue] (7,0) -- (7,\Ti);
\draw[name path=Y1,loosely dotted, semithick, blue] (0,1) -- (\Sp,1);
\draw[name path=Y2, loosely dotted, semithick, blue] (0,2) -- (\Sp,2);
\draw[name path=Y3, loosely dotted, semithick, blue] (0,3) -- (\Sp,3);
\draw[name path=Y4,loosely dotted, semithick, blue] (0,4) -- (\Sp,4);
\draw[name path=Y5,loosely dotted, semithick, blue] (0,5) -- (\Sp,5);
\draw[name path=Y6,loosely dotted, semithick, blue] (0,6) -- (\Sp,6);
\draw[name path=Y7,loosely dotted, semithick, blue] (0,7) -- (\Sp,7);

\foreach\Q in { 1,2,3, 4, 5, 6, 7}\foreach\P in { 1,2,3, 4, 5, 6, 7}\fill[blue](\P,\Q)circle(0.03);


\end{tikzpicture}
    \subcaption{}\label{a}
  \end{minipage}
  \begin{minipage}[b]{0.45\linewidth}
    \centering
     \begin{tikzpicture}[scale=0.6]
\coordinate(Z11)at(\Ct+1,\Ct-1);
\coordinate(Z12)at(\Ct+1,\Ct);
\coordinate(Z13)at(\Ct+2,\Ct);
\coordinate(Z14)at(\Ct+2,\Ct-1);
\filldraw[fill=gray](Z11)--(Z12)--(Z13)--(Z14)--cycle;
\coordinate(Z21)at(\Ct+2,\Ct-1);
\coordinate(Z22)at(\Ct+2,\Ct);
\coordinate(Z23)at(\Ct+3,\Ct);
\coordinate(Z24)at(\Ct+3,\Ct-1);
\filldraw[fill=teal,opacity=0.6](Z21)--(Z22)--(Z23)--(Z24)--cycle;

\coordinate(A11)at(\Ct,\Ct);
\coordinate(A12)at(\Ct,\Ct+1);
\coordinate(A13)at(\Ct+1,\Ct+1);
\coordinate(A14)at(\Ct+1,\Ct);
\filldraw[fill=gray](A11)--(A12)--(A13)--(A14)--cycle;
\coordinate(A21)at(\Ct+1,\Ct);
\coordinate(A22)at(\Ct+1,\Ct+1);
\coordinate(A23)at(\Ct+2,\Ct+1);
\coordinate(A24)at(\Ct+2,\Ct);
\filldraw[fill=gray](A21)--(A22)--(A23)--(A24)--cycle;
\coordinate(A31)at(\Ct+2,\Ct);
\coordinate(A32)at(\Ct+2,\Ct+1);
\coordinate(A33)at(\Ct+3,\Ct+1);
\coordinate(A34)at(\Ct+3,\Ct);
\filldraw[fill=teal, opacity=0.6](A31)--(A32)--(A33)--(A34)--cycle;
\coordinate(A41)at(\Ct+2+1,\Ct);
\coordinate(A42)at(\Ct+2+1,\Ct+1);
\coordinate(A43)at(\Ct+3+1,\Ct+1);
\coordinate(A44)at(\Ct+3+1,\Ct);
\filldraw[fill=lightgray](A41)--(A42)--(A43)--(A44)--cycle;

\coordinate(B11)at(\Ct,\Ct+1);
\coordinate(B12)at(\Ct,\Ct+2);
\coordinate(B13)at(\Ct+1,\Ct+2);
\coordinate(B14)at(\Ct+1,\Ct+1);
\filldraw[fill=gray](B11)--(B12)--(B13)--(B14)--cycle;
\coordinate(B21)at(\Ct+1,\Ct+1);
\coordinate(B22)at(\Ct+1,\Ct+2);
\coordinate(B23)at(\Ct+2,\Ct+2);
\coordinate(B24)at(\Ct+2,\Ct+1);
\filldraw[fill=gray](B21)--(B22)--(B23)--(B24)--cycle;
\coordinate(B31)at(\Ct+4,\Ct+1);
\coordinate(B32)at(\Ct+4,\Ct+2);
\coordinate(B33)at(\Ct+5,\Ct+2);
\coordinate(B34)at(\Ct+5,\Ct+1);
\filldraw[fill=lightgray](B31)--(B32)--(B33)--(B34)--cycle;
\coordinate(B41)at(\Ct+3,\Ct+1);
\coordinate(B42)at(\Ct+3,\Ct+2);
\coordinate(B43)at(\Ct+4,\Ct+2);
\coordinate(B44)at(\Ct+4,\Ct+1);
\filldraw[fill=lightgray](B41)--(B42)--(B43)--(B44)--cycle;
\coordinate(B51)at(\Ct+2,\Ct+1);
\coordinate(B52)at(\Ct+2,\Ct+2);
\coordinate(B53)at(\Ct+3,\Ct+2);
\coordinate(B54)at(\Ct+3,\Ct+1);
\filldraw[fill=teal, opacity=0.6](B51)--(B52)--(B53)--(B54)--cycle;

\coordinate(C11)at(\Ct-1,\Ct+2);
\coordinate(C12)at(\Ct-1,\Ct+3);
\coordinate(C13)at(\Ct,\Ct+3);
\coordinate(C14)at(\Ct,\Ct+2);
\filldraw[fill=gray](C11)--(C12)--(C13)--(C14)--cycle;
\coordinate(C21)at(\Ct,\Ct+2);
\coordinate(C22)at(\Ct,\Ct+3);
\coordinate(C23)at(\Ct+1,\Ct+3);
\coordinate(C24)at(\Ct+1,\Ct+2);
\filldraw[fill=gray](C21)--(C22)--(C23)--(C24)--cycle;
\coordinate(C31)at(\Ct+3,\Ct+2);
\coordinate(C32)at(\Ct+3,\Ct+3);
\coordinate(C33)at(\Ct+4,\Ct+3);
\coordinate(C34)at(\Ct+4,\Ct+2);
\filldraw[fill=lightgray](C31)--(C32)--(C33)--(C34)--cycle;
\coordinate(C41)at(\Ct+3+1,\Ct+2);
\coordinate(C42)at(\Ct+3+1,\Ct+3);
\coordinate(C43)at(\Ct+4+1,\Ct+3);
\coordinate(C44)at(\Ct+4+1,\Ct+2);
\filldraw[fill=lightgray](C41)--(C42)--(C43)--(C44)--cycle;
\coordinate(C51)at(\Ct+1,\Ct+2);
\coordinate(C52)at(\Ct+1,\Ct+3);
\coordinate(C53)at(\Ct+2,\Ct+3);
\coordinate(C54)at(\Ct+2,\Ct+2);
\filldraw[fill=gray](C51)--(C52)--(C53)--(C54)--cycle;
\coordinate(C61)at(\Ct+2,\Ct+2);
\coordinate(C62)at(\Ct+2,\Ct+3);
\coordinate(C63)at(\Ct+3,\Ct+3);
\coordinate(C64)at(\Ct+3,\Ct+2);
\filldraw[fill=magenta, opacity=0.5](C61)--(C62)--(C63)--(C64)--cycle;

\coordinate(D11)at(\Ct-1,\Ct+3);
\coordinate(D12)at(\Ct-1,\Ct+4);
\coordinate(D13)at(\Ct,\Ct+4);
\coordinate(D14)at(\Ct,\Ct+3);
\filldraw[fill=gray](D11)--(D12)--(D13)--(D14)--cycle;
\coordinate(D21)at(\Ct+4-2,\Ct+3);
\coordinate(D22)at(\Ct+4-2,\Ct+4);
\coordinate(D23)at(\Ct+5-2,\Ct+4);
\coordinate(D24)at(\Ct+5-2,\Ct+3);
\filldraw[fill=magenta, opacity=0.5](D21)--(D22)--(D23)--(D24)--cycle;

\coordinate(D31)at(\Ct+1,\Ct+3);
\coordinate(D32)at(\Ct+1,\Ct+4);
\coordinate(D33)at(\Ct+2,\Ct+4);
\coordinate(D34)at(\Ct+2,\Ct+3);
\filldraw[fill=gray](D31)--(D32)--(D33)--(D34)--cycle;
\coordinate(D41)at(\Ct,\Ct+3);
\coordinate(D42)at(\Ct,\Ct+4);
\coordinate(D43)at(\Ct+1,\Ct+4);
\coordinate(D44)at(\Ct+1,\Ct+3);
\filldraw[fill=gray](D41)--(D42)--(D43)--(D44)--cycle;

\coordinate(D51)at(\Ct+4-1,\Ct+3);
\coordinate(D52)at(\Ct+4-1,\Ct+4);
\coordinate(D53)at(\Ct+5-1,\Ct+4);
\coordinate(D54)at(\Ct+5-1,\Ct+3);
\filldraw[fill=lightgray](D51)--(D52)--(D53)--(D54)--cycle;

\coordinate(E11)at(\Ct,\Ct+3+1);
\coordinate(E12)at(\Ct,\Ct+4+1);
\coordinate(E13)at(\Ct+1,\Ct+4+1);
\coordinate(E14)at(\Ct+1,\Ct+3+1);
\filldraw[fill=gray](E11)--(E12)--(E13)--(E14)--cycle;
\coordinate(E21)at(\Ct+1,\Ct+3+1);
\coordinate(E22)at(\Ct+1,\Ct+4+1);
\coordinate(E23)at(\Ct+2,\Ct+4+1);
\coordinate(E24)at(\Ct+2,\Ct+3+1);
\filldraw[fill=gray](E21)--(E22)--(E23)--(E24)--cycle;
\coordinate(E41)at(\Ct+2,\Ct+3+1);
\coordinate(E42)at(\Ct+2,\Ct+4+1);
\coordinate(E43)at(\Ct+3,\Ct+4+1);
\coordinate(E44)at(\Ct+3,\Ct+3+1);
\filldraw[fill=magenta, opacity=0.5](E41)--(E42)--(E43)--(E44)--cycle;


\draw[-,thick, teal] (5,0)--(5,\Sp)node[right]{$x=5$}; 


\draw[->,>=stealth,semithick] (0,0)--(\Ti,0)node[right]{$x$}; 
\draw[->,>=stealth,semithick] (0,0)--(0,\Sp)node[right]{$t$}; 
\draw[name path=X1,loosely dotted,  semithick, blue] (1,0) -- (1,\Ti);
\draw[name path=X2,loosely dotted, semithick, blue] (2,0) -- (2,\Ti);
\draw[name path=X3,loosely dotted, semithick, blue] (3,0) -- (3,\Ti);
\draw[name path=X4,loosely dotted, semithick, blue] (5,0) -- (5,\Ti);
\draw[name path=X5,loosely dotted, semithick, blue] (4,0) -- (4,\Ti);
\draw[name path=X6,loosely dotted, semithick, blue] (6,0) -- (6,\Ti);
\draw[name path=X7,loosely dotted, semithick, blue] (7,0) -- (7,\Ti);
\draw[name path=Y1,loosely dotted, semithick, blue] (0,1) -- (\Sp,1);
\draw[name path=Y2, loosely dotted, semithick, blue] (0,2) -- (\Sp,2);
\draw[name path=Y3, loosely dotted, semithick, blue] (0,3) -- (\Sp,3);
\draw[name path=Y4,loosely dotted, semithick, blue] (0,4) -- (\Sp,4);
\draw[name path=Y5,loosely dotted, semithick, blue] (0,5) -- (\Sp,5);
\draw[name path=Y6,loosely dotted, semithick, blue] (0,6) -- (\Sp,6);
\draw[name path=Y7,loosely dotted, semithick, blue] (0,7) -- (\Sp,7);

\foreach\Q in { 1,2,3, 4, 5, 6, 7}\foreach\P in { 1,2,3, 4, 5, 6, 7}\fill[blue](\P,\Q)circle(0.03);


\end{tikzpicture}
    \subcaption{}\label{b}
  \end{minipage}
  \caption{The dark gray region in Figure (a) represents $\bs{D}_q$, while the light gray region represents $\bs{D}_c$. In Figure (b), the green region represents $\bs{D}_q(5)$, and the red region represents $\bs{D}_c(5)$.}\label{PicDs}
\end{figure}

 Let $B\subseteq \vLa$.
Suppose that we are given an electron configuration $\bn_{B, \rm e}=(n_{B, \rm e}(x))_{ x\in B } \in \mathcal{N}_{B, \rm e}$. 
We say that the site $x\in B$ is in the $\ell$-ground state with respect to the electronic configuration $\bn_{B, \rm e}$, 
if there is an $\ell$ such that $n_{B, \rm e}(y)=\bs{g}_{\rm e}^{(\ell)}(y)$ holds for all $y\in U(x)$.
If $x$ is not in the ground state, then $x$ is said to be in an excited state with respect to $\bn_{B, \rm e}$.
Similarly, considering a phonon configuration denoted by $\bn_{B, \rm p}=(n_{B, \rm p}(x))_{ x\in B } \in \mathcal{N}_{B, \rm p}$, we define a site $x\in B$ to be in the ground state with respect to the phonon  configuration $\bn_{B, \rm p}$ if $n_{B, \rm p}(x)=0$. If $x$ is not in the ground state, we say that it is in an excited state with respect to $\bn_{B, \rm p}$. Using the above terminology, we can define the following: 
\begin{align}
\mathcal{E}_{B,  \rm e}(x)&=\left\{
\bn_{B,  \rm e} : \mbox{$x$ is in an excited state with respect to  $\bn_{B, \rm e}$}
\right\},\\
\mathcal{E}_{B,  \rm p}(x)&=\left\{
\bn_{B,  \rm p} : \mbox{$x$ is in an excited state with respect to $\bn_{B,  \rm p}$}
\right\}.
\end{align}
Furthermore, 
 for each $\bnd=(\bn_{\bs{D}_{c},  {\rm e}}, \bn_{\bs{D}_{c},  \rm p})\in \bE_{\bs{D}_c}$, we define
\begin{align}
\mathscr{E}_{\rm e}(\bne)&=\left\{(x, t) \in \bs{D}_c : 
\bn^{(t)}_{D^{(t)}_c,  \rm e}\in \mathcal{E}_{D^{(t)}_c,  \rm e}(x)\right\},\\
\mathscr{E}_{\rm p}(\bnp)&=\left\{(x, t) \in \bs{D}_c : \bn^{(t)}_{D^{(t)}_c,  \rm p}\in \mathcal{E}_{D^{(t)}_c,  \rm p}(x)
 \right\}.
\end{align}
The set $\mathscr{E}_{\rm e}(\bne)$ consists of  space-time lattice sites that are in excited states with respect to the electron configuration $\bne$. Similarly, $\mathscr{E}_{\rm p}(\bnp)$ consists of  space-time lattice sites that are in excited states with respect to the phonon configuration $\bnp$. It holds that $\mathscr{E}_{\rm e}(\bne)\cup \mathscr{E}_{\rm p}(\bnp)=[\bs{D}_c]$. However, note that $\mathscr{E}_{\rm e}(\bne)\cap \mathscr{E}_{\rm p}(\bnp)\ne \varnothing$ in general.

For a  given site $(x, t)\in \mathbb{L}_{\vLa}$, we define an operator $u_{\bs{\vphi}}(x, t)$ as follows: 
Firstly, we consider the case where $(x,t)\in [\bs{D}_q]$.
In this case, we need to introduce some symbols to specify $u_{\bs{\vphi}}(x, t)$.
 Let $\bs{m}$, $\pi$, $s_1,\dots,s_{|\bs{m}|}$ be as introduced in Subsection \ref{Sec3.1}. For $A\in\supp \bs{m}$, we set $\varTheta_A = \exp(\mathrm{i}\vartheta_A)$, where $\vartheta_A$ is defined in \eqref{DefTht}. 
 If $x\in\supp A$, then we define $\varTheta_{A,x}=\prod_{\sigma,\kappa:(x,\sigma,\kappa)\in A}\ex^{ \im \vartheta_{(x, \sigma, \kappa)}}$. We readily confirm that  $\varTheta_A=\prod_{x\in\supp A}\varTheta_{A,x}$ holds. If  $x\notin \supp A$,  then we simply set  $\varTheta_{A,x}=\mathbbm{1}$.
Next,  for each 
 $\omega^{(t)}=(\bs{\tau}^{(t)}, \bs{m}^{(t)})\in \varXi_{D_q^{(t)}}$, we define
\be
(\tilde{\varTheta}_{1, x}, \dots, \tilde{\varTheta}_{|\bs{m}^{(t)}|, x})=\pi(
\varTheta_{A_1, x}, \dots, \varTheta_{A_1, x}, \dots, \varTheta_{A_k, x}, \dots, \varTheta_{A_k, x}
).
\ee
Recall here that $\supp \bs{m}^{(t)}$ is represented as $\{A_1, \dots, A_k\}$.
With the above setup,   $u_{\bs{\vphi}}(x, t)$ is defined as
\be
u_{\bs{\vphi}}(x, t)=\tilde{\varTheta}_{1, x}(s_1) \dots \tilde{\varTheta}_{|\bs{m}^{(t)}|, x}(s_{|\bs{m}^{(t)}|})\, 
\ex^{-\tilde{\beta} \omega_0 N_x}\Ket{n^{(t)}_{D^{(t)},  \rm p}(x)} \Bra{n^{(t)}_{D^{(t)},  \rm p}(x)}, 
\ee
where $ \tilde{\varTheta}_{A_i, x}(s)=\ex^{-s \omega_0 N_x}\tilde{\varTheta}_{A_i, x}\ex^{s \omega_0 N_x}$ ($N_x=b_x^*b_x$).

Secondly, let us consider the case where $(x, t)\in [\bs{D}_c]$. In this case, we define $u_{\bs{\vphi}}(x, t)$ by
\begin{align}
u_{\bs{\vphi}}(x, t)=
\Ket{n^{(t)}_{D^{(t)},  \rm p}(x)} \Bra{n^{(t)}_{D^{(t)},  \rm p}(x)}.
\end{align}
Finally, if $(x, t)\in [\mathbb{T}_{\vLa}\setminus \bs{D}]$, we define $u_{\bs{\vphi}}(x, t)=P_{\varnothing}:=\ket{\varnothing}\! \bra{\varnothing}$. To summarize the above, we have
\begin{align}
u_{{\bphi}}(x, t)=
\begin{cases}
\tilde{\varTheta}_{1, x}(s_1) \cdots \tilde{\varTheta}_{|\bs{m}^{(t)}|, x}(s_{|\bs{m}^{(t)}|})
\, \ex^{-\tilde{\beta} \omega_0 N_x}\Ket{n^{(t)}_{D^{(t)},  \rm p}(x)} \Bra{n^{(t)}_{D^{(t)},  \rm p}(x)} &  (x, t)\in [\bs{D}_{q}]\\
\Ket{n^{(t)}_{D^{(t)},  \rm p}(x)} \Bra{n^{(t)}_{D^{(t)},  \rm p}(x)}  & (x, t)\in [\bs{D}_c]\\
P_{\varnothing} & (x, t)\notin [\bs{D}].
\end{cases}\label{Deflittleu}
\end{align}

Using the operator $u_{{\bphi}}(x, t)$ thus defined, we can factorize $S_{\rm p}(\bphi)$ in the following manner:
\be
S_{\rm p}(\bs{\vphi})=\prod_{x\in [\bs{D}]_{\rm S}} S_{\rm p}(\bs{\vphi}(x)),\quad
S_{\rm p}(\bs{\vphi}(x))=\Braket{\varnothing | \prod_{t\in \mathbb{M}} u_{\bs{\vphi}}(x, t)|\varnothing},  \label{DefS1}
\ee
where $\prod_{t\in \mathbb{M}} A_t=A_1A_2\cdots A_M$ denotes the time-ordered product.
Note that this factorization representation is advantageous in analyzing contour activities, which is the focus of this section.

The summation over excited configurations can be expressed as follows: 
\begin{align}
\sum_{\bn_{\bs{D}_c}\in \bE_{\bs{D}_c}} =\sum_{{X_{\rm e}, X_{\rm p} \subseteq \bs{D}_c}\atop{X_{\rm e} \cup X_{\rm p}=\bs{D}_c}}
\sum_{\substack{\bn_{\bs{D}_c}:  \\ \mathscr{E}_{\rm e}(\bne)=[X_{\rm e}] \\ \mathscr{E}_{\rm p}(\bnp)=[X_{\rm p}]}}.
\end{align}
Here, $\sum_{{X_{\rm e}, X_{\rm p} \subseteq \bs{D}c}\atop{X{\rm e} \cup X_{\rm p}=\bs{D}c}}$ denotes the  sum over collections of elementary cubes $X_{\rm e}$ and $X_{\rm p}$ that satisfy the following conditions (i) and (ii): (i) $X_{\rm e}, X_{\rm p} \subseteq \bs{D}_c$; (ii) $X_{\rm e} \cup X_{\rm p}=\bs{D}_c$.

Using this representation, we get
\begin{align}
&\sum_{\bn_{\bs{D}_{q},  \rm p}\in \bN_{\bs{D}_{q},  \rm p}} \sum_{\bn_{\bs{D}_c}\in \bE_{\bs{D}_c}} \ex^{-\tilde{\beta}  \bs{E}(\bn_{\bs{D}_c })} S_{\rm e}(\bs{\vphi}_{\rm e})S_{\rm p}(\bs{\vphi}_{\rm p})\no
=&\sum_{\bn_{\bs{D}_{q},  \rm p}\in \bN_{\bs{D}_{q},  \rm p}}  \sum_{{X_{\rm e}, X_{\rm p} \subseteq \bs{D}_c}\atop{X_{\rm e} \cup X_{\rm p}=\bs{D}_c}}
\sum_{\substack{\bn_{\bs{D}_c}:  \\ \mathscr{E}_{\rm e}(\bne)=[X_{\rm e}]  \\ \mathscr{E}_{\rm p}(\bnp)=[X_{\rm p}]}}\ex^{-\tilde{\beta}  \bs{E}(\bn_{\bs{D}_c })} S_{\rm e}(\bs{\vphi}_{\rm e})S_{\rm p}(\bs{\vphi}_{\rm p})\no
=&\sum_{{X_{\rm e}, X_{\rm p} \subseteq \bs{D}_c}\atop{X_{\rm e} \cup X_{\rm p}=\bs{D}_c}}
\sum_{\substack{\bn_{\bs{D}_c,  \rm e}:  \\ \mathscr{E}_{\rm e}(\bn_{\bs{D}_c,  \rm e})=[X_{\rm e}]}}
 \ex^{-\tilde{\beta}  E_{\rm e}(\bn_{X_{\rm e},  \rm e})} S_{\rm e}(\bs{\vphi}_{\rm e})Q_{X_{\rm p}}(\bs{\vphi}_{ \rm e}),
\end{align}
where  $Q_{X_{\rm p}}(\bs{\vphi}_{ \rm e})$ is given by 
\begin{align}
Q_{X_{\rm p}}(\bs{\vphi}_{ \rm e})=\sum_{\bn_{\bs{D}_q,  \rm p}\in \bN_{\bs{D}_{q},  \rm p}} \sum_{\substack{\bn_{\bs{D}_c,  \rm p}:\\
\mathscr{E}_{\rm p}(\bn_{\bs{D}_c,  \rm p})=[X_{\rm p}]
}}\ex^{-\tilde{\beta}  E_{\rm p}(\bn_{\bs{D}_c,  \rm p })} S_{\rm p}(\bs{\vphi}_{\rm p})\label{DefQ1}
\end{align}
with 
\be
E_{\rm p}(\bn_{\bs{D}_c,  \rm p })= \sum_{(x, t)\in [\bs{D}_{c}]} \omega_0 n_{D_{c}^{(t)}, \rm p}^{(t)}(x),\quad
E_{\rm e}(\bn_{\bs{D}_c, \rm e })=\sum_{t=1}^M \sum_{x\in D^{(t)} \setminus \overline{D}_q^{(t)}}\varPhi^{(\ell)}_{{\rm eff}, x}\left(\bn^{(t)}_{D_c^{(t)}, \rm e}\right).
\ee
It is straightforward to verify that the equality
$\bs{E}(\bn_{\bs{D}_c})=E_{\rm e}(\bn_{\bs{D}_c, \rm e })+E_{\rm p}(\bn_{\bs{D}_c,  \rm p})$ holds. 
Summing up the above, the contour activity of $Y$ can be expressed as follows:
\begin{align}
\rho(Y)=\sum_{\bs{D}_q \subseteq \bs{D}} 
\sum_{{X_{\rm e}, X_{\rm p} \subseteq \bs{D}_c}\atop{X_{\rm e} \cup X_{\rm p}=\bs{D}_c}}
\sum_{\substack{\bn_{\bs{D}_c,  \rm e}:  \\ \mathscr{E}_{\rm e}(\bn_{\bs{D}_c,  \rm e})=[X_{\rm e}]}}
\sum_{\bn_{\bs{D}_q,  \rm e}\in \bN_{\bs{D}_q,  \rm e}} \int_{\bs{\varXi}_{\bs{D}_q}} \mathscr{D}(\bome)
\, \ex^{-\tilde{\beta} E_{\rm e}(\bn_{\bs{D}_{c}, \rm e })} S_{\rm e}(\bs{\vphi}_{\rm e})Q_{X_{\rm p}}(\bs{\vphi}_{\rm e}).
\label{ActRep}
\end{align}
Note that we will prove that $Q_{X_{\rm p}}$ converges absolutely, which justifies the exchange of the sum and integral that we used to derive \eqref{ActRep}.

To obtain an upper bound for $\rho(Y)$, we initially estimate $Q_{X_{\rm p}}$.

\begin{Lemm}\label{QInq1}
One obtains the following:
\be
\left| 
Q_{X_{\rm p}}\left(\bs{\vphi}_{\rm e}\right)
\right| \le \ex^{-\tilde{\beta} \omega_0 |X_{\rm p}|}.
\ee
\end{Lemm}
\begin{proof}
Recall that we defined $N_x=b_x^*b_x$.
It is readily verified that the following operator inequality holds: 
$
\ex^{-\tilde{\beta} \omega_0 N_x} P_{\varnothing}^{\perp} \le \ex^{-\tilde{\beta} \omega_0 }P_{\varnothing}^{\perp}, 
$
where
$P^{\perp}_{\varnothing}=\mathbbm{1}-\ket{\varnothing}\!\bra{\varnothing}$.
From this,  we get the following:
\be
\|
\ex^{-\tilde{\beta} \omega_0 N_x} P_{\varnothing}^{\perp} 
\|\le \ex^{-\tilde{\beta} \omega_0 }.\label{NInq}
\ee
Note the following facts concerning  the harmonic oscillator:
\be
\sum_{n=0}^{\infty}\ex^{-\tilde{\beta} \omega_0 n}\ket{n}\!\bra{n}=\ex^{-\tilde{\beta} \omega_0 N_x}, \quad
\sum_{n=1}^{\infty}\ex^{-\tilde{\beta} \omega_0 n}\ket{n}\!\bra{n}=\ex^{-\tilde{\beta} \omega_0 N_x}P_{\varnothing}^{\perp}, 
\ee
where $\ket{n}$ is the eigenvector of the operator $N_x$ with the corresponding eigenvalue $n$, i.e., $N_x\ket{n}=n \ket{n}$. Utilizing the aforementioned facts, one can represent $Q_{X_{\rm p}}$ in the following manner:
\begin{align}
Q_{X_{\rm p}}\left(\bs{\vphi}_{\rm e}\right)
=&
\sum_{\bn_{\bs{D}_c, \rm p}\in \BbbN^{[X_{\rm p}]}}\sum_{\bn_{\bs{D}_{q},  \rm p}\in \bN_{\bs{D}_{q},  \rm p}}
\prod_{x\in [\bs{D}]_{\rm S}}  \ex^{-\tilde{\beta}E_{\rm p}(\bnp(x))}S_{\rm p}(\bphi(x))
\no
=&\prod_{x\in [\bs{D}]_{\rm S}}\Braket{
\varnothing| \prod_{ {t \in [\mathsf{S}_x(\bs{D}})]_{\rm T}}U_{\bphi}(x, t) |\varnothing
},  \label{QRep}
\end{align}
 where the operator $U_{\bphi}(x, t)$ is defined as follows:
\begin{align}
U_{{\bphi}}(x, t)=
\begin{cases}
\tilde{\varTheta}_{1, x}(s_1) \cdots \tilde{\varTheta}_{|\bs{m}^{(t)}|, x}(s_{|\bs{m}^{(t)}|})
\, \ex^{-\tilde{\beta} \omega_0 N_x} &  (x, t)\in [\bs{D}_{q}]\\
\ex^{-\tilde{\beta} \omega_0 N_x}P^{\perp}_{\varnothing}
& (x, t) \in [X_{\rm p}]\\
\ex^{-\tilde{\beta} \omega_0 N_x}  & (x, t)\in [\bs{D}_c\setminus X_{\rm p}]\\
P_{\varnothing} & (x, t)\notin [\bs{D}].
\end{cases}\label{DefU}
\end{align}
Additionally, the product $\prod_{t \in [\mathsf{S}_x(\bs{D})]_{\rm T}} A_t$ in \eqref{QRep}
  denotes the time-ordered product.\footnote{When $[\mathsf{S}_x(\bs{D})]_{\rm T}$ can be represented as 
$[\mathsf{S}_x(\bs{D})]_{\rm T}=\{t_1, t_2, \dots, t_n : t_1<t_2<\cdots<t_n\}$, then the time-ordered product is defined as
$\prod_{ {t \in [\mathsf{S}_x(\bs{D}})]_{\rm T}} A_T=A_{t_1}A_{t_2}\cdots A_{t_n}$.
}

The desired claim can be derived by utilizing the inequality \eqref{NInq}  and the fact that $\|U_{{\bphi}}(x, t)\|\le 1$ for $(x, t)\in [\bs{D}_{q}\cup (\bs{D}_c \setminus X_{\rm p})]$.
\end{proof}

\begin{Prop}\label{ActBd1}
Suppose that 
$\lambda\in \BbbR, \tilde{\beta}>0$ and $\gamma_Q\ge 0$ satisfy the following:
\be
(\ex-1) \tilde{\beta}|\lambda| \|t\|_{\gamma_Q} \le 1,\quad (\gamma_Q-1)R_0^{-d}-\gamma_{\rm e}-|e_{\rm e}|> 0. \label{ConvClC}
\ee
If the contour $Y$ does not wind around the torus  $\mathbb{T}_{\vLa}$, then the following holds:
\be
|\rho(Y)| \le \ex^{-(\tilde{\beta} e_{\rm e}+\gamma) |\supp Y|}, 
\ee
where $\gamma=\min\left\{
\min\{\gamma_{\rm e}, \omega_0\},\,  (\gamma_Q-1)R_0^{-d}-\gamma_{\rm e}-|e_{\rm e}|
\right\}-5 \log 2$.
\end{Prop}
\begin{proof}
Using Lemma \ref{QInq1}, we have
\begin{align}
&\left| \sum_{\bn_{\bs{D}_{q,  \rm p}}\in \bN_{\bs{D}_{q},  \rm p}} \sum_{\bn_{\bs{D}_c}\in \bE_{\bs{D}_c}} \ex^{-\tilde{\beta}  \bs{E}(\bn_{\bs{D}_c })} S_{\rm e}(\bs{\vphi}_{\rm e})S_{\rm p}\left(\bs{\vphi}_{\rm p}\right) \right|\no
\le & 
\sum_{{X_{\rm e}, X_{\rm p} \subseteq \bs{D}_c}\atop{X_{\rm e} \cup X_{\rm p}=\bs{D}_c}}
\sum_{\substack{\bn_{\bs{D}_c,  \rm e}:  \\ \mathscr{E}_{\rm e}(\bn_{\bs{D}_c,  \rm e})=[X_{\rm e}]}}
 \ex^{-\tilde{\beta}  
 E_{\rm e}(\bn_{\bs{D}_{c}, \rm e })
 } | S_{\rm e}(\bs{\vphi}_{\rm e})||Q_{X_{\rm p}}(\bs{\vphi}_{ \rm e})|
\no
\le &
\sum_{{X_{\rm e}, X_{\rm p} \subseteq \bs{D}_c}\atop{X_{\rm e} \cup X_{\rm p}=\bs{D}_c}}
\sum_{\substack{\bn_{\bs{D}_c,  \rm e}:  \\ \mathscr{E}_{\rm e}(\bn_{\bs{D}_c,  \rm e})=[X_{\rm e}]}}
 \ex^{-\tilde{\beta} E_{\rm e}(\bn_{\bs{D}_{c}, \rm e })} | S_{\rm e}(\bs{\vphi}_{ \rm e})|\, \ex^{-\tilde{\beta} \omega_0 |X_{\rm p}|}\no
= & \sum_{{X_{\rm e}, X_{\rm p} \subseteq \bs{D}_c}\atop{X_{\rm e} \cup X_{\rm p}=\bs{D}_c}}
\sum_{\substack{\bn_{\bs{D}_c,  \rm e}:  \\ \mathscr{E}_{\rm e}(\bn_{\bs{D}_c,  \rm e})=[X_{\rm e}]}}\ex^{-\tilde{\beta}  E_{\rm eff}(\bn_{\bs{D}_c,  \rm e }, X_{\rm p})} |S_{\rm e}(\bs{\vphi}_{\rm e})|,
\end{align}
where
\be
E_{\rm eff}(\bn_{\bs{D}_c,  \rm e }, X_{\rm p})=E_{\rm e}(\bn_{\bs{D}_c, \rm e })+\omega_0 \left|X_{\rm p}\right|. \label{DefEeff}
\ee
Therefore, one obtains 
\begin{align}
|\rho(Y)|\le 
\sum_{\bs{D}_q \subseteq \bs{D}} 
\sum_{{X_{\rm e}, X_{\rm p} \subseteq \bs{D}_c}\atop{X_{\rm e} \cup X_{\rm p}=\bs{D}_c}}
\sum_{\substack{\bn_{\bs{D}_c,  \rm e}:  \\ \mathscr{E}_{\rm e}(\bn_{\bs{D}_c,  \rm e})=[X_{\rm e}]}}
\sum_{\bn_{\bs{D}_q,  \rm e}\in \bN_{\bs{D}_q,  \rm e}} \int_{\bs{\varXi}_{\bs{D}_q}} \mathscr{D}(\bome)\ex^{-\tilde{\beta}  E_{\rm eff}(\bn_{\bs{D}_c,  \rm e }, X_{\rm p})} 
|S_{\rm e}(\bs{\vphi}_{\rm e})|.
\end{align}
We shall proceed to derive an upper bound for the right-hand side of this inequality.
Going back to the definition, we obtain the following:
\begin{align}
 \int_{\bs{\varXi}_{\bs{D}_q}} \mathscr{D}(\bome)
|S_{\rm e}(\bs{\vphi}_{\rm e})|
=
\prod_{t=1}^M\left|\Braket{\bs{g}^{(*)}_{\vLa\setminus D^{(t)}, \rm e}\times   \bn_{D^{(t)}, \rm e}^{(t)}| \mathcal{T}_{D_q^{(t+1)}, \rm e}\left(\bn^{(t+1)}_{\overline{\partial D}_q^{(t+1)}, \rm e}\right) |\bs{g}^{(*)}_{\vLa\setminus D^{(t+1)}, \rm e}\times  \bn_{D^{(t+1)}, \rm e}^{(t+1)}}\right|, 
\end{align}
where
\be
\mathcal{T}_{B, \rm e}(\bn_{\overline{\partial B}, \rm e})=\int_{\varXi_B} \mathscr{D}(\omega) \mathcal{T}_{B, \rm e}(\omega, \bn_{\overline{\partial B}, \rm e}).
\ee
By using  \cite[Lemma 4.2]{Borgs1996} and the fact that $|\overline{B}| \le R_0^d |B|$, we can estimate $\mathcal{T}_{B, \rm e}(\bn_{\overline{\partial B}, \rm e})$ as follows:
\begin{align}
\left\| \mathcal{T}_{B, \rm e} (\bn_{\overline{\partial B}}) \right\| &\le 
\ex^{-\tilde{\beta} e_{\rm e} |\overline{B}|}
\sum_{{\mathcal{A}=\{A_1, \dots, A_k\}}\atop{\cup_{j=1}^k A_j=B}} \prod_{A\in \mathcal{A}} \left[
\sum_{m_A=1}^{\infty} \frac{(\tilde{\beta} |\lambda| |t_A|)^{m_A}}{m_A!}
\right]\no
& \le \ex^{-\tilde{\beta} e_{\rm e} |\overline{B}|} \ex^{-(\gamma_Q-1)|B|}\no
&\le \ex^{-\tilde{\beta} e_{\rm e} |\overline{B}|} \ex^{-(\gamma_Q-1)|\overline{B}|R_0^{-d}}.
\end{align}
Setting $\tilde{\gamma}$ as  $\tilde{\gamma}=\min\left\{
\min\{\gamma_{\rm e}, \omega_0\},\,  (\gamma_Q-1)R_0^{-d}-\gamma_{\rm e}-|e_{\rm e}|
\right\}$ and using Lemma \ref{EstGa}, we obtain the following inequality:
\be
E_{\rm eff}(\bn_{\bs{D}_c,  \rm e }, X_{\rm p})+e_{\rm e}|\overline{\bs{D}}_q|+(\gamma_Q-1)R_0^{-d} |\overline{\bs{D}}_q|\ge (e_{\rm e}+\tilde{\gamma})|\bs{D}|,   \label{TilGamma}
\ee
which implies  that 
\begin{align}
|\rho(Y)|&\le 
\sum_{\bs{D}_q \subseteq \bs{D}} 
\sum_{{X_{\rm e}, X_{\rm p} \subseteq \bs{D}_c}\atop{X_{\rm e} \cup X_{\rm p}=\bs{D}_c}}
\sum_{\substack{\bn_{\bs{D}_c,  \rm e}:  \\ \mathscr{E}_{\rm e}(\bn_{\bs{D}_c,  \rm e})=[X_{\rm e}]}}
\sum_{\bn_{\bs{D}_q,  \rm e}\in \bN_{\bs{D}_q,  \rm e}}
\ex^{-\tilde{\beta} 
E_{\rm eff}(\bn_{\bs{D}_c,  \rm e }, X_{\rm p})
}
\ex^{-\tilde{\beta}e_{\rm e}|\overline{\bs{D}}_q|}
\ex^{-\tilde{\beta}(\gamma_Q-1)R_0^{-d} |\overline{\bs{D}}_q|}\no
&\le 
\sum_{\bs{D}_q \subseteq \bs{D}} 
\sum_{{X_{\rm e}, X_{\rm p} \subseteq \bs{D}_c}\atop{X_{\rm e} \cup X_{\rm p}=\bs{D}_c}}
\sum_{\substack{\bn_{\bs{D}_c,  \rm e}:  \\ \mathscr{E}_{\rm e}(\bn_{\bs{D}_c,  \rm e})=[X_{\rm e}]}}
\sum_{\bn_{\bs{D}_q,  \rm e}\in \bN_{\bs{D}_q,  \rm e}}
\ex^{-(\tilde{\beta}e_{\rm e}+\tilde{\gamma}) |\supp Y|}.
\end{align}
Additionally, utilizing the the following crude estimates: 
\begin{align}
\sum_{\substack{\bn_{\bs{D}_c,  \rm e}:  \\ \mathscr{E}_{\rm e}(\bn_{\bs{D}_c,  \rm e})=[X_{\rm e}]}}\sum_{\bn_{\bs{D}_q,  \rm e}\in \bN_{\bs{D}_q,  \rm e}}1
\le 3^{|\bs{D}_c|} 4^{|\bs{D}_{q}|} \le \ex^{ (\log4) |\supp Y|},\qquad
\sum_{\bs{D}_q \subseteq \bs{D}} 
\sum_{{X_{\rm e}, X_{\rm p} \subseteq \bs{D}_c}\atop{X_{\rm e} \cup X_{\rm p}=\bs{D}_c}}1
\le \ex^{ (\log 8) |\supp Y|}, 
\end{align}
we can estimate the right-hand side of the above inequality and obtain the desired claim for $\rho(Y)$.
\end{proof}

\subsection{The case where $Y$ winds around $\mathbb{T}_{\vLa}$}\label{Sect4.3}
In this subsection,  we consider the case where $\bs{D}=\supp Y$ winds around $\mathbb{T}_{\vLa}$. See Figure \ref{Grph5}.
In the analysis of contour activities in Subsection \ref{Sec4.2}, concerning the operators $U_{\bphi}(x, t)$, only the vacuum expectations of their products appeared.
However, when $\bs{D}=\supp Y$ winds around $\mathbb{T}_{\vLa}$, the evaluation of contour activities becomes more intricate as the vacuum expectations and traces of the products of the operators $U_{\bphi}(x, t)$ simultaneously come into play.

\begin{figure}[htbp]

    \centering
     \includegraphics[scale=0.5]{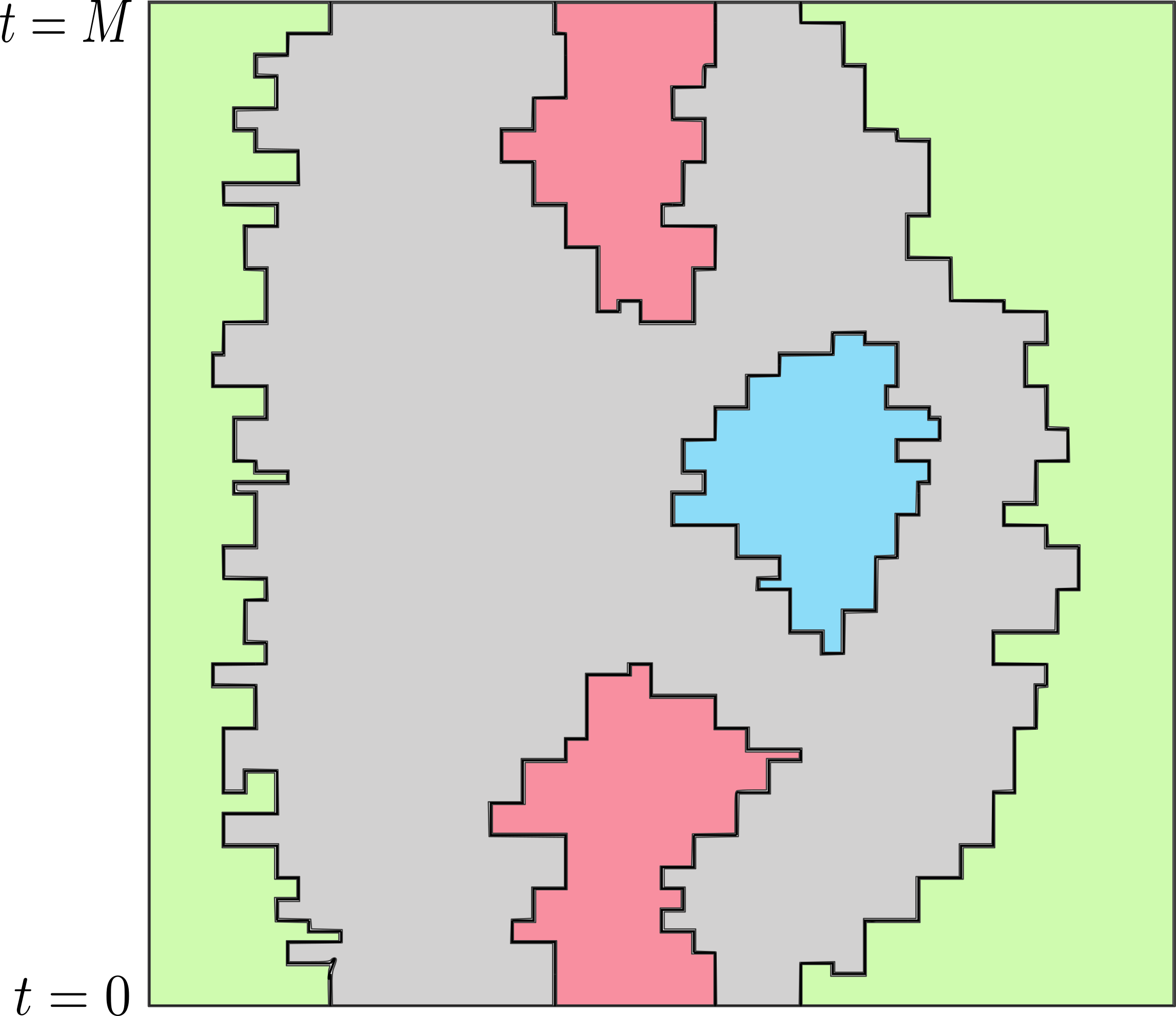}
    \caption{Image depicting a contour  winding around $\mathbb{T}_{\vLa}$. (To be precise,  
    section of the support of a contour  winding around $\mathbb{T}_{\vLa}$ on a specific plane defined by $x_i = \rm const.$)  }\label{Grph5}

\end{figure}

We define
\be
W(\bs{D})=\left\{x\in \vLa : [\mathsf{S}_x(\bs{D})]_{\rm T}=\mathbb{M}\right\}\quad (\mathbb{M}=\{1, \dots, M\}).
\ee
If $W(\bs{D})\neq \varnothing$, we say that $\bs{D}$ contains a {\it simple loop}. Furthermore, for $x\in W(\bs{D})$, we refer to $\bigcup_{t\in \mathbb{M}} C(x, t)\subseteq \bs{D}$ as a simple loop. See Figure \ref{SLoop}.

\begin{figure}[htbp]

    \centering
     \begin{tikzpicture}[scale=0.7]

\coordinate(X11)at(\Ct+1,\Ct-1.5);
\coordinate(X12)at(\Ct+1,\Ct-1);
\coordinate(X13)at(\Ct+2,\Ct-1);
\coordinate(X14)at(\Ct+2,\Ct-1.5);
\filldraw[fill=red,opacity=0.6](X11)--(X12)--(X13)--(X14)--cycle;
\coordinate(X21)at(\Ct+2,\Ct-1.5);
\coordinate(X22)at(\Ct+2,\Ct-1);
\coordinate(X23)at(\Ct+3,\Ct-1);
\coordinate(X24)at(\Ct+3,\Ct-1.5);
\filldraw[fill=red,opacity=0.6](X21)--(X22)--(X23)--(X24)--cycle;

\coordinate(Z11)at(\Ct+1,\Ct-1);
\coordinate(Z12)at(\Ct+1,\Ct);
\coordinate(Z13)at(\Ct+2,\Ct);
\coordinate(Z14)at(\Ct+2,\Ct-1);
\filldraw[fill=red,opacity=0.6](Z11)--(Z12)--(Z13)--(Z14)--cycle;
\coordinate(Z21)at(\Ct+2,\Ct-1);
\coordinate(Z22)at(\Ct+2,\Ct);
\coordinate(Z23)at(\Ct+3,\Ct);
\coordinate(Z24)at(\Ct+3,\Ct-1);
\filldraw[fill=red,opacity=0.6](Z21)--(Z22)--(Z23)--(Z24)--cycle;

\coordinate(Z31)at(\Ct,\Ct-1);
\coordinate(Z32)at(\Ct,\Ct);
\coordinate(Z33)at(\Ct+1,\Ct);
\coordinate(Z34)at(\Ct+1,\Ct-1);
\filldraw[fill=lightgray](Z31)--(Z32)--(Z33)--(Z34)--cycle;

\coordinate(A11)at(\Ct,\Ct);
\coordinate(A12)at(\Ct,\Ct+1);
\coordinate(A13)at(\Ct+1,\Ct+1);
\coordinate(A14)at(\Ct+1,\Ct);
\filldraw[fill=lightgray](A11)--(A12)--(A13)--(A14)--cycle;
\coordinate(A21)at(\Ct+1,\Ct);
\coordinate(A22)at(\Ct+1,\Ct+1);
\coordinate(A23)at(\Ct+2,\Ct+1);
\coordinate(A24)at(\Ct+2,\Ct);
\filldraw[fill=red,opacity=0.6](A21)--(A22)--(A23)--(A24)--cycle;
\coordinate(A31)at(\Ct+2,\Ct);
\coordinate(A32)at(\Ct+2,\Ct+1);
\coordinate(A33)at(\Ct+3,\Ct+1);
\coordinate(A34)at(\Ct+3,\Ct);
\filldraw[fill=red,opacity=0.6](A31)--(A32)--(A33)--(A34)--cycle;
\coordinate(A41)at(\Ct+2+1,\Ct);
\coordinate(A42)at(\Ct+2+1,\Ct+1);
\coordinate(A43)at(\Ct+3+1,\Ct+1);
\coordinate(A44)at(\Ct+3+1,\Ct);
\filldraw[fill=lightgray](A41)--(A42)--(A43)--(A44)--cycle;

\coordinate(B11)at(\Ct,\Ct+1);
\coordinate(B12)at(\Ct,\Ct+2);
\coordinate(B13)at(\Ct+1,\Ct+2);
\coordinate(B14)at(\Ct+1,\Ct+1);
\filldraw[fill=lightgray](B11)--(B12)--(B13)--(B14)--cycle;
\coordinate(B21)at(\Ct+1,\Ct+1);
\coordinate(B22)at(\Ct+1,\Ct+2);
\coordinate(B23)at(\Ct+2,\Ct+2);
\coordinate(B24)at(\Ct+2,\Ct+1);
\filldraw[fill=red,opacity=0.6](B21)--(B22)--(B23)--(B24)--cycle;
\coordinate(B31)at(\Ct+4,\Ct+1);
\coordinate(B32)at(\Ct+4,\Ct+2);
\coordinate(B33)at(\Ct+5,\Ct+2);
\coordinate(B34)at(\Ct+5,\Ct+1);
\filldraw[fill=lightgray](B31)--(B32)--(B33)--(B34)--cycle;
\coordinate(B41)at(\Ct+3,\Ct+1);
\coordinate(B42)at(\Ct+3,\Ct+2);
\coordinate(B43)at(\Ct+4,\Ct+2);
\coordinate(B44)at(\Ct+4,\Ct+1);
\filldraw[fill=lightgray](B41)--(B42)--(B43)--(B44)--cycle;
\coordinate(B51)at(\Ct+2,\Ct+1);
\coordinate(B52)at(\Ct+2,\Ct+2);
\coordinate(B53)at(\Ct+3,\Ct+2);
\coordinate(B54)at(\Ct+3,\Ct+1);
\filldraw[fill=red,opacity=0.6](B51)--(B52)--(B53)--(B54)--cycle;

\coordinate(C11)at(\Ct-1,\Ct+2);
\coordinate(C12)at(\Ct-1,\Ct+3);
\coordinate(C13)at(\Ct,\Ct+3);
\coordinate(C14)at(\Ct,\Ct+2);
\filldraw[fill=lightgray](C11)--(C12)--(C13)--(C14)--cycle;
\coordinate(C21)at(\Ct,\Ct+2);
\coordinate(C22)at(\Ct,\Ct+3);
\coordinate(C23)at(\Ct+1,\Ct+3);
\coordinate(C24)at(\Ct+1,\Ct+2);
\filldraw[fill=lightgray](C21)--(C22)--(C23)--(C24)--cycle;
\coordinate(C31)at(\Ct+3,\Ct+2);
\coordinate(C32)at(\Ct+3,\Ct+3);
\coordinate(C33)at(\Ct+4,\Ct+3);
\coordinate(C34)at(\Ct+4,\Ct+2);
\filldraw[fill=lightgray](C31)--(C32)--(C33)--(C34)--cycle;
\coordinate(C41)at(\Ct+3+1,\Ct+2);
\coordinate(C42)at(\Ct+3+1,\Ct+3);
\coordinate(C43)at(\Ct+4+1,\Ct+3);
\coordinate(C44)at(\Ct+4+1,\Ct+2);
\filldraw[fill=lightgray](C41)--(C42)--(C43)--(C44)--cycle;
\coordinate(C51)at(\Ct+1,\Ct+2);
\coordinate(C52)at(\Ct+1,\Ct+3);
\coordinate(C53)at(\Ct+2,\Ct+3);
\coordinate(C54)at(\Ct+2,\Ct+2);
\filldraw[fill=red,opacity=0.6](C51)--(C52)--(C53)--(C54)--cycle;
\coordinate(C61)at(\Ct+2,\Ct+2);
\coordinate(C62)at(\Ct+2,\Ct+3);
\coordinate(C63)at(\Ct+3,\Ct+3);
\coordinate(C64)at(\Ct+3,\Ct+2);
\filldraw[fill=red,opacity=0.6](C61)--(C62)--(C63)--(C64)--cycle;

\coordinate(D11)at(\Ct-1,\Ct+3);
\coordinate(D12)at(\Ct-1,\Ct+4);
\coordinate(D13)at(\Ct,\Ct+4);
\coordinate(D14)at(\Ct,\Ct+3);
\filldraw[fill=lightgray](D11)--(D12)--(D13)--(D14)--cycle;
\coordinate(D21)at(\Ct+4-2,\Ct+3);
\coordinate(D22)at(\Ct+4-2,\Ct+4);
\coordinate(D23)at(\Ct+5-2,\Ct+4);
\coordinate(D24)at(\Ct+5-2,\Ct+3);
\filldraw[fill=red,opacity=0.6](D21)--(D22)--(D23)--(D24)--cycle;

\coordinate(D31)at(\Ct+1,\Ct+3);
\coordinate(D32)at(\Ct+1,\Ct+4);
\coordinate(D33)at(\Ct+2,\Ct+4);
\coordinate(D34)at(\Ct+2,\Ct+3);
\filldraw[fill=red,opacity=0.6](D31)--(D32)--(D33)--(D34)--cycle;
\coordinate(D41)at(\Ct,\Ct+3);
\coordinate(D42)at(\Ct,\Ct+4);
\coordinate(D43)at(\Ct+1,\Ct+4);
\coordinate(D44)at(\Ct+1,\Ct+3);
\filldraw[fill=lightgray](D41)--(D42)--(D43)--(D44)--cycle;

\coordinate(D51)at(\Ct+4-1,\Ct+3);
\coordinate(D52)at(\Ct+4-1,\Ct+4);
\coordinate(D53)at(\Ct+5-1,\Ct+4);
\coordinate(D54)at(\Ct+5-1,\Ct+3);
\filldraw[fill=lightgray](D51)--(D52)--(D53)--(D54)--cycle;

\coordinate(E21)at(\Ct+1,\Ct+3+1);
\coordinate(E22)at(\Ct+1,\Ct+4+0.5);
\coordinate(E23)at(\Ct+2,\Ct+4+0.5);
\coordinate(E24)at(\Ct+2,\Ct+3+1);
\filldraw[fill=red,opacity=0.6](E21)--(E22)--(E23)--(E24)--cycle;
\coordinate(E41)at(\Ct+2,\Ct+3+1);
\coordinate(E42)at(\Ct+2,\Ct+4+0.5);
\coordinate(E43)at(\Ct+3,\Ct+4+0.5);
\coordinate(E44)at(\Ct+3,\Ct+3+1);
\filldraw[fill=red,opacity=0.6](E41)--(E42)--(E43)--(E44)--cycle;




\draw[->,>=stealth,semithick] (0, 1)--(\Ti, 1)node[right]{$x$}; 
\draw[->,>=stealth,semithick] (0,1)--(0,\Sp)node[right]{$t$}; 
\draw[-, semithick] (0, 7)node[left]{$t=M$}--(\Ti, 7); 

\draw[name path=X1,loosely dotted,  semithick, blue] (1,1) -- (1,\Ti-1);
\draw[name path=X2,loosely dotted, semithick, blue] (2,1) -- (2,\Ti-1);
\draw[name path=X3,loosely dotted, semithick, blue] (3,1) -- (3,\Ti-1);
\draw[name path=X4,loosely dotted, semithick, blue] (5,1) -- (5,\Ti-1);
\draw[name path=X5,loosely dotted, semithick, blue] (4,1) -- (4,\Ti-1);
\draw[name path=X6,loosely dotted, semithick, blue] (6,1) -- (6,\Ti-1);
\draw[name path=X7,loosely dotted, semithick, blue] (7,1) -- (7,\Ti-1);
\draw[name path=Y1,loosely dotted, semithick, blue] (0,1) -- (\Sp,1);
\draw[name path=Y2, loosely dotted, semithick, blue] (0,2) -- (\Sp,2);
\draw[name path=Y3, loosely dotted, semithick, blue] (0,3) -- (\Sp,3);
\draw[name path=Y4,loosely dotted, semithick, blue] (0,4) -- (\Sp,4);
\draw[name path=Y5,loosely dotted, semithick, blue] (0,5) -- (\Sp,5);
\draw[name path=Y6,loosely dotted, semithick, blue] (0,6) -- (\Sp,6);
\draw[name path=Y7,loosely dotted, semithick, blue] (0,7) -- (\Sp,7);

\foreach\Q in { 1,2,3, 4, 5, 6, 7}\foreach\P in { 1,2,3, 4, 5, 6, 7}\fill[blue](\P,\Q)circle(0.03);


\end{tikzpicture}
    \caption{
Image of a contour winding around $\mathbb{T}_{\vLa}$. The red region represents two simple loops.
     }\label{SLoop}

\end{figure}
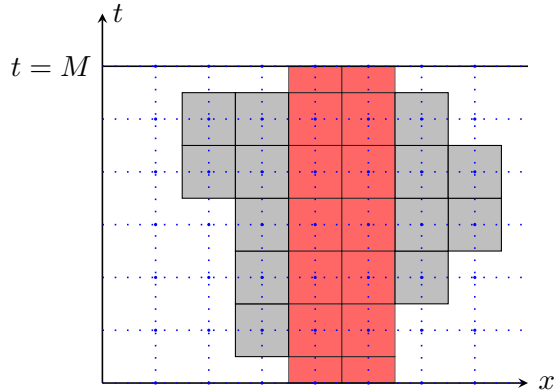
In the case where $\bs{D}$ has a simple loop, we must be careful in estimating contour activities, as seen below.

\begin{Lemm}\label{QwithL}
Let $z=(1-\ex^{-\tilde{\beta} \omega_0})^{-1}$ and $N=|W(\bs{D})|$.
Then, we have
\be
\left| 
Q_{X_{\rm p}}\left(\bs{\vphi}_{\rm e}\right)
\right| \le \ex^{-\tilde{\beta} \omega_0 |X_{\rm p}|}z^N.
\ee
\end{Lemm}
\begin{proof}
Suppose that $\bphi\in \bs{\varPi}_Y$ is chosen in such a way $W(\bs{D})$ has a simple loop.
To simplify matters, let us assume that the portion of $\bs{D}$ with all simple loops removed does not intersect with the hyperplanes represented by $t=0$ and $t=\beta$ in the space-time\footnote{In the general case, it is possible to prove the claim by suitably translating certain portions of the proof in the time direction.}.
Again, in this case, $S_{\rm p}(\bphi)$ can be decomposed as follows:
\be
S_{\rm p}(\bs{\vphi})=\prod_{x\in [\bs{D}]_{\rm S}} S_{\rm p}(\bs{\vphi}(x)),
\ee
where
\be
S_{\rm p}(\bs{\vphi}(x))=\Tr\left[\prod_{t\in \mathbb{M}} u_{\bs{\vphi}}(x, t)\right].
\ee
Compare the expression obtained here with that of \eqref{DefS1}.
Recalling the definition of $Q_{X_{\rm p}}$ in \eqref{DefQ1}, $Q_{X_{\rm p}}\left(\bs{\vphi}_{\rm e}\right)$ can be expressed as follows:
\begin{align}
Q_{X_{\rm p}}\left(\bs{\vphi}_{ \rm e}\right)
=&
\sum_{\bn_{\bs{D}_c,   \rm p}\in \BbbN^{[X_{\rm p}]}}\sum_{\bn_{\bs{D}_{q},  \rm p}\in \bN_{\bs{D}_{q},  \rm p}}
 \ex^{-\tilde{\beta}E_{\rm p}(\bnp)}
\prod_{x\in [\bs{D}]_{\rm S}} S_{\rm p}(\bphi(x))\no
=&R_{[\bs{D}]_{\rm S}\setminus W(\bs{D})} R_{W(\bs{D})},
\end{align}
where $R_{W(\bs{D})}$ and $R_{[\bs{D}]_{\rm S}\setminus W(\bs{D})}$ are given by
\begin{align}
 R_{W(\bs{D})}=\prod_{x\in  W(\bs{D})} \Tr\left[\prod_{ {t  \in \mathbb{M}}}U_{\bphi}(x, t)\right],\ \
 R_{[\bs{D}]_{\rm S}\setminus W(\bs{D})} 
 =\prod_{x\in [\bs{D}]_{\rm S}\setminus W(\bs{D})}\Tr\left[
 \prod_{ {t \in  \mathbb{M}}}U_{\bphi}(x, t)\right]. \label{ProdRR}
\end{align}
Here, reall that $U_{\bphi}(x, t)$ is defined by \eqref{DefU}.

Regarding $R_{[\bs{D}]_{\rm S}\setminus W(\bs{D})}$, it can be expressed as follows:
\be
R_{[\bs{D}]_{\rm S}\setminus W(\bs{D})}=
\prod_{x\in [\bs{D}]_{\rm S}\setminus W(\bs D)}\Braket{
\varnothing| \prod_{ {t \in [\mathsf{S}_x(\bs{D}})]_{\rm T}}U_{\bphi}(x, t) |\varnothing
}.
\ee
Therefore, using a similar method to the proof of Lemma \ref{QInq1}, we can obtain the following upper bound:
\be
\left|R_{[\bs{D}]_{\rm S}\setminus W(\bs{D})}\right|\le \ex^{-\tilde{\beta} \omega_0 |X_{\rm p}\setminus  \mathbb{T}_{W(\bs{D})}|}. \label{EstR1}
\ee

Next, we aim to obtain an upper bound for $|R_{W(\bs{D})}|$. For this purpose, it suffices to provide an upper bound for  $\left|\Tr\left[
\prod_{ {t \in \mathbb{M}}}U_{\bphi}(x, t)\right]\right|$ for each $x\in W(\bs{D})$.

Let $x \in W(\bs{D})$. If $\mathsf{S}_x(X_{\rm p})=\varnothing$, then the product of operators $\prod_{ {t \in \mathbb{M}}}U_{\bphi}(x, t)$ does not contain $\ex^{-\tilde{\beta} \omega_0 N_x}P^{\perp}_{\varnothing}$. Therefore, $\Tr\left[\prod_{ {t \in \mathbb{M}}}U_{\bphi}(x, t)\right]$ can be expressed in the form of the right-hand side of \eqref{CorF2}.
Therefore, by using the inequality \eqref{TherCorr}, we obtain
\be
\left|\Tr\left[
 \prod_{ {t \in  \mathbb{M}}}U_{\bphi}(x, t)\right]\right| \le (1-\ex^{-\beta \omega_0})^{-1} \le z.
\ee

Let us consider the case where $\mathsf{S}_x(X_{\rm p})\ne\varnothing$. 
We employ the notation $Z_1$ to denote operators that take the form $\tilde{\varTheta}_{1,x}(s_1) \cdots \tilde{\varTheta}_{|\bs{m}^{(t)}|, x}(s_{|\bs{m}^{(t)}|})\, \ex^{-\tilde{\beta} \omega_0 N_x}$, or alternatively, the operator $\ex^{-\tilde{\beta} \omega_0 N_x}$. Similarly, we use $Z_2$ to refer to the operator $\ex^{- \tilde{\beta} \omega_0 N_x}P^{\perp}_{\varnothing}$.
 Consequently, the product $\prod_{ {t \in \mathbb{M}}}U_{\bphi}(x, t)$ contains $n_x=|\mathsf{S}_x(X_{\rm p})|$ operators of the $Z_2$ type and $M-n_x$ operators of the $Z_1$ type,  denoted by $Z_{1}^{(1)}, \dots, Z_1^{(M-n_x)}$.
Hence, by repeatedly applying  the cyclicity of the trace and the well-known inequality\footnote{See, for example, \cite{Conway1985}.}:
\be
 |\Tr[AB]| \le \Tr[|AB|] \le \|A\| \Tr[|B|], 
 \ee
we obtain
 \be
 \left|\Tr\left[
 \prod_{ {t \in  \mathbb{M}}}U_{\bphi}(x, t)\right]\right| \le
 \left[\prod_{i=1}^{M-n_x}\|Z_1^{(i)}\|\right]\|Z_2\|^{n_x-1} \Tr[Z_2]
 \le  \frac{\ex^{-\tilde{\beta} \omega_0 n_x}}{1-\ex^{-\tilde{\beta} \omega_0}}=\ex^{-\tilde{\beta} \omega_0 n_x} z, 
 \ee
where, in the second inequality, we have utilized the facts that $\|Z_1^{(i)}\|\le 1$, $\|Z_2\| \le \ex^{-\tilde{\beta} \omega_0}$, and $\Tr[Z_2]=\ex^{-\tilde{\beta} \omega_0}/(1-\ex^{-\tilde{\beta} \omega_0})$.

 To summarize, when we let $k=\left|\left[\mathbb{T}_{W(\bs{D})} \cap X_{\rm p}\right]_{\rm S} \right|$, we have
\be
|R_{W(\bs{D})}| \le 
z^k \ex^{-\tilde{\beta} \omega_0|X_{\rm p} \cap \mathbb{T}_{W(\bs{D})} |} z^{N-k}
= z^N \ex^{-\tilde{\beta} \omega_0|X_{\rm p} \cap \mathbb{T}_{W(\bs{D})} |}. \label{EstR2}
\ee
Combining  \eqref{EstR1} and \eqref{EstR2} leads to the desired assertion.
\end{proof}

\begin{Prop}\label{RhoEst2}
Assume that $\lambda\in \BbbR, \tilde{\beta}>0$ and $\gamma_Q\ge 0$ satisfy \eqref{ConvClC}.  Furthermore, assume that $\tilde{\beta} \ge \beta_0/2$. If the contour $Y$ winds around $\mathbb{T}_{\vLa}$, then the following holds: 
\be
|\rho(Y)| \le \ex^{-(\tilde{\beta} e_{\rm e}+\tilde{\beta} c +\gamma) |\supp Y|}.
\ee
Here, $\gamma$ is defined in Proposition \ref{ActBd1}, and $c$ is given by $c=-\omega_0f(\omega_0\beta_0/2)$, where $f(x)=-x^{-1}\log (1-\ex^{-x})$.
\end{Prop}

\begin{proof}
By using Lemma \ref{QwithL}, we have
\begin{align}
&\left| \sum_{\bn_{\bs{D}_c}\in \bE_{\bs{D}_c}} \ex^{-\tilde{\beta}  \bs{E}(\bn_{\bs{D}_c })} S_{\rm e}(\bs{\vphi}_{\rm e})S_{\rm p}\left(\bs{\vphi}_{\rm p}\right) \right|\no
\le & 
\sum_{\substack{X_{\rm e}, X_{\rm p} \subseteq \bs{D}_c\\ X_{\rm e} \cup X_{\rm p}=\bs{D}_c
 }}
\sum_{\substack{\bn_{\bs{D}_c}:  \\ \mathscr{E}_{\rm e}(\bne)=[X_{\rm e}] \\ \mathscr{E}_{\rm p}(\bnp)=[X_{\rm p}]}}
 \ex^{-\tilde{\beta}   E_{\rm e}(\bn_{\bs{D}_c,  \rm e})} | S_{\rm e}(\bs{\vphi}_{\rm e})||Q_{X_{\rm p}}(\bs{\vphi}_{ \rm e})|
\no
\le & z^N \sum_{\substack{X_{\rm e}, X_{\rm p} \subseteq \bs{D}_c\\ X_{\rm e} \cup X_{\rm p}=\bs{D}_c
 }}
\sum_{\substack{\bn_{\bs{D}_c,  \rm e}:  \\ \mathscr{E}_{\rm e}(\bn_{\bs{D}_c,  \rm e})=[X_{\rm e}]}}\ex^{-\tilde{\beta}  E_{\rm eff}(\bn_{\bs{D}_{c}, \rm e }, X_{\rm p})} |S_{\rm e}(\bs{\vphi}_{\rm e})|.  \label{nDcSum}
\end{align}
Here, recall that $E_{\rm eff}$ is defined by \eqref{DefEeff}.

For $x > 0$, the function $f(x)$ is a monotonically decreasing positive function, diverging at $x=0$ and approaching a limit of zero as $x$ approaches infinity.
Hence, for $\beta>\beta_0/2$, one obtains
\be
\ex^{N\log z}=\ex^{N\tilde{\beta} \omega_0 \frac{\log z }{\tilde{\beta} \omega_0}}\le \ex^{|\supp Y| \tilde{\beta}\omega_0f(\omega_0\tilde{\beta})}\le 
\ex^{|\supp Y| \tilde{\beta}\omega_0f(\omega_0\beta_0/2)},  
\ee
which implies that 

\begin{align}
\mbox{the R.H.S. of \eqref{nDcSum}}
\le 
\ex^{|\supp Y| \tilde{\beta}\omega_0f(\omega_0\beta_0)}
 \sum_{\substack{X_{\rm e}, X_{\rm p} \subseteq \bs{D}_c\\ X_{\rm e} \cup X_{\rm p}=\bs{D}_c
 }}
\sum_{\substack{ \bn_{\bs{D}_c,  \rm e}:  \\ \mathscr{E}_{\rm e}(\bn_{\bs{D}_c,  \rm e})=[X_{\rm e}]}} \ex^{-\tilde{\beta}  E_{\rm eff}(\bn_{\bs{D}_c })} |S_{\rm e}(\bs{\vphi}_{\rm e})|.\label{SeUpBd}
\end{align}
A commensurate upper bound for the term on the right-hand side of inequality \eqref{SeUpBd} can be obtained through a similar approach to the one used in proving Proposition \ref{ActBd1}.
\end{proof}

\subsection{Estimate of the derivatives of the contour activities}

\begin{Lemm}\label{PartialS}
Assume that $\lambda\in \BbbR, \tilde{\beta}>0$ and $\gamma_Q\ge 0$ satisfy \eqref{ConvClC}. Then one obtains
\be
\left|
\frac{\partial }{\partial \overline{\nu}_i}\int_{\bs{\varXi}_{\bs{D}_q}} \mathscr{D}(\bome) S_{\rm e}(\bphi_{\rm e})
\right| \le \left(C_0 \tilde{\beta}+\frac{\ex}{\ex-1}\right) |\overline{\bs{D}}_q|\, 
\ex^{-(\gamma_Q-1) R_0^{-d}|\overline{\bs{D}}_q|}\ex^{-\tilde{\beta} e_{\rm e}|\overline{\bs{D}}_q|}.
\ee
\end{Lemm}

\begin{proof}
The proof below is almost identical to that in  \cite[Proposition 4.3]{Borgs1996}.
For the reader's convenience, we provide the details.
For each $B\subset \vLa$, we observe that
\begin{align}
\frac{\partial }{\partial \overline{\nu}_i} \mathcal{T}_{B, \rm e}(\omega,  \bn_{\overline{\partial B}, \rm e})
&=\sum_{j=1}^{|\bs{m}|}\tilde{h}_{{\rm e}, 1, \bn_{\overline{\partial B}}}(s_1)\cdots 
R_j
\cdots \tilde{h}_{{\rm e}, |\bs{m}|, \bn_{\overline{\partial B}}}(s_{|\bs{m}|})\, \ex^{-\tilde{\beta} H_{\ell, B, \rm e}^{(0)}(\bn_{\overline{\partial B}, \rm e})}, \\
R_j
&=\left\{ -(s_j-s_{j-1})\frac{\partial}{\partial \overline{\nu}_i}  H_{\ell, B, \rm e}^{(0)}(\bn_{\overline{\partial B}, \rm e})\right\} \tilde{h}_{{\rm e}, j, \bn_{\overline{\partial B}, \rm e }}(s_j)
+\left(
\frac{\partial}{\partial \overline{\nu}_i}\tilde{h}_{{\rm e}, j, \bn_{\overline{\partial B}, \rm e}}
\right)(s_j).
\end{align}
Hence, one obtains 
\begin{align}
\left\|
\frac{\partial }{\partial \overline{\nu}_i} \mathcal{T}_{B, \rm e}(\omega,  \bn_{\overline{\partial B}, \rm e})
\right\|
&\le 
\sum_{j=1}^{|\bs{m}|}\left\{
(s_j-s_{j-1}) C_0 |\overline{B}|\prod_{A\in \mathcal{A}}|t_{A}|^{m_A} 
+\left|\frac{\partial t_{A_{\pi(j)}}}{\partial \overline{\nu}_i}
\right|
\prod_{\substack{A\in \mathcal{A}
\\
A\neq A_{\pi(j)}
}}|t_{A}|^{m_A} 
\right\}
\ex^{-\tilde{\beta} e_{\rm e}|\overline{B}|}\no
&= \ex^{-\tilde{\beta} e_{\rm e}|\overline{B}|}
\left\{
\tilde{\beta} C_0 |\overline{B}|\prod_{A\in \mathcal{A}}|t_{A}|^{m_A} 
+\sum_{j=1}^{k} m_{A_j}\left|\frac{\partial t_{A_{j}}}{\partial \overline{\nu}_i}
\right|
\prod_{\substack{A\in \mathcal{A}
\\
A\neq A_{j}
}}|t_{A}|^{m_A} 
\right\},
\end{align}
 where $C_0$ is the constant introduced in \hyperlink{A5}{\bf (A. 5)}.
Using this bound, we get
\begin{align}
&\left\|
\frac{\partial }{\partial \overline{\nu}_i} \mathcal{T}_{B, \rm e}(\bn_{\overline{\partial B}, \rm e})
\right\|\no
\le& \sum_{{\mathcal{A}=\{A_1, \dots, A_k\}}\atop{\cup_{j=1}^k A_j=B}} 
\prod_{A\in \mathcal{A}}
\Bigg[
\sum_{m_A=1}^{\infty}
\frac{(|\lambda| |t_{A}| \tilde{\beta})^{m_{A}}}{m_{A}!} 
\Bigg]
\ex^{-\tilde{\beta} e_{\rm e}|\overline{B}|}\times \no
&\times 
\left\{
\tilde{\beta} C_0 |\overline{B}|\prod_{A\in \mathcal{A}}|t_{A}|^{m_A} 
+\sum_{j=1}^{k} m_{A_k}\left|\frac{\partial t_{A_{j}}}{\partial \overline{\nu}_i}
\right|
\prod_{\substack{A\in \mathcal{A}
\\
A\neq A_{j}
}}|t_{A}|^{m_A} 
\right\}\no
\le&\sum_{{\mathcal{A}=\{A_1, \dots, A_k\}}\atop{\cup_{j=1}^k A_j=B}} 
\ex^{-\tilde{\beta} e_{\rm e}|\overline{B}|}
\left\{
C_0 |\overline{B}|\prod_{A\in \mathcal{A}} 
\left\{
 \tilde{\beta}  (\ex-1)|t_{A}|
 \right\}
+\sum_{j=1}^k \ex\tilde{\beta} \left|\frac{\partial t_{A_{j}}}{\partial \overline{\nu}_i}
\right|\prod_{\substack{A\in \mathcal{A}
\\
A\neq A_{j}
}}
\left\{
\tilde{\beta}
(\ex-1)
|t_{A}|
\right\}
\right\}. \label{DiffEst}
\end{align}
The first term on the right-hand side of \eqref{DiffEst} can be evaluated as follows:
\begin{align}
\sum_{{\mathcal{A}=\{A_1, \dots, A_k\}}\atop{\cup_{j=1}^k A_j=B}} \prod_{A\in \mathcal{A}} \left\{ \tilde{\beta}  (\ex-1)|t_{A}|\right\}
&\le \ex^{-\gamma_Q |B|} \sum_{{\mathcal{A}=\{A_1, \dots, A_k\}}\atop{\cup_{j=1}^k A_j=B}} \prod_{A\in \mathcal{A}} 
\left\{ 
\tilde{\beta}   (\ex-1)|t_{A}|\ex^{\gamma_Q |A|}
\right\}\no
&\le \ex^{-\gamma_Q |B|} \sum_{k=1}^{\infty}\frac{1}{k!} 
\prod_{j=1}^k
\sum_{\substack{A_j\in \mathcal{A}_0\\
A_j\cap B\neq \varnothing}
}\left\{
 \tilde{\beta}   (\ex-1)|t_{A_j}|\ex^{\gamma_Q |A_j|}
 \right\}\no
&\le \ex^{-\gamma_Q |B|} \sum_{k=1}^{\infty}\frac{1}{k!} 
\prod_{j=1}^k\sum_{x\in B}
\sum_{\substack{A_j\in \mathcal{A}_0\\
x\in A_j}
}
\left\{
 \tilde{\beta}  (\ex-1)|t_{A_j}|\ex^{\gamma_Q |A_j|}
 \right\}\no
&\le\ex^{-\gamma_Q |B|} \sum_{k=1}^{\infty}\frac{1}{k!} 
\left\{
\tilde{\beta} (\ex-1)\|t\|_{\gamma_Q}|B|
\right\}^k
\no
&\le \ex^{-\gamma_Q |B|} \ex^{|B|}.
\end{align}
Setting $C_1=\sum_{j=1}^k \left|\frac{\partial t_{A}}{\partial \overline{\nu}_i}
\right|/ |t_{A_j}|$, the second term on the right-hand side of \eqref{DiffEst} can be evaluated in the similar way as above:
\begin{align}
& \sum_{{\mathcal{A}=\{A_1, \dots, A_k\}}\atop{\cup_{j=1}^k A_j=B}} 
\ex\tilde{\beta} \sum_{j=1}^k\left|\frac{\partial t_{A}}{\partial \overline{\nu}_i}
\right|\prod_{\substack{A\in \mathcal{A}
\\
A\neq A_{j}
}}
\left\{
\tilde{\beta}
(\ex-1)
|t_{A}|
\right\}
\no
&=\sum_{{\mathcal{A}=\{A_1, \dots, A_k\}}\atop{\cup_{j=1}^k A_j=B}} 
C_1\frac{\ex}{\ex-1} \prod_{A\in \mathcal{A}
}
\left\{
\tilde{\beta}
(\ex-1)
|t_{A}|
\right\}
\no
&\le \ex^{-\gamma_Q |B|} \sum_{{\mathcal{A}=\{A_1, \dots, A_k\}}\atop{\cup_{j=1}^k A_j=B}} C_1\frac{\ex}{\ex-1} 
\prod_{A\in \mathcal{A}} \left\{ \tilde{\beta}   (\ex-1)|t_{A}|\ex^{\gamma_Q |A|}\right\}\no
&\le \frac{\ex}{\ex-1} \ex^{-\gamma_Q |B|} \sum_{k=1}^{\infty}\frac{1}{k!} C_1
\prod_{j=1}^k
\sum_{\substack{A_j\in \mathcal{A}_0\\
A_j\cap B\neq \varnothing}
}\left\{ \tilde{\beta}   (\ex-1)|t_{A_j}|\ex^{\gamma_Q |A_j|}\right\}\no
&\le \frac{\ex}{\ex-1}\ex^{-\gamma_Q |B|} \sum_{k=1}^{\infty}\frac{1}{k!} C_1
\prod_{j=1}^k\sum_{x\in B}
\sum_{\substack{A_j\in \mathcal{A}_0\\
x\in A_j}
}\left\{ \tilde{\beta}   (\ex-1)|t_{A_j}|\ex^{\gamma_Q |A_j|}\right\}\no
&\le \frac{\ex}{\ex-1}\ex^{-\gamma_Q |B|} \sum_{k=1}^{\infty}\frac{1}{(k-1)!} 
\left\{
\tilde{\beta} (\ex-1) \|t\|_{\gamma_Q} |B|
\right\}^k
\no
&\le \frac{\ex}{\ex-1} |B|\, \ex^{-\gamma_Q |B|} \ex^{|B|}. 
\end{align}
Summing up the above, we obtain
\be
\left\|
\frac{\partial }{\partial \overline{\nu}_i} \mathcal{T}_{B, \rm e}(\bn_{\overline{\partial B}, \rm e})
\right\|
\le \ex^{-\tilde{\beta}e_{\rm e} |\overline{B}|} \ex^{-(\gamma_Q-1)|B|} \left(C_0 \tilde{\beta}+\frac{\ex}{\ex-1}\right)|\overline{B}|.\label{DiffEstT}
\ee
Recalling the definition of $S_{\rm e}$, we can see that
\begin{align}
&\frac{\partial }{\partial \overline{\nu}_i} \int_{\bs{\varXi}_{\bs{D}_q}} \mathscr{D}(\bome)S_{\rm e}(\bphi_{\rm e})\no
=&\frac{\partial }{\partial \overline{\nu}_i} \prod_{t=1}^M \Braket{\bs{g}^{(*)}_{\vLa\setminus D^{(t)},  \rm e} \times  \bn_{D^{(t)},  \rm e}^{(t)}| 
\mathcal{T}_{D_q^{(t+1)}, \rm e}\left(\bn_{\overline{\partial D}_q^{(t+1)}, \rm e}\right)
 |\bs{g}^{(*)}_{\vLa\setminus D^{(t+1)},  \rm e}\times  \bn_{D^{(t+1)},  \rm e}^{(t+1)}}.
\end{align}
Combining this with \eqref{DiffEstT}, we obtain the desired assertion in Lemma \ref{PartialS}.
\end{proof}

\begin{Lemm}\label{PartialQ}
Suppose that $\omega_0>\log 2/\tilde{\beta}$.
One obtains the following:\footnote{We consider $\alpha$ instead of $g$ as a parameter in what follows.}
\begin{align}
\left|
\frac{\partial}{\partial \alpha} Q_{X_{\rm p}}(\bs{\varphi}_{\rm e})
\right|
&\le \frac{5}{\alpha\,  \ex^{1/2}}z^N  |\supp Y| \, \ex^{-\tilde{\beta} \omega_{\dagger}|X_{\rm p}|+\tilde{\beta} \omega_0},  \label{PartialQ1}\\
\frac{\partial}{\partial \nu_i} Q_{X_{\rm p}}(\bs{\varphi}_{\rm e})
&=0, \label{DerQ=0}
\end{align}
where $\omega_{\dagger}$ is defined by $\omega_{\dagger}=\omega_0-\log 2/\tilde{\beta}$. Recall here that $N$ is defined by $N=|W(\bs{D})|$.
\end{Lemm}

\begin{proof}
To simplify matters, let us assume that the portion of $\bs{D}$ with all simple loops removed does not intersect with the hyperplanes represented by $t=0$ and $t=\beta$ in the space-time\footnote{In the general case, it is possible to prove the claim by suitably translating certain portions of the proof in the time direction.}.

By using \eqref{ProdRR}, we have
\begin{align}
\frac{\partial}{\partial \alpha}Q_{X_{\rm p}}\left(\bs{\vphi}_{ \rm e}\right)
=\left( \frac{\partial}{\partial \alpha}R_{[\bs{D}]_{\rm S}\setminus W(\bs{D})}\right) R_{W(\bs{D})}+R_{[\bs{D}]_{\rm S}\setminus W(\bs{D})} \left( \frac{\partial}{\partial \alpha} R_{W(\bs{D})} \right). \label{PartialR0}
\end{align}

First, we evaluate an upper bound for  $|\partial R_{[\bs{D}]_{\rm S}\setminus W(\bs{D})}/\partial \alpha|$. 
Let $0\le s_1 \le \cdots \le s_n \le \beta$ and $\vepsilon_1, \dots, \vepsilon_n\in \{-1, 1\}$. We define the $n$-point correlation function as follows:
\begin{align}
&G_n(s_1, \dots, s_n; \vepsilon_1, \dots, \vepsilon_n)\no
=&\big\la\varnothing\big|\ex^{-s_1 \omega_0 N_{x}} \ex^{\im   \vepsilon_1 \alpha \phi_x} \ex^{-(s_2-s_1) \omega_0 N_x} \cdots \ex^{-(s_n-s_{n-1})\omega_0 N_x} \ex^{\im \vepsilon_n \alpha \phi_x} \ex^{-(\beta-s_n) \omega_0 N_x}\big|\varnothing\big\ra.
\end{align}
The expectation value $\braket{\varnothing| \prod_{ {t \in [\mathsf{S}_x(\bs{D}})]_{\rm T}}U_{\bphi}(x, t) |\varnothing}$ can be expressed as a product of several correlation functions $G_{n_1}, G_{n_2}, \dots$, and the maximum number of correlation functions appearing in this product is $|\mathsf{S}_x(X_{\rm e})|+1$. 
By virtue of this fact and Lemma \ref{ProCor1}, we obtain the following:
\begin{align}
\left|
\frac{\partial}{\partial \alpha}
\Braket{
\varnothing| \prod_{ {t \in [\mathsf{S}_x(\bs{D}})]_{\rm T}}U_{\bphi}(x, t) |\varnothing
}
\right|
\le \frac{1}{\alpha\, \ex^{1/2}}(|\mathsf{S}_x(X_{ \rm e})|+1). \label{DiffR}
\end{align}
Thus, utilizing the inequality \eqref{DiffR} obtained above and recalling the definition \eqref{ProdRR} of $R_{[\bs{D}]_{\rm S}\setminus W(\bs{D})}$, we obtain 
\begin{align}
\left|
\frac{\partial}{\partial \alpha}
R_{[\bs{D}]_{\rm S}\setminus W(\bs{D})}
\right|
\le& \frac{1}{\alpha\, \ex^{1/2}} \, \ex^{-\tilde{\beta} \omega_0(|X_{\rm p}\cap (\bs{D}\setminus \mathbb{T}_{W(\bs{D})})|-1)}
\sum_{x\in [\bs{D}]_{\rm S}\setminus W(\bs{D})}(|\mathsf{S}_x(X_{ \rm e})|+1)\no
\le& \frac{2}{\alpha\, \ex^{1/2}} |\supp Y| \, \ex^{-\tilde{\beta} \omega_0(|X_{\rm p}\cap (\bs{D}\setminus \mathbb{T}_{W(\bs{D})})|-1)}.
\label{PartialR1}
\end{align}

Next, we will evaluate an upper bound for $|\partial R_{W(\bs{D})}/\partial \alpha |$. To do so, it suffices to evaluate $\left|
\partial
\Tr\left[
\prod_{ {t \in \mathbb{M}}}U_{\bphi}(x, t)\right]/\partial \alpha\right|$ for each $x\in W(\bs{D})$.

Let  $x\in W(\bs{D})$. If  $\mathsf{S}_x(X_{\rm p})=\varnothing$, then the operator $\ex^{-\tilde{\beta} \omega_0 N_x}P^{\perp}_{\varnothing}$ does not appear in the product of operators $\prod_{t \in \mathbb{M}}U_{\bphi}(x, t)$. As a result, we can represent $\Tr\left[\prod_{t \in \mathbb{M}}U_{\bphi}(x, t)\right]$ in the form of the right-hand side of \eqref{CorF2}. 
Therefore, by using the inequality  \eqref{PartialG2}, we obtain
\be
\left|
\frac{\partial}{\partial \alpha}
\Tr\left[
 \prod_{ {t \in  \mathbb{M}}}U_{\bphi}(x, t)\right]\right| \le \frac{1}{\alpha\, \ex^{1/2}}z.
\ee

Let us consider the case where $\mathsf{S}_x(X_{\rm p})\neq \varnothing$. We will employ the symbols $Z_1$ and $Z_2$ introduced in the proof of Lemma \ref{QwithL}. The product of operators $\prod_{t\in \mathbb{M}}U_{\bphi}(x,t)$ contains $n_x=|\mathsf{S}_x(X_{\rm p})|$  operators of  the type $Z_2$ and $M-n_x$ operators of the  type $Z_1$, denoted by $Z_1^{(1)},\dots,Z_1^{(M-n_x)}$. 
By decomposing $Z_2$ as $Z_2=\tilde{Z}_2-\tilde{Z}_2P_{\varnothing}$, where  $\tilde{Z}_2=\exp(-\tilde{\beta}\omega_0 N_x)$, we can represent the product of operators $\prod_{t\in \mathbb{M}}U_{\bphi}(x,t)$ as a sum of $2^{n_x}$ operators:
\be
\prod_{ {t \in  \mathbb{M}}}U_{\bphi}(x, t)=\sum_{i=1}^{2^{n_x}} J_i.
\ee
Each $J_i$ can be expressed as a product of $M-n_x$ operators of the type $Z_1$ and a total of $n_x$ operators, which are a combination of repeated appearances of two operators: $\tilde{Z}_2$ and $\tilde{Z}_2P_{\varnothing}$.
Note that $J_1$ is defined as the only term that does not include $P_{\varnothing}$. For example, we have
\be
Z_1Z_2=Z_1\tilde{Z}_2-Z_1 \tilde{Z}_2P_{\varnothing}=J_1+J_2.
\ee
Under the above setup, we have
\be
 \Tr\left[
 \prod_{ {t \in  \mathbb{M}}}U_{\bphi}(x, t)\right]=
 \sum_{i=1}^{2^{n_x}}\Tr[J_i].
 \ee
 Since $\Tr[J_1]$ does not include $P_{\varnothing}$, it can be expressed in the form of the right-hand side of \eqref{CorF2}. Therefore, using \eqref{PartialG2}, we obtain
 \be
 \left|
\frac{\partial}{\partial \alpha}
\Tr\left[
 J_1
 \right]\right| \le  \frac{1}{\alpha\, \ex^{1/2}}z.
 \ee
 For $i\geq 2$, $J_i$ contains at least one $P_{\varnothing}$, hence $\Tr[J_i]$ can be expressed as a product of correlation functions:
 \be
\Tr[J_i]=G_{n^{(i)}_1} \cdots G_{n^{(i)}_k}, 
\ee
 where $k\le n_x+1$. Therefore, using \eqref{PartialG1}, we obtain
 \be
\left|
\frac{\partial}{\partial \alpha}
\Tr\left[
 J_i
 \right]\right| \le  \frac{k}{\alpha\, \ex^{1/2}} \le \frac{n_x+1}{\alpha\, \ex^{1/2}}.
\ee
 Putting everything together, we have
 \be
\left|
\frac{\partial}{\partial \alpha}
\Tr\left[ 
\prod_{ {t \in  \mathbb{M}}}U_{\bphi}(x, t)
\right]
\right|\le \frac{n_x+1}{\alpha\, \ex^{1/2}} (2^{n_x}-1)+ \frac{1}{\alpha\, \ex^{1/2}}z\le  \frac{n_x+1}{\alpha\, \ex^{1/2}} 2^{n_x} z.
\ee
Therefore, by setting $k=|W(\bs{D})|-\left|\left[\mathbb{T}_{W(\bs{D})} \cap X_{\rm p}\right]_{\rm S} \right|$, one obtains 
\begin{align}
\left|
\frac{\partial}{\partial \alpha}
R_{ W(\bs{D})}
\right|
\le &
 k z^{k-1}\frac{1}{\alpha\, \ex^{1/2}}\, \ex^{-\tilde{\beta}\omega_0 (\left|\mathbb{T}_{W(\bs{D})} \cap X_{\rm p}\right|-1)}\no
& +\sum_{x\in
  \left[\mathbb{T}_{W(\bs{D})} \cap X_{\rm p}\right]_{\rm S}
 }z^{k+1}\frac{n_x+1}{\alpha\, \ex^{1/2}} 2^{n_x} \, \ex^{-\tilde{\beta}\omega_0 (\left|\mathbb{T}_{W(\bs{D})} \cap X_{\rm p}\right|-1)}\no
 \le &\frac{3}{\alpha\, \ex^{1/2}}
 z^N|\supp Y|\, 
 \ex^{-\tilde{\beta} \omega_{\dagger}|\mathbb{T}_{W(\bs{D})}\cap X_{\rm p}|+\tilde{\beta} \omega_0}. \label{PartialR2}
\end{align}
Combining \eqref{PartialR0}, \eqref{PartialR1}, and \eqref{PartialR2} yields \eqref{PartialQ1}.

Equation \eqref{DerQ=0} is evident from the fact that $Q_{X_{\rm p}}(\bs{\varphi}_{\rm e})$ does not depend on the $\nu_i$.
\end{proof}

\begin{Prop}\label{DelRhoEst}
Assume that $\lambda\in \BbbR, \tilde{\beta}>0$ and $\gamma_Q\ge 0$ satisfy \eqref{ConvClC}.  Furthermore, assume that $\tilde{\beta} \ge \beta_0/2$ and $(\gamma_Q-1-\omega_0) R_0^{-d}-\gamma_{\rm e}-|e_{\rm e}|>0$. Then one obtains 
\be
\left|
\frac{\partial}{\partial \overline{\nu}_i} \rho(Y)
\right|
\le \left(
2\tilde{\beta} C_0+\frac{\ex}{\ex-1}+\frac{5}{\alpha \, \ex^{1/2}}
\right) |\supp Y|\,  \ex^{-(\tilde{\beta}e_{\rm e}+\tilde{\beta}c+\gamma_{\dagger})|\supp Y|},
\ee
where 
$\gamma_{\dagger}$ is defined as  
$\gamma_{\dagger}=\min\left\{(\gamma_Q-1-\omega_0) R_0^{-d}-\gamma_{\rm e}-|e_{\rm e}|,\,  \tilde{\beta} \min\{\gamma_{\rm e}, \omega_{\dagger}R_0^{-d}\}\right\}-5 \log 2$, and $c$ is given in Proposition \ref{RhoEst2}.
\end{Prop}

\begin{proof}
Using the expression for $\rho(Y)$ given in \eqref{ActRep}, we obtain\footnote{
It should be noted that the justification for exchanging the order of derivatives and integrals here is provided by the inequalities \eqref{Inq1}-\eqref{Inq3}.
}
\begin{align}
&\frac{\partial}{\partial \overline{\nu}_i} \rho(Y)\no
=&\sum_{\bs{D}_q \subseteq \bs{D}} 
\sum_{{X_{\rm e}, X_{\rm p} \subseteq \bs{D}_c}\atop{X_{\rm e} \cup X_{\rm p}=\bs{D}_c}}
\sum_{\substack{\bn_{\bs{D}_c,  \rm e}:  \\ \mathscr{E}_{\rm e}(\bn_{\bs{D}_c,   \rm e})=X_{\rm e}}}
\sum_{\bn_{\bs{D}_q,  \rm e}\in \bN_{\bs{D}_q,  \rm e}} \int_{\bs{\varXi}_{\bs{D}_q}} \mathscr{D}(\bome)
\, \ex^{-\tilde{\beta} E_{\rm e}(\bn_{\bs{D}_c, \rm e})} \times \no 
&\times \left\{ 
\left(-\tilde{\beta} \frac{\partial}{\partial \overline{\nu}_i} \bs{E}(\bn_{\bs{D}_c})\right)
S_{\rm e}(\bs{\vphi}_{\rm e})Q_{X_{\rm p}}(\bs{\vphi}_{ \rm e})
+
\left( \frac{\partial}{\partial \overline{\nu}_i} S_{\rm e}(\bs{\vphi}_{\rm e})\right)Q_{X_{\rm p}}(\bs{\vphi}_{ \rm e})
+
S_{\rm e}(\bs{\vphi}_{\rm e})\left( \frac{\partial}{\partial \overline{\nu}_i} Q_{X_{\rm p}}(\bs{\vphi}_{ \rm e})\right)
\right\}.
\end{align} 
Using the inequality obtained from the assumption \hyperlink{A5}{\bf (A. 5)}:
$|\frac{\partial}{\partial \overline{\nu}_i} \bs{E}(\bn_{\bs{D}_c})|\le C_0 |\overline{\bs{D}}_c|$
and Lemmas \ref{PartialS} and \ref{PartialQ}, we can follow a similar argument as in the proof of Proposition \ref{ActBd1} to obtain the following:
\begin{align}
&\int_{\bs{\varXi}_{\bs{D}_q}} \mathscr{D}(\bome)
\, \ex^{-\tilde{\beta} E_{\rm e}(\bn_{\bs{D}_c, \rm e })} \left|
\left(\tilde{\beta} \frac{\partial}{\partial \overline{\nu}_i} \bs{E}(\bn_{\bs{D}_c})\right)
S_{\rm e}(\bs{\vphi}_{\rm e})Q_{X_{\rm p}}(\bs{\vphi}_{ \rm e})
\right|\label{Inq1}\no
\le & c_0\tilde{\beta}|\supp Y| \ex^{-(\tilde{\beta}e_{\rm e}+\tilde{\beta}c+\gamma+5 \log 2)|\supp Y|},\\
&\int_{\bs{\varXi}_{\bs{D}_q}} \mathscr{D}(\bome)
\, \ex^{-\tilde{\beta} E_{\rm e}(\bn_{\bs{D}_c, \rm e })}
\left|
\left( \frac{\partial}{\partial \overline{\nu}_i} S_{\rm e}(\bs{\vphi}_{\rm e})\right)Q_{X_{\rm p}}(\bs{\vphi}_{ \rm e})
\right|\no
\le & 
\left(
c_0 \tilde{\beta}+\frac{\ex}{\ex-1}
\right)|\supp Y|
\ex^{-(\tilde{\beta}e_{\rm e}+\tilde{\beta}c+\gamma+5 \log 2)|\supp Y|},\\
&\int_{\bs{\varXi}_{\bs{D}_q}} \mathscr{D}(\bome)
\, \ex^{-\tilde{\beta} E_{\rm e}(\bn_{\bs{D}_c, \rm e })} \left|
S_{\rm e}(\bs{\vphi}_{\rm e})\left( \frac{\partial}{\partial \overline{\nu}_i} Q_{X_{\rm p}}(\bs{\vphi}_{ \rm e})\right)
\right|\no
\le & 
\frac{5}{\alpha \, \ex^{1/2}} |\supp Y\, |\ex^{-(\tilde{\beta}e_{\rm e}+\tilde{\beta}c+\gamma_{\dagger}+5\log 2)|\supp Y|}.\label{Inq3}
\end{align}
Putting everything together, we obtain the desired claim.
\end{proof}

\subsection{Contour representation of thermal expectations}
For a given observable $\varPsi$, its thermal expectation value is given by
\be
\la \vPsi \ra_{\beta, \vLa}^{(\ell)}=\frac{\Tr\left[\, \vPsi\,  \ex^{-\beta H_{\ell, \vLa}} \right]}{Z_{\ell, \vLa}}
=\frac{Z_{\ell, \vLa}^{\vPsi}}{Z_{\ell, \vLa}}.
\ee
Considering the Lang--Firsov transformation, we can assume, without loss of generality, that $\varPsi$ takes a simple form as follows:
\be
\varPsi=\varPsi_{\rm e}\otimes \varPsi_{\rm p},\label{ProdPsi}
\ee
 where $\varPsi_{\rm e} \in \mathfrak{A}_{\rm e}$ with $\supp \varPsi_{\rm e}$  being a finite and connected set, and $\varPsi_{\rm p}=\exp(\im \sum_{x\in \supp \varPsi_{\rm e}} \mu_x q_x)$ with $\mu_x\in \BbbR$.
Here, recall that $q_x$ is defined by \eqref{Defpq}.

We will discuss the contour representation of $Z_{\ell, \vLa}^{\vPsi}$ here, while the contour representation of the partition function $Z_{\ell, \vLa}$ has been extensively discussed in the previous subsections.
By using \eqref{TSigma}, we have the following expression:
\be
Z_{\ell, \vLa}^{\vPsi}=\left[
\prod_{t=1}^M \sum_{\vSi_t}
\right] W_{\vPsi}(\vSi_1, \dots, \vSi_M),\quad W_{\vPsi}(\vSi_1, \dots, \vSi_M)=\Tr\left[
\vPsi K(\vSi_1)\cdots K(\vSi_M)
\right].
\ee
We define $D(\vPsi)=\bigcup_{x\in \supp \vPsi_{\rm e}}C(x, 1)\in \mathsf{E}$.
Given $\bs{\vSi}_{\mathbb{L}_{\vLa}}$, the supports of the corresponding contours are defined by the connected components of the set $\mathbb{D} \cup D(\vPsi)$.
At least one contour corresponding to $\bs{\vSi}_{\mathbb{L}_{\vLa}}$ must necessarily include $D(\vPsi)$ in its support. We shall denote this contour by $Y_{\vPsi}$. Using a procedure similar to the one used to obtain the contour representation of the partition function, we can derive the following contour representation of  $Z_{\ell, \vLa}^{\vPsi}$:
\be
Z_{\ell, \vLa}^{\vPsi}=\sum_{\{Y_{\vPsi}, Y_1, \dots, Y_n\}} \ex^{-\tilde{\beta} \sum_{m} e_{m} |V_{m}|}\rho_{\vPsi}(Y_{\vPsi}) \rho(Y_1)\cdots \rho(Y_n).
\ee
We provide   the  definition of $\rho_{\vPsi}(Y_{\vPsi})$  below:
Suppose that $\supp Y_{\vPsi}$ can be expressed as $\supp Y_{\vPsi}=\bs{D}\cup D(\vPsi)$.
By defining
\be
D_{\vPsi}^{(t)}=
\begin{cases}
D^{(t)}\cup \supp \varPsi_{} & t=1, M\\
D^{(t)} & \mbox{otherwise},
\end{cases}
\ee
we can express $\rho_{\vPsi}(Y_{\vPsi})$ concretely as follows:
\begin{align}
\rho_{\vPsi}(Y_{\vPsi})
=&
\sum_{\bs{D}_q \subseteq \bs{D}} \sum_{\bn_{\bs{D}_q}\in \bN_{\bs{D}_q}} \sum_{\bn_{\bs{D}_c}\in \bE_{\bs{D}_c}} 
\sum_{\bn^{(1)}_{ \supp \vPsi_{\rm e}\setminus D_q^{(1)}}\in \mathcal{N}_{\supp \vPsi_{\rm e}\setminus D_q^{(1)}}}
\sum_{
\bn^{(M)}_{\supp \vPsi_{\rm e}  \setminus D_q^{(M)}} \in \mathcal{N}_{\supp \vPsi_{\rm e}\setminus D_q^{(M)}}
}
\ex^{-\tilde{\beta}  \bs{E}(\bn_{\bs{D}_c })}
\times  \no
&\times 
\Braket{\bs{g}^{(*)}_{\vLa\setminus D_{\vPsi}^{(M)}}\times   \bn_{D_{\vPsi}^{(M)}}^{(M)}| \, \vPsi\,  \mathcal{T}_{D_q^{(1)}}\left(\bn_{\overline{\partial D}_q^{(1)}}\right) |\bs{g}^{(*)}_{\vLa\setminus D_{\vPsi}^{(1)}}\times  \bn_{D_{\vPsi}^{(1)}}^{(1)}}\times \no
&\times \prod_{t=1}^{M-1}\Braket{\bs{g}^{(*)}_{\vLa\setminus D_{\vPsi}^{(t)}}\times   \bn_{D_{\vPsi}^{(t)}}^{(t)}| \mathcal{T}_{D_q^{(t+1)}}\left(\bn_{\overline{\partial D}_q^{(t+1)}}\right) |\bs{g}^{(*)}_{\vLa\setminus D_{\vPsi}^{(t+1)}}\times  \bn_{D_{\vPsi}^{(t+1)}}^{(t+1)}}. \label{DefRhoP1}
\end{align}
When deriving an upper bound for $\rho_{\vPsi}(Y_{\vPsi})$, it is convenient to use the following representation that separates the electron and phonon parts, as derived in \eqref{RhoEP}:
\begin{align}
&\rho_{\vPsi}(Y_{\vPsi})\no
=&\sum_{\substack{
\bs{D}_q \subseteq \bs{D}
 }}
\sum_{{X_{\rm e}, X_{\rm p} \subseteq \bs{D}_c}\atop{X_{\rm e} \cup X_{\rm p}=\bs{D}_c}}
\sum_{\substack{\bn_{\bs{D}_c,  \rm e}:  \\ \mathscr{E}_{\rm e}(\bn_{\bs{D}_c,  \rm e})=[X_{\rm e}]}}
\sum_{\bn_{\bs{D}_q,  \rm e}\in \bN_{\bs{D}_q,  \rm e}}
\sum_{\bn^{(1)}_{ \supp \vPsi_{\rm e} \setminus D_q^{(1)}, \rm e}\in \mathcal{N}_{\supp \vPsi_{\rm e} \setminus D_q^{(1)}, \rm e}}
\sum_{
\bn^{(M)}_{\supp \vPsi_{\rm e}  \setminus D_q^{(M)}, \rm e} \in \mathcal{N}_{\supp \vPsi_{\rm e} \setminus D_q^{(M)}, \rm e}
}\times \no
 &\times \int_{\bs{\varXi}_{\bs{D}_q}} \mathscr{D}(\bome)\,  \ex^{-\tilde{\beta} E_{\rm e}(\bn_{\bs{D}_c, \rm e })} S_{\vPsi_{\rm e}, \rm e}(\bs{\vphi}_{\rm e})Q_{\varPsi_{\rm p}, X_{\rm p}}(\bs{\vphi}_{\rm e}), 
\end{align}
where $S_{\vPsi_{\rm e}, \rm e }(\bphi_{\rm e })$ is defined by 
\begin{align}
S_{\vPsi_{\rm e}, \rm e }(\bphi_{\rm e })&=
\Braket{\bs{g}^{(*)}_{\vLa\setminus D_{\vPsi}^{(M)},  \rm \rm e } \times  \bn_{D_{\vPsi}^{(M)},  \rm \rm e }^{(M)}| \, \vPsi_{\rm e}\, \mathcal{T}_{D_q^{(1)}, \rm e}\left(\omega^{(1)},  \bn_{\overline{\partial D}_q^{(1)}, \rm e}\right)
 |\bs{g}^{(*)}_{\vLa\setminus D_{\vPsi}^{(1)},  \rm \rm e }\times  \bn_{D_{\vPsi}^{(1)},  \rm \rm e }^{(1)}}\times \no
&\times \prod_{t=1}^{M-1} 
\Braket{\bs{g}^{(*)}_{\vLa\setminus D_{\vPsi}^{(t)},  \rm \rm e } \times  \bn_{D_{\vPsi}^{(t)},  \rm \rm e }^{(t)}| \mathcal{T}_{D_q^{(t+1)}, \rm e}\left(\omega^{(t+1)},  \bn_{\overline{\partial D}_q^{(t+1)}, \rm e}\right)
 |\bs{g}^{(*)}_{\vLa\setminus D_{\vPsi}^{(t+1)},  \rm \rm e }\times  \bn_{D_{\vPsi}^{(t+1)},  \rm \rm e }^{(t+1)}}.
\end{align}
In addition, $Q_{\varPsi_{\rm p}, X_{\rm p}}(\bs{\vphi}_{\rm e})$ is defined by replacing $S_{\rm p}(\bs{\vphi}(x))$ in the definition \eqref{DefQ1} of $Q_{X_{\rm p}}$ with the following $S_{\varPsi_{\rm p}, \rm p}(\bs{\vphi}(x))$:
\begin{align}
S_{\varPsi_{\rm p}, \rm p}(\bs{\vphi}(x))=
\begin{cases}
\Tr\left[\ex^{\im \mu_x q_x}\prod_{t\in \mathbb{M}} u_{\bs{\vphi}}(x, t)\right] & x\in \supp \varPsi_{\rm e}\\
\Tr\left[\prod_{t\in \mathbb{M}} u_{\bs{\vphi}}(x, t)\right] & x\in [\bs{D}]_{\rm S} \setminus \supp \varPsi_{\rm e}.
\end{cases}
\end{align}

Performing the similar argument as in Proposition \ref{ActBd1}, we obtain the following proposition:
\begin{Prop}\label{ExpRhoPsi}
Assume that $\lambda\in \BbbR, \tilde{\beta}>0$ and $\gamma_Q\ge 0$ satisfy \eqref{ConvClC}.  Furthermore, assume that $\tilde{\beta} \ge \beta_0/2$.
Then one obtains
\be
\left|\rho_{\vPsi}(Y_{\vPsi})\right|
\le \|\vPsi\| \ex^{-(\tilde{\beta} e_{\rm e}+\tilde{\beta} c)} \ex^{-\gamma |\supp Y \setminus D(\vPsi)|}
(1+\ex^{-\gamma})^{|D(\vPsi)|}.
\ee
\end{Prop}

\section{Proof of Theorem \ref{MainThm}}\label{Sec5}
The proof of Theorem \ref{MainThm} is almost similar to that of \cite{Borgs2000}. For the reader's convenience, we give a brief outline.
Choose $\beta_0$ and $\gamma_Q$ sufficiently large and set $\tilde{\beta}\in (\frac{1}{2}\beta_0, \beta_0]$. This determines $M$. Choose $\lambda$ such that $|\lambda|\le \lambda_0$. Then, by following a similar procedure as in the proof of \cite[Theorem 2.1]{Borgs1996}, we can show that  (i)-(iii) of Theorem \ref{MainThm} hold as long as the condition $\mathrm{Re}f_{\ell}(\mu)-\min_m \mathrm{Re} f_m(\mu)=0$ is satisfied. By applying the method of \cite[Theorem 3.2]{Borgs2000}, we obtain  (vii) of Theorem \ref{MainThm}.
We can prove (viii) of Theorem \ref{MainThm} using a similar argument as the proof of \cite[Theorem 3.4]{Borgs2000}.  In what follows, we assume all the conditions listed above.

To examine the ground state expectation value, we introduce some notations. First, we define
\be
\mathcal{Z}_{\ell, \vLa}=\braket{\bs{g}^{(\ell)}_{\vLa}|\, \ex^{-\beta H_{\ell, \vLa}}|\bs{g}^{(\ell)}_{\vLa}}
\ee
and
\be
\ket{\varOmega_{\ell, \vLa}}=\frac{1}{\sqrt{\mathcal{Z}_{\ell, \vLa}}} \ex^{-\frac{\beta}{2}H_{\ell, \vLa}} \ket{\bs{g}_{\vLa}^{(\ell)}}. 
\ee
Given any local observable $\vPsi$ expressed in the form of \eqref{ProdPsi}, we now consider the following expectation value:
\be
\la\!\la \vPsi\ra\!\ra_{\beta, \vLa}^{(\ell)}=\braket{\varOmega_{\ell, \vLa}|\vPsi|\varOmega_{\ell, \vLa} }.
\ee
By using a similar argument to the proof of \cite[Lemma 3.3]{Borgs2000}, we can obtain  
 \be
\left| \frac{1}{\beta |\vLa|} \log\mathcal{Z}_{\ell, \vLa} +f_{\ell}\right| \le \mathcal{O}\left(
\frac{1}{\beta}+\frac{1}{\beta_0} \frac{|\partial \vLa|}{|\vLa|}
\right)
\ee
and
\be
\la\!\la \vPsi\ra\!\ra_{\beta, \vLa}^{(\ell)}=\la \vPsi\ra_{\rm gs}^{(\ell)}+\mathcal{O}\left(
\ex^{-c\gamma \min\{\mathrm{dist}(\supp \vPsi, \partial \vLa), M/2\}}
\right).
\ee
 
 Furthermore, comparing the convergence cluster expansions of $\la\!\la \vPsi\ra\!\ra_{\beta, \vLa}^{(\ell)}$ and $\la \varPsi\ra_{\beta}^{(\ell)}$, the only difference between them comes from the contributions of contours crossing the boundary $t=M$ and those crossing $\partial \vLa$, which can be estimated to be of the order  $\mathcal{O}\left(
\ex^{-c\gamma \min\{\mathrm{dist}(\supp \vPsi, \partial \vLa), M/2\}}
\right)$. By utilizing the abovementioned results, we can obtain (iv) of Theorem \ref{MainThm}. Moreover, (v) of Theorem \ref{MainThm} follows directly from (iv) of Theorem \ref{MainThm}.
\qed

\appendix
\section{Bounds on   correlation functions for a boson system}\label{AppCor}

Let $\h$ denote a separable complex Hilbert space. We assume that $\h$ is of finite dimension, which is sufficient for the purpose of this paper.
The bosonic Fock space over $\h$ is defined as $\mathfrak{F}(\h):=\bigoplus_{n=0}^{\infty} \otimes^n_{\rm s} \h$. The annihilation and creation operators on $\mathfrak{F}(\h)$ are denoted as $b(f)$ and $b(f)^*$, respectively. It is a well-known fact that these operators satisfy the standard commutation relations:\footnote{
More precisely, these commutation relations are valid on suitable subspaces of $\mathfrak{F}(\h)$. For example, it is known that they hold on finite particle spaces of $\mathfrak{F}(\h)$.  For more details, see \cite{Arai2016}.
}
\be
[b(f), b(g)^*]=\braket{f| g},\quad [b(f), b(g)]=0 \quad (f, g\in \h).
\ee
For any $f\in \h$, we define the operator $\phi(f)$ as
\be
\phi(f)=b(f)+b(f)^*.
\ee
It should be noted that  the operator $\phi(f)$ is essentially self-adjoint, and hence, we shall denote its closure by the same symbol.

Let $\omega$ denote a positive self-adjoint operator acting on $\h$. We define a positive self-adjoint operator $H_0$ by setting $H_0=\dG(\omega)$, where $\dG(\omega)$ represents the second quantization of $\omega$.\footnote{
The precise definition is as follows:
\be
\dG(\omega): =0\oplus \left[ \bigoplus_{n=1}^{\infty} \sum_{j=1}^n \one\otimes \cdots \otimes \overbrace{\omega}^{j^{\rm th}} \otimes  \cdots \otimes \one\right].
\ee
For the basic properties of the second quantization operators, see \cite{Arai2016}.
}
For each $s_1, \dots, s_n\in \BbbR$ with  $0\le s_1\le \cdots \le s_n \le \beta$ and $f_1, \dots, f_n\in \h$, we  define the $n$-point correlation function $G_n(s_1, \dots, s_n; f_1, \dots, f_n)$ as follows:
\begin{align}
&G_n(s_1, \dots, s_n; f_1, \dots, f_n)\no
=&\big\la\varnothing\big|\ex^{-s_1 H_0} \ex^{\im   \phi(f_1)} \ex^{-(s_2-s_1) H_0} \cdots \ex^{-(s_n-s_{n-1}) H_0} \ex^{\im \phi(f_n)} \ex^{-(\beta-s_n) H_0}\big|\varnothing\big\ra,
\end{align}
where $\ket{\varnothing}$ denotes the Fock vacuum.

\begin{Lemm}\label{ProCor1}
Suppose that $0\le s_1\le \cdots \le s_n \le \beta$.
For every $f_1, \dots, f_n\in \h$, the following holds:
\begin{align}
G_n(s_1, \dots, s_n; f_1, \dots, f_n)
=\exp\bigg[
- \frac{1}{2}\sum_{i, j=1}^n \la f_i|\ex^{-|s_i-s_j| \omega} f_j\ra
\bigg]. \label{CorF1}
\end{align}
As a consequence, we obtain the following:
\begin{itemize}
\item[\rm (i)]
$0\le G_n(s_1, \dots, s_n; f_1, \dots, f_n) \le 1$.
\item[\rm (ii)]
The function $G_n(s_1, \dots, s_n; \alpha f_1, \dots, \alpha f_n)$ is differentiable with respect to $\alpha>0$, and the following holds:
\begin{align}
\left|\frac{\partial }{\partial \alpha} G_n(s_1, \dots, s_n; \alpha f_1, \dots, \alpha f_n)\right| \le \frac{1}{\alpha \, \ex^{1/2}}. \label{PartialG1}
\end{align}
\end{itemize}
\end{Lemm}
\begin{proof}
The formula \eqref{CorF1} is a well-known result, and it can also be obtained by taking the limit $\beta\to\infty$ in \eqref{CorF2} given below.

For $x>0$, consider the function $f(x)=\ex^{-\alpha^2 x/2}$. We can show that $|\partial f(x)/\partial \alpha|$ attains its maximum value of $1/(\alpha \, \mathrm{e}^{1/2})$, when $x=1/\alpha^2$. Using this result, we can prove (ii).
\end{proof}

In what follows, we assume that $\omega$ is a strictly positive operator.
We define the partition function as $Z_{\beta}=\mathrm{Tr}[\mathrm{e}^{-\beta H_0}]$. It is well-known that the following formula  holds:
\begin{align}
Z_{\beta}=\frac{1}{\det(\one-\ex^{-\beta \omega})}. \label{part-func}
\end{align}

Now, the thermal correlation function is defined as follows:
For $0\le s_1\le \cdots \le s_n < \beta$ and $f_1, \dots, f_n\in \h$, we define
\begin{align}
&G_{\beta, n}(s_1, \dots, s_n; f_1, \dots, f_n)\no
=&\frac{1}{Z_{\beta}}\Tr\left[
\ex^{-s_1H_0}\ex^{\im \phi(f_1)}\ex^{-(s_2-s_1)H_0} \ex^{\im \phi(f_2)}\cdots \ex^{\im \phi(f_n)} \ex^{-(\beta-s_n)H_0}
\right].
\end{align}

\begin{Lemm}\label{ProCor2}
Suppose that $0\le s_1\le \cdots \le s_n < \beta$.
For every $f_1, \dots, f_n\in \h$, the following holds:
\begin{align}
&G_{\beta, n}(s_1, \dots, s_n; f_1, \dots, f_n)\no
=& \exp
\left[
-\frac{1}{2} \sum_{i, j=1}^n \Big\la f_i\Big| \left(\one-\ex^{-\beta \omega}\right)^{-1}
\left( \ex^{-|s_i-s_j| \omega }+\ex^{-(\beta-|s_i-s_j|)\omega}\right)
 f_j\Big\ra
\right].\label{CorF2}
\end{align}
As a consequence, we  obtain the following:
\begin{itemize}
\item[\rm (i)]
\be
\left|
\Tr\left[
\ex^{-s_1H_0}\ex^{\im \phi(f_1)}\ex^{-(s_2-s_1)H_0} \ex^{\im \phi(f_2)}\cdots \ex^{\im \phi(f_n)} \ex^{-(\beta-s_n)H_0}
\right]
\right|
\le Z_{\beta}. \label{TherCorr}
\ee

\item[\rm (ii)]
The function $G_{\beta, n}(s_1, \dots, s_n; \alpha f_1, \dots, \alpha f_n)$ is differentiable with respect to $\alpha>0$, and the following holds:
\begin{align}
\left|\frac{\partial }{\partial \alpha} G_{\beta,n }(s_1, \dots, s_n; \alpha f_1, \dots, \alpha f_n)\right| \le \frac{1}{\alpha\,  \ex^{1/2}}.
\label{PartialG2}
\end{align}

\end{itemize}

\end{Lemm}
\begin{proof}
\eqref{CorF2} is a well-known formula. For the proof, see, for example, \cite[Theorem 4.15]{Arai2022}.
A key property to keep in mind when applying this formula is the following:
\be
\sum_{i, j=1}^n \Big\la f_i\Big| \left(\one-\ex^{-\beta \omega}\right)^{-1}
\left( \ex^{-|s_i-s_j| \omega }+\ex^{-(\beta-|s_i-s_j|)\omega}\right)
 f_j\Big\ra \ge 0.
\ee
A concise explanation for why this holds is as follows. Let $L^2(0, \beta)$ be the Hilbert space of square-integrable functions on $(0, \beta)$, and let $\h_{\beta}=L^2(0, \beta)\otimes \h$.
Consider the periodic Laplacian $\Delta_{\rm P}$ acting on  $L^2(0, \beta)$. We also introduce a norm on $\h_{\beta}$ defined as follows:
\be
\|
f
\|_{\beta}^2=
\frac{1}{2\beta} \left\| 
\sqrt{\frac{\one\otimes \omega}{(-\Delta_{\rm P})\otimes \one+\one\otimes \omega^2}} f
\right\|^2\quad (f\in \h_{\beta}).
\ee
We use the same symbol $\h_{\beta}$ to denote the completion of $\h_{\beta}$ with respect to the norm $\|\cdot \|_{\beta}$.
Then, we can easily verify that the following identity holds:
\be
\la \delta_s\otimes f|\delta_t\otimes g\ra_{\beta}
=
\frac{1}{2}
\Big\la f \Big| \left(\one-\ex^{-\beta \omega}\right)^{-1}
\left( \ex^{-|s-t| \omega }+\ex^{-(\beta-|s-t|)\omega}\right)
 g\Big\ra,
\ee
where $\delta_s$ is the Dirac delta function.
Using this identity, we obtain the following inequality:
\be
\frac{1}{2}\sum_{i, j=1}^n \Big\la f_i\Big| \left(\one-\ex^{-\beta \omega}\right)^{-1}
\left( \ex^{-|s_i-s_j| \omega }+\ex^{-(\beta-|s_i-s_j|)\omega}\right)
 f_j\Big\ra
= \left\| \sum_{i=1}^n \delta_{s_i}\otimes f_i\right\|_{\beta}^2\ge 0,
\ee
which gives the desired property.

The property (i) is readily derived from the fact established above.
(ii) can be shown in the same way as the proof of Lemma \ref{ProCor1}.
\end{proof}

\section{Proof of \eqref{TilGamma}}

\begin{Lemm}\label{EstGa}

Suppose that $(\gamma_Q-1)R_0^{-d}-\gamma_{\rm e}-|e_{\rm e}|>0$. Then, under the setup of the proof of Proposition \ref{ActBd1}, the following holds:
\be
E_{\rm eff}(\bn_{\bs{D}_c,  \rm e }, X_{\rm p})+e_{\rm e}|\overline{\bs{D}}_q|+(\gamma_Q-1)R_0^{-d} |\overline{\bs{D}}_q|\ge (e_{\rm e}+\tilde{\gamma})|\bs{D}|.
\ee
\end{Lemm}
\begin{proof}
First, we note the following:
\begin{align}
E_{\rm e}(\bn_{\bs{D}_c, \rm e})
\ge &(\gamma_{\rm e}+e_{\rm e})\left|
X_{\rm e} \setminus (X_{\rm e}\cap \overline{\bs{D}}_q)
\right|
+e_{\rm e} \left|
(\bs{D}_c\setminus X_{\rm e})\setminus [
(\bs{D}_c \setminus X_{\rm e}) \cap \overline{\bs{D}}_q
]
\right|\no
= &(\gamma_{\rm e}+e_{\rm e})\left(
|X_{\rm e}| -\left|
X_{\rm e} \cap \overline{\bs{D}}_q
\right|
\right)
+e_{\rm e}\left(
|\bs{D}_c|-| X_{\rm e}|-\left|
\bs{D}_c\cap \overline{\bs{D}}_q
\right|
+
\left|
\overline{\bs{D}}_q\cap X_{\rm e}
\right|
\right)\no
=& \gamma_{\rm e}\left(
|X_{\rm e}| -\left|
X_{\rm e} \cap \overline{\bs{D}}_q
\right|
\right)
+e_{\rm e}\left(
|\bs{D}_c|-\left|
\bs{D}_c\cap \overline{\bs{D}}_q
\right|
\right).
\end{align}
Hence, one obtains that 
\begin{align}
&E_{\rm eff}(\bn_{\bs{D}_c,  \rm e })+e_{\rm e}|\overline{\bs{D}}_q|+(\gamma_Q-1)R_0^{-d} |\overline{\bs{D}}_q|\no
=& E_{\rm e}(\bn_{\bs{D}_c,  \rm e }, X_{\rm p})
+\omega_0 |X_{\rm p}|
+e_{\rm e}|\overline{\bs{D}}_q|+(\gamma_Q-1)R_0^{-d} |\overline{\bs{D}}_q|\no
\ge &
 \gamma_{\rm e}\left(
|X_{\rm e}| -\left|
X_{\rm e} \cap \overline{\bs{D}}_q
\right|
\right)
+e_{\rm e}\left(
|\bs{D}_c|-\left|
\bs{D}_c\cap \overline{\bs{D}}_q
\right|
\right)+\omega_0|X_{\rm p}|\no
&+e_{\rm e} \left(|\bs{D}_q|+\left|
\overline{\bs{D}}_q \setminus \bs{D}_q
\right|\right)+(\gamma_Q-1)R_0^{-d} |\overline{\bs{D}}_q|\no
\ge &
e_{\rm e}|\bs{D}|+\min\{\gamma_{\rm e}, \omega_0\} |\bs{D}_c|
-\gamma_{\rm e}\left|
X_{\rm e} \cap \overline{\bs{D}}_q
\right|
+e_{\rm e} \left(
\left|
\overline{\bs{D}}_q \setminus \bs{D}_q
\right|
-\left|
\overline{\bs{D}}_q \cap \bs{D}_c
\right|
\right)+(\gamma_Q-1)R_0^{-d} \left|
\overline{\bs{D}}_q
\right|
\no
\ge &e_{\rm e}|\bs{D}|+\min\{\gamma_{\rm e}, \omega_0\} |\bs{D}_c|
-\gamma_{\rm e} \left|
\overline{\bs{D}}_q
\right|
-|e_{\rm e}| \left|
\overline{\bs{D}}_q
\right|
+(\gamma_Q-1)R_0^{-d} \left|
\overline{\bs{D}}_q
\right|\no
\ge &
e_{\rm e}|\bs{D}|+\min\{\gamma_{\rm e}, \omega_0\} |\bs{D}_c|
+\left\{
(\gamma_Q-1) R_0^{-d}-\gamma_{\rm e}-|e_{\rm e}|
\right\}\left|
\overline{\bs{D}}_q
\right|\no
\ge &
(e_{\rm e}+\tilde{\gamma})|\bs{D}|.
\end{align}
\end{proof}



\begin{thebibliography}{10}

\bibitem{Arai2016}
A.~Arai.
\newblock {\em {Analysis on Fock Spaces and Mathematical Theory of Quantum
  Fields}}.
\newblock {WORLD} {SCIENTIFIC}, Dec. 2016.
\newblock \href {https://doi.org/10.1142/10367} {\path{doi:10.1142/10367}}.

\bibitem{Arai2022}
A.~Arai.
\newblock {\em {Infinite-Dimensional Dirac Operators and Supersymmetric Quantum
  Fields}}.
\newblock Springer Nature Singapore, 2022.
\newblock \href {https://doi.org/10.1007/978-981-19-5678-2}
  {\path{doi:10.1007/978-981-19-5678-2}}.

\bibitem{Borgs_1996}
C.~Borgs, J.~Jedrzejewski, and R.~Koteck{\'{y}}.
\newblock {The staggered charge-order phase of the extended Hubbard model in
  the atomic limit}.
\newblock {\em Journal of Physics A: Mathematical and General}, 29(4):733, feb
  1996.
\newblock  \href
  {https://doi.org/10.1088/0305-4470/29/4/005}
  {\path{doi:10.1088/0305-4470/29/4/005}}.

\bibitem{Borgs2000}
C.~Borgs and R.~Koteck{\'{y}}.
\newblock {Low Temperature Phase Diagrams of Fermionic Lattice Systems}.
\newblock {\em Communications in Mathematical Physics}, 208(3):575--604, Jan.
  2000.
\newblock \href {https://doi.org/10.1007/s002200050002}
  {\path{doi:10.1007/s002200050002}}.

\bibitem{Borgs1996}
C.~Borgs, R.~Koteck{\'{y}}, and D.~Ueltschi.
\newblock {Low temperature phase diagrams for quantum perturbations of
  classical spin systems}.
\newblock {\em Communications in Mathematical Physics}, 181(2):409--446, Nov.
  1996.
\newblock \href {https://doi.org/10.1007/bf02101010}
  {\path{doi:10.1007/bf02101010}}.

\bibitem{Bratteli1997}
O.~Bratteli and D.~W. Robinson.
\newblock {\em {Operator Algebras and Quantum Statistical Mechanics 2:
  Equilibrium States. Models in Quantum Statistical Mechanics}}.
\newblock Springer Berlin Heidelberg, 1997.
\newblock \href {https://doi.org/10.1007/978-3-662-03444-6}
  {\path{doi:10.1007/978-3-662-03444-6}}.

\bibitem{Conway1985}
J.~B. Conway.
\newblock {\em {A Course in Functional Analysis}}.
\newblock Springer New York, 1985.
\newblock \href {https://doi.org/10.1007/978-1-4757-3828-5}
  {\path{doi:10.1007/978-1-4757-3828-5}}.

\bibitem{Datta1996}
N.~Datta, R.~Fern{\'{a}}ndez, and J.~Fr\"{o}hlich.
\newblock {Low-temperature phase diagrams of quantum lattice systems. I.
  Stability for quantum perturbations of classical systems with finitely-many
  ground states}.
\newblock {\em Journal of Statistical Physics}, 84(3-4):455--534, Aug. 1996.
\newblock \href {https://doi.org/10.1007/bf02179651}
  {\path{doi:10.1007/bf02179651}}.

\bibitem{datta1996low}
N.~Datta, J.~Fr{\"o}hlich, L.~Rey-Bellet, and R.~Fern{\'a}ndez.
\newblock {Low-temperature phase diagrams of quantum lattice systems. II.
  Convergent perturbation expansions and stability in systems with infinite
  degeneracy}.
\newblock {\em Helvetica Physica Acta}, 69(5):752--820, 1996.

\bibitem{Datta2000}
N.~Datta, A.~Messager, and B.~Nachtergaele.
\newblock {Rigidity of Interfaces in the Falicov-Kimball Model}.
\newblock {\em Journal of Statistical Physics}, 99(1/2):461--555, 2000.
\newblock \href {https://doi.org/10.1023/a:1018609126399}
  {\path{doi:10.1023/a:1018609126399}}.

\bibitem{Freericks1995}
J.~K. Freericks and E.~H. Lieb.
\newblock {Ground state of a general electron-phonon Hamiltonian is a spin
  singlet}.
\newblock {\em Physical Review B}, 51(5):2812--2821, Feb. 1995.
\newblock \href {https://doi.org/10.1103/physrevb.51.2812}
  {\path{doi:10.1103/physrevb.51.2812}}.

\bibitem{Friedli2017}
S.~Friedli and Y.~Velenik.
\newblock {\em {Statistical Mechanics of Lattice Systems}}.
\newblock Cambridge University Press, Nov. 2017.
\newblock \href {https://doi.org/10.1017/9781316882603}
  {\path{doi:10.1017/9781316882603}}.

\bibitem{Koteck1999}
R.~Koteck{\'{y}} and D.~Ueltschi.
\newblock {Effective Interactions Due to Quantum Fluctuations}.
\newblock {\em Communications in Mathematical Physics}, 206(2):289--335, Oct.
  1999.
\newblock \href {https://doi.org/10.1007/s002200050707}
  {\path{doi:10.1007/s002200050707}}.

\bibitem{Lang1963}
I.~G. Lang and Y.~A. Firsov.
\newblock {Kinetic Theory of Semiconductors with Low Mobility}.
\newblock {\em Journal of Experimental and Theoretical Physics},
  16(5):1301--1312, May 1963.

\bibitem{Miyao2015}
T.~Miyao.
\newblock {Upper Bounds on the Charge Susceptibility of Many-Electron Systems
  Coupled to the Quantized Radiation Field}.
\newblock {\em Letters in Mathematical Physics}, 105(8):1119--1133, June 2015.
\newblock \href {https://doi.org/10.1007/s11005-015-0775-9}
  {\path{doi:10.1007/s11005-015-0775-9}}.

\bibitem{Miyao2016}
T.~Miyao.
\newblock {Rigorous Results Concerning the Holstein{\textendash}Hubbard Model}.
\newblock {\em Annales Henri Poincar{\'{e}}}, 18(1):193--232, June 2016.
\newblock \href {https://doi.org/10.1007/s00023-016-0506-5}
  {\path{doi:10.1007/s00023-016-0506-5}}.

\bibitem{Miyao2018}
T.~Miyao.
\newblock {Ground State Properties of the Holstein{\textendash}Hubbard Model}.
\newblock {\em Annales Henri Poincar{\'{e}}}, 19(8):2543--2555, May 2018.
\newblock \href {https://doi.org/10.1007/s00023-018-0690-6}
  {\path{doi:10.1007/s00023-018-0690-6}}.

\bibitem{Miyao2019}
T.~Miyao.
\newblock {Stability of Ferromagnetism in Many-Electron Systems}.
\newblock {\em Journal of Statistical Physics}, 176(5):1211--1271, 2019.
\newblock \href {https://doi.org/10.1007/s10955-019-02341-0}
  {\path{doi:10.1007/s10955-019-02341-0}}.

\bibitem{Pirogov1975}
S.~A. Pirogov and Y.~G. Sinai.
\newblock {Phase diagrams of classical lattice systems}.
\newblock {\em Theoretical and Mathematical Physics}, 25(3):1185--1192, Dec.
  1975.
\newblock \href {https://doi.org/10.1007/bf01040127}
  {\path{doi:10.1007/bf01040127}}.

\bibitem{Schmdgen2012}
K.~Schm\"{u}dgen.
\newblock {\em {Unbounded Self-adjoint Operators on Hilbert Space}}.
\newblock Springer Netherlands, 2012.
\newblock  \href
  {https://doi.org/10.1007/978-94-007-4753-1}
  {\path{doi:10.1007/978-94-007-4753-1}}.

\bibitem{simon2005}
B.~Simon.
\newblock {\em {Trace ideals and their applications: Second Edition
  (Mathematical Surveys and Monographs)}}.
\newblock Number 120. American Mathematical Society, 2005.

\bibitem{Reed1975}
B.~Simon and M.~Reed.
\newblock {\em {Methods of Modern Mathematical Physics, Vol II: Fourier
  Analysis, Self-Adjointness}}.
\newblock Academic Press, 1975.

\bibitem{sinai1982theory}
Y.~G. Sinai.
\newblock {\em {Theory of phase transitions: rigorous results Pergamon Press}}.
\newblock Oxford, 1982.

\bibitem{Tasaki2020}
H.~Tasaki.
\newblock {\em {Physics and Mathematics of Quantum Many-Body Systems}}.
\newblock Springer International Publishing, 2020.
\newblock \href {https://doi.org/10.1007/978-3-030-41265-4}
  {\path{doi:10.1007/978-3-030-41265-4}}.

\bibitem{Ueltschi2002}
D.~Ueltschi.
\newblock {\em {Geometric and Probabilistic Aspects of Boson Lattice Models}},
  pages 363--391.
\newblock Birkh{\"a}user Boston, Boston, MA, 2002.
\newblock \href {https://doi.org/10.1007/978-1-4612-0063-5_17}
  {\path{doi:10.1007/978-1-4612-0063-5_17}}.

\end{thebibliography}


\end{document}